\documentclass[11pt]{article}

\newcommand{\ifabs}[2]{#2}

\usepackage{amsmath}
\usepackage{amsfonts}
\usepackage{amssymb}
\usepackage{tikz}
\usepackage{amsthm}
\usepackage{fullpage}
\usepackage{caption}
\usepackage[lined,boxed,linesnumbered,ruled,vlined]{algorithm2e}
\usepackage{bm}
\usepackage[numbers]{natbib}
\usepackage[colorlinks=true,citecolor=blue, linkcolor=blue, urlcolor=blue]{hyperref}
\usepackage{times}
\usepackage{float}
\usepackage{cleveref}

\usepackage{comment}

\excludecomment{suppress}

\usepackage{tabularx}
\usepackage{array}
\newcolumntype{L}[1]{>{\raggedright\arraybackslash}p{#1}}
\newcolumntype{C}[1]{>{\centering\arraybackslash}m{#1}}
\newcolumntype{R}[1]{>{\raggedleft\arraybackslash}p{#1}}

\usepackage{makecell}
\usepackage[T1]{fontenc}
\usepackage{footnote}
\makesavenoteenv{tabular}

\newtheorem{theorem}{Theorem}[section]
\newtheorem{observation}[theorem]{Observation}
\newtheorem{claim}[theorem]{Claim}
\newtheorem{condition}[theorem]{Condition}
\newtheorem{example}[theorem]{Example}

\newtheorem{lemma}[theorem]{Lemma}
\newtheorem{proposition}[theorem]{Proposition}
\newtheorem{corollary}[theorem]{Corollary}
\theoremstyle{definition}

\newtheorem{definition}{Definition}[section]
\newtheorem{remark}{Remark}[section]
\newtheorem*{remark*}{Remark}
\newtheorem{assumption}[theorem]{Assumption}

\newcommand{\concept}[1]{\emph{#1}}
\newcommand{\todo}[1]{\typeout{TODO: \the\inputlineno: #1}\textbf{{\color{red}[[[ #1 ]]]}}}

\newcommand{\tO}{\widetilde{O}}

\newcommand{\gr}{\emph{good} }
\newcommand{\co}[1]{\emph{#1-correlated} }

\newcommand{\LOCAL}{$\mathsf{LOCAL}$}
\newcommand{\SLOCAL}{$\mathsf{SLOCAL}$}
\newcommand{\SLV}{$\mathsf{SLOCAL}\text{-}\mathsf{LV}$}

\newcommand{\dist}{\mathrm{dist}}
\newcommand{\id}{\mathrm{id}}
\newcommand{\vbl}{\mathsf{vbl}}
\newcommand{\Alg}[1]{\mathcal{{#1}}}

\renewcommand{\vec}[1]{\boldsymbol{{#1}}}
\newcommand{\RV}[1]{\boldsymbol{\mathcal{{#1}}}}

\newcommand{\RecursiveSampling}{\emph{RecursiveSampling}}

\newcommand{\Class}[1]{\mathfrak{{#1}}}
\newcommand{\Inst}[1]{\mathcal{{#1}}}

\newcommand{\TODO}[1]{\typeout{TODO: \the\inputlineno: #1}\textbf{{\color{red}[[[ #1 ]]]}}}

\newcommand{\aug}[3]{A^{{#1}}_{\boldsymbol{\lambda}({#3},{#2})}}

\newcommand{\ring}[3]{R^{{#1}}_{{#2}}({#3})}
\newcommand{\ri}[1]{R_{{#1}}}

\newcommand{\partialsat}[6]{
    P_{{#2},{#4}}^{#6}\left({#3}, {#5}\right)
}
\newcommand{\psat}[4]{
    P\left({#2}, {#4}\right)
}

\newcommand{\eventaround}[4]{
    V^{{#1}}_{[{#2},{#3}]}({#4})
}
\newcommand{\earound}[2]{
    V_{[{#1},{#2}]}
}

\newcommand{\eventin}[4]{
    V^{{#1}}_{({#2},{#3})}({#4})
}
\newcommand{\ein}[2]{
    V_{({#1},{#2})}
}

\newcommand{\ellf}[0]{
    g
}

\usepackage{titlesec}

\title{Perfect Simulation of Las Vegas Algorithms via Local Computation}

\ifabs{}{
 \author{
 Xinyu Fu\thanks{State Key Laboratory for Novel Software Technology, New Cornerstone Science Laboratory, Nanjing University. Emails: {xyfu@smail.nju.edu.cn}, {yinyt@nju.edu.cn}.} 
 \and Yonggang Jiang\thanks{MPI-INF and Saarland University, Germany, {yjiang@mpi-inf.mpg.de}}
 \and Yitong Yin\footnotemark[1]}
}
\date{}

\begin{document}

\maketitle

\begin{abstract}
The notion of Las Vegas algorithms was introduced by Babai (1979) and can be defined in two  ways:
\begin{itemize}
\item
In Babai's original definition, a randomized algorithm is called Las Vegas if it has a finitely bounded running time and certifiable random failure.
\item
Another definition widely accepted today is that Las Vegas algorithms refer to zero-error randomized algorithms with random running times.
\end{itemize}
The equivalence between the two definitions is straightforward. 
Specifically, for randomized algorithms with certifiable failures, repeatedly running the algorithm until no failure is encountered allows for faithful simulation of the correct output when it executes successfully.

We show that a similar perfect simulation can also be achieved in distributed local computation.
Specifically, in the \LOCAL{} model, with polylogarithmic overhead in time complexity, any Las Vegas algorithm with finitely bounded running time and locally certifiable failures can be converted to a zero-error Las Vegas algorithm. 
This transformed algorithm faithfully reproduces the correct output of the original algorithm in successful executions.
\end{abstract}


\pagenumbering{roman}
\newpage

\tableofcontents
\pagebreak

\pagenumbering{arabic}

\section{Introduction}
The Las Vegas algorithm, introduced by Babai~\cite{babai1979monte}, stands as a cornerstone in the theory of computing,
defining the \textbf{ZPP} class for decision problems efficiently solvable by Las Vegas algorithms. 
Beyond decision problems, Las Vegas algorithms are pivotal in tackling optimization~\cite{kalai1992subexponential,clarkson1995vegas}, searching~\cite{mitchell1992new,moser2010constructive}, or sampling~\cite{PW96,GJL19} problems. 


Las Vegas algorithms can indeed be defined in two distinct yet related ways.
In Babai's original work \cite{babai1979monte}, Las Vegas algorithms are defined as randomized algorithms whose failures are \emph{certifiable}:
\begin{itemize}
    \item A Las Vegas algorithm produces the correct output or reports failure within a finite bounded time.
\end{itemize}
An alternative definition of the Las Vegas algorithm, widely accepted today and utilized in various contexts, e.g.~in \cite{luby1993optimal, motwani1995randomized, mitzenmacher2005probability}, 
defines Las Vegas algorithms as \emph{zero-error} randomized algorithms:
\begin{itemize}
    \item A Las Vegas algorithm may exhibit random running time but always produces the correct output.
\end{itemize}

The equivalence between these two definitions is evident. 
Through truncation, we can transform a zero-error Las Vegas algorithm with random running time into a Las Vegas algorithm with bounded running time and certifiable failure. 
Conversely, restarting the algorithm upon failure allows us to convert a Las Vegas algorithm with certifiable failure into a zero-error Las Vegas algorithm. 
This strategy, retrying with independent random choice until success, defines a \emph{rejection sampling} procedure that faithfully simulates the random output of the Las Vegas algorithms, provided they return successfully without failure.

This strategy for perfectly simulating the Las Vegas algorithm relies on a global coordination machinery:
every part of the algorithm must be notified if a failure occurs anywhere. 
However, our focus is on exploring how this could be accomplished through local computations, without the need for global coordination.

\paragraph{The \LOCAL{} model.}
Local computations are formally characterized by the \LOCAL{} model~\cite{linial1992locality, peleg2000distributed}.
An {instance} consists of a network $G=(V,E)$, which is an undirected graph, 
and a vector $\vec{x}=(x_v)_{v\in V}$  of local inputs. 
Each node $v\in V$ receives $x_v$ and $n=|V|$ as its input and can access to private random bits. 
Communications are synchronized and proceed in rounds. 
In each round, each node may perform arbitrary local computation based on all information collected so far
and send messages of arbitrary sizes to all its neighbors. 
%
This gives a \LOCAL{} algorithm. The time complexity is measured by the number of rounds spent by the algorithm until all nodes terminate.
A \LOCAL{} algorithm is said to be a $t(n)$-round \LOCAL{} algorithm on a class of instances,
if it always terminates within $t(n)$ rounds on every instance from that class, where $n$ represents the number of nodes of the instance.

The (Babai's) Las Vegas algorithm can be defined in the \LOCAL{} mode with \emph{locally certifiable failures}.
A $t(n)$-round \LOCAL{} algorithm is called \emph{Las Vegas} 
if each node $v\in V$ returns a pair $(Y_v,F_v)$, where $Y_v$ stands for the local output at $v$, 
and $F_v\in\{0,1\}$ indicates whether the algorithm failed locally at $v$.
The algorithm successfully returns if none of the nodes fails.
Furthermore, it is guaranteed that $\sum_{v\in V}\mathbb{E}[F_v]<1$.
This notion of Las Vegas \LOCAL{} algorithm was formulated in~\cite{ghaffari2018derandomizing}.

In this paper, we try to answer the following fundamental question:
\begin{quote}
    \emph{Can we faithfully simulate the correct output avoiding all local failures via local computation?}
\end{quote}
Specifically, 
we wonder whether a fixed-round Las Vegas \LOCAL{} algorithm with locally certifiable failures, 
can be converted into a zero-error Las Vegas \LOCAL{} algorithm that produces the correct output $(Y_v)_{v\in V}$ conditioned on $\sum_{v\in V}F_v=0$,
where the distribution of the correct output is faithfully preserved. 
%
This asks for what used to be achieved by the rejection sampling in the absence of global coordination.

In this paper, for the first time, we answer this question affirmatively.
We prove the following result for the perfect simulation of Las Vegas algorithms via local computation.

\begin{theorem}[main theorem, informal]
\label{thm:main-informal}
Any $t(n)$-round Las Vegas \LOCAL{} algorithm can be converted to a zero-error Las Vegas \LOCAL{} algorithm,
which terminates within $t(n)\cdot\mathrm{polylog}(n)$ rounds with probability $1-n^{-O(1)}$,
and returns the output of the $t(n)$-round Las Vegas \LOCAL{} algorithm conditioned on no failure. 
\end{theorem}

In the above theorem, the output of the zero-error Las Vegas \LOCAL{} algorithm is identically distributed as the output of the $t(n)$-round Las Vegas \LOCAL{} algorithm conditioned on that none of the nodes fails, 
i.e.~the zero-error algorithm perfectly simulates a successful running of the algorithm that may locally fail.

%
%

To see how nontrivial this is, 
consider a weakened task:
to generate an assignment of random bits so that under this random choice the algorithm terminates without failure.
One may think of this as ``solving'' the random bits under which the algorithm successfully returns, 
which is weaker than our goal, 
where the generated random bits are further required to follow the correct distribution.
However, for local computation, just solving the feasible random bits without bothering their distribution is already highly nontrivial.

In a seminal work~\cite{ghaffari2018derandomizing}, Ghaffari, Harris and Kuhn 
gave a systematic approach for solving the good random bits under which a Las Vegas \LOCAL{} algorithm successfully returns.
Their derandomization based approach preserves the \emph{support} of distribution, hence was more suitable for the \emph{distributed graph problems} for constructing feasible solutions on graphs.
Specifically, $\mathrm{polylog}(n)$-round reductions were established between the two types of the Las Vegas \LOCAL{} algorithms for such problems.
%
%
%
%

In contrast,
the perfect simulation is guaranteed in
\Cref{thm:main-informal} preserves the distribution,  
therefore, the result can apply to the problems beyond constructing feasible solutions, for example, the sampling problems.

Consider the \emph{Gibbs distributions} defined on the network $G=(V, E)$. 
Each node $v$ corresponds to a variable with finite domain $\Sigma$.
Let $\mathcal{F}$ be a class of \emph{constraints}, where each $f\in\mathcal{F}$ is a nonnegative-valued function $f:\Sigma^{\vbl(f)}\to\mathbb{R}_{\ge 0}$ defined on the variables in $\vbl(f)\subseteq V$.
This defines a Gibbs distribution $\mu$ over all assignments $\sigma\in\Sigma^V$ by:
\[
\mu(\sigma) \propto \prod_{f\in\mathcal{F}}f\left(\sigma_{\vbl(f)}\right).
\]
Such Gibbs distribution $\mu$ is said to be \emph{local}, if: (1) for any $f\in\mathcal{F}$, the diameter of $\vbl(f)$ in graph $G$ is bounded by a constant; 
and (2) for any partial assignment $\sigma\in\Sigma^\Lambda$ specified on $\Lambda\subseteq V$, if $\sigma$ is locally feasible, 
i.e.~if $f(\sigma_{\vbl(f)})>0$ for all $f\in\mathcal{F}$ with $\vbl(f)\subseteq \Lambda$, 
then $\sigma$ is (globally) feasible, which means that $\sigma$ can be extended to a feasible full assignment $\tau\in\Sigma^V$ such that $\tau_\Lambda=\sigma$ and $\mu(\tau)>0$.\footnote{In~\cite{feng2018local}, this property of local feasibility implying global feasibility in Gibbs distribution was called ``locally admissible''.}

A Gibbs distribution $\mu$ is said to have \emph{strong spatial mixing with exponential decay}
if 
the discrepancy (measured in total variation distance) 
between the marginal distributions $\mu_v^\sigma,\mu_v^\tau$ at any $v\in V$ given the respective feasible boundary conditions $\sigma,\tau\in\Sigma^{\Lambda}$ on $\Lambda\subseteq V$ that differ over an arbitrary $\Delta\subseteq\Lambda$,
 satisfies: 
\[
d_{\mathrm{TV}}\left({\mu_v^{\sigma}},{\mu_v^{\tau}}\right)\le |V|^{O(1)}\cdot \exp(\Omega(\dist_G(v,\Delta))).
\]
%
The strong spatial mixing is a key property for sampling algorithms.
Its implication for efficient sampling from general Gibbs distributions is a major open problem.
In~\cite[Corollary 5.3]{feng2018local}, 
a $\mathrm{polylog}(n)$-round Las Vegas \LOCAL{} algorithm with bounded local failures was given, for perfect sampling from local Gibbs distributions that have strong spatial mixing with exponential decay.
By \Cref{thm:main-informal}, this immediately implies the following result for perfect simulation of Gibbs distributions via local computation.

\begin{corollary}
\label{cor:sampling-informal}
For any class of local Gibbs distributions that have strong spatial mixing with exponential decay,
there is a \LOCAL{} algorithm for perfect sampling from the Gibbs distribution, which terminates within $\mathrm{polylog}(n)$ rounds with probability $1-n^{-O(1)}$.
\end{corollary}

\begin{remark}
The computational tractability of sampling from general Gibbs distributions with strong spatial mixing remains unresolved, even in the sequential and centralized setting. 
\Cref{cor:sampling-informal} certainly does not resolve the computational complexity of sampling.
This is because the \LOCAL{} model does not impose any bounds on the computational costs, 
and generic reductions such as the one stated in \Cref{thm:main-informal} must apply to all \LOCAL{} algorithms, 
including those that run in super-polynomial time.
Such lack of limitation on computational power was indeed common in general (black-box) reductions in the \LOCAL{} model and may be necessary.
For instance, the derandomization of general \LOCAL{} algorithms in~\cite{ghaffari2018derandomizing} also relied on local exhaustive search of random bits. 
Here, our objective is even more ambitious: to preserve the distribution of the random bits of a sucessive running.

On a positive note, if it is known that the correct output has strong spatial mixing, as in the case of \Cref{cor:sampling-informal}, then the perfect simulation in \Cref{thm:main-informal} can be simplified (using \Cref{alg:recursive-sampling-with-cd} instead of \Cref{alg:recursive-sampling-without-cd}), avoiding exhaustive enumerations.  
Consequently, the computational cost of the algorithm in \Cref{cor:sampling-informal} is dominated by approximately sampling from the marginal distributions $\mu_v^\sigma$ and estimating the marginal probabilities within an inverse-polynomial accuracy.
Indeed, the non-trivial computation costs in both the sampling algorithm in~\cite{feng2018local} and our perfect simulation algorithm  when there is a strong spatial mixing, are dominated by these tasks which are computationally equivalent to sampling. 

In summary,  the focus of \Cref{cor:sampling-informal} is the locality of Gibbs sampling.
Our proof of it may also imply:
locality of computation imposes no additional significant barrier to the task of Gibbs sampling with strong spatial mixing, provided the sampling problem is computationally tractable.
\end{remark}


%

\subsection{A sampling Lov\'{a}sz local lemma in the \LOCAL{} model}\label{sec:samplingsatisfyingsolution}

The perfect simulation of Las Vegas \LOCAL{} algorithms stated in \Cref{thm:main-informal} is achieved by resolving a more general problem,
namely, generating a random sample avoiding all bad events.
This problem is formulated as a natural sampling problem in the variable-framework of the Lov\'{a}sz local lemma.


An instance for the variable-framework Lov\'{a}sz local lemma (LLL) is given by $I=\left(\{X_i\}_{i\in U},\{A_v\}_{v\in V}\right)$,
where $\{X_i\}_{i\in U}$ is a set of mutually independent random variables, 
such that each $X_i$ follows a distribution $\nu_i$ over a finite domain $\Sigma_i$;
and $\{A_v\}_{v\in V}$ is a set of \emph{bad events}, such that for each $v\in V$, the occurrence of $A_v$ is determined by the evaluation of $X_{\vbl(v)}=(X_i)_{i\in \vbl(v)}$, where $\vbl(v)\subseteq U$ denotes the subset of variables on which $A_v$ is defined. 
%
%
The LLL instance $I$ defines a \emph{dependency graph} $D=D_I=(V,E)$, such that each vertex $v\in V$ represents a bad event $A_v$ and each $\{u,v\}\in E$ iff $\vbl(v)\cap \vbl(u)\neq\emptyset$.

An LLL instance $I$ is said to be \concept{$\gamma$-satisfiable}, if the probability avoiding all bad events is bounded as:
\begin{align}
\Pr
\left(\,\bigcap_{v\in V}\overline{A_v}\,\right)\ge \gamma.\label{eq:admissibility}
\end{align}
The Lov\'{a}sz local lemma~\cite{erdos1975problems} states a sufficient condition on the dependency graph for $\gamma$ to be positive.

A satisfiable LLL instance $I$ gives rise to a natural probability distribution over satisfying assignments.
Let $\mu=\mu_I$ denote the distribution of the random vector $\boldsymbol{X}=(X_i)_{i\in U}$ conditioning on that none of the bad events $\{A_v\}_{v\in V}$ occur.
Formally, denote by $\Omega=\Omega_I\triangleq\{\sigma\in\bigotimes_{i\in U}\Sigma_i\mid \sigma\text{ avoids }A_v\text{ for all }v\in V\}$   the space of all satisfying assignments, 
and $\nu=\nu_I\triangleq\prod_{i\in U}\nu_i$ the product measure. Then
\begin{align}\label{eq:def-partition-function}
\forall \sigma\in\Omega,\quad
\mu({\sigma})\triangleq \Pr(\boldsymbol{X}=\sigma\mid \boldsymbol{X}\in\Omega)=\frac{\nu({\sigma})}{\nu(\Omega)}.
\end{align}

This distribution $\mu=\mu_I$ of random satisfying assignment was referred to as the \emph{LLL distribution} in~\cite{haeupler2011new,harris2020new}. It is a {Gibbs distribution} defined by hard constraints.

The following defines a computational problem for generating a sample according to the distribution $\mu$.

  \par\addvspace{.5\baselineskip}
\framebox{
  \noindent
  \begin{tabularx}{14cm}{@{\hspace{\parindent}} l X c}
    \multicolumn{2}{@{\hspace{\parindent}}l}{\underline{Sampling Satisfying Solution of Lov\'{a}sz Local Lemma}} \\
    \textbf{Input :} & a $\gamma$-satisfiable LLL instance $I$ with dependency graph $D_I=(V,E)$;\\ 
    \textbf{Output:} & a random satisfying assignment $\vec{X}^*=(X^*_i)_{i\in U}$ distributed as $\mu_I$.
  \end{tabularx}
 }
\par\addvspace{.5\baselineskip}

When the problem is solved in the \LOCAL{} model, the input is presented to the algorithm as follows: 
\begin{enumerate}
    \item The network $G$ of the \LOCAL{} model is just the dependency graph $D_I$.
    \item Each node $v\in V$ receives as input the values of $n=|V|$ and $\gamma$, along with the definition of the bad event $A_v$, and the distributions $\nu_i$ of the random variables $\{X_i\mid i\in\vbl(v)\}$ on which $A_v$ is defined
    so that it can locally draw independent evaluations of the random variables $\{X_i\mid i\in\vbl(v)\}$ or check the occurrence of $A_v$ on such evaluation.
\end{enumerate}

Our main technical result is an efficient \LOCAL{} algorithm for sampling satisfying solution according to the LLL distribution, for any LLL instance that is not prohibitively scarce to satisfy. 
We call this result a ``\LOCAL{} sampling lemma'' since it uses the local lemma framework to give a \LOCAL{} sampling algorithm.

\begin{theorem}[\LOCAL{} sampling lemma]\label{thm: sample-gibbs}
There is a randomized \LOCAL{} algorithm such that for any LLL instance $I$ with $n$ bad events, 
if $I$ is $\gamma$-satisfiable, then the algorithm returns a random satisfying assignment drawn from $\mu_I$, 
within $\tO\left(\log^6 n\cdot \log^4\frac{1}{\gamma}\right)$ rounds in expectation,
and within $\tO\left(\log^6 n\cdot \log^4\frac{1}{\gamma}\cdot\log^6\frac{1}{\epsilon}\right)$ rounds with probability at least $1-\epsilon$ for any $0<\epsilon<1$.
%
\end{theorem}

While the classical algorithmic Lov\'{a}sz Local Lemma (LLL) aims to construct a satisfying assignment, 
our objective is more ambitious: 
to generate a random satisfying assignment according to the LLL distribution $\mu$. 
This sampling problem has been shown to require polynomially stronger conditions than the conventional algorithmic LLL~\cite{bezakova2019approximation, GJL19, harris2020new}. 
Recently, polynomial-time centralized algorithms have been discovered for sampling uniform satisfying assignments for general LLL instances~\cite{he2022sampling,jain2022towards,feng2021fast,moitra2019approximate}. 
However, in our context of perfectly simulating Las Vegas \LOCAL{} algorithms, 
we do not impose any local lemma type of condition on the dependency graph. 
Instead, we show that the sampling problem is tractable via local computation, 
as long as the chance for a random assignment to be satisfying is non-negligible. 
To the best of our knowledge, this marks the first result of its kind.

As a byproduct, the algorithm in \Cref{thm: sample-gibbs} gives a uniform sampler for LCLs in the \LOCAL{} model.
The \concept{locally checkable labelings} (\concept{LCLs}) (see~\cite{naor1995can} for formal definitions) are vertex labelings satisfying local constraints in the network.
Suppose that each label is drawn uniformly and independently at random, 
and violating the locally checkable constraint at each node $v\in V$ defines a bad event $A_v$.
Then the LLL distribution corresponds to the uniform distribution over LCLs, and \Cref{thm: sample-gibbs} has the following corollary.

%

\begin{corollary}
For any locally checkable labeling (LCL) problem, if uniform random labeling can always generate a correct LCL with probability at least $2^{-\mathrm{polylog(n)}}$, then there is a randomized \LOCAL{} algorithm that outputs a uniform random LCL within $\mathrm{polylog}(n)$ rounds with high probability.
\end{corollary}


\subsection{Perfect simulation via \LOCAL{} sampling lemma}
Recall that for \LOCAL{} algorithms, an instance $\Inst{I}=(G,\vec{x})$ consists of a network $G=(V,E)$ and a vector $\vec{x}=(x_v)_{v\in V}$  specifying the local input $x_v$ to each node $v\in V$.
Let $\Class{C}$ be a class of instances.
A \emph{$t(n)$-round Las Vegas \LOCAL{} algorithm with success probability $\gamma(n)$} on instance class $\Class{C}$, 
is a randomized \LOCAL{} algorithm such that on every instance $\Inst{I}=(G,\vec{x})\in\Class{C}$, where $G$ is a network with $n=|V|$ nodes, 
at every node $v\in V$ the algorithm terminates within $t(n)$ rounds and outputs a pair $(Y_v,F_v)$ where $F_v\in\{0,1\}$ indicates whether the algorithm failed locally at $v$, and the probability that the algorithm succeeds is 
\[
\Pr\left(\forall v\in V: \vec{F}_{v}={0}\right)\ge\gamma(n).
\]
Denote by $(\vec{Y}_\Inst{I}, \vec{F}_\Inst{I})=\left\langle(Y_v)_{v\in V},(F_v)_{v\in V}\right\rangle$ the output of the Las Vegas \LOCAL{} algorithm on instance $\Inst{I}=(G,\vec{x})$ with network $G=(V,E)$.
The following theorem is a formal restatement of \Cref{thm:main-informal}, which gives a zero-error Las Vegas \LOCAL{} algorithm that perfectly simulates the good output $(\vec{Y}_\Inst{I}\mid \vec{F}_\Inst{I}=\vec{0})$ when there is no failure everywhere in the network.

\begin{theorem}[formal restatement of \Cref{thm:main-informal}]
\label{thm:main-formal}
Let $t:\mathbb{N}\to\mathbb{N}$ and $\gamma:\mathbb{N}\to[0,1]$ be two functions.
Let $\Alg{A}$ be a $t(n)$-round Las Vegas \LOCAL{} algorithm with success probability $\gamma(n)$ on instance class $\Class{C}$.
%
There is a \LOCAL{} algorithm $\Alg{B}$ such that on every instance $\Inst{I}\in\Class{C}$ of $n$ nodes, for any $0<\epsilon<1$,
\begin{itemize}
\item  
$\Alg{B}$ terminates within $t(n)\cdot\tO\left(\log^6 n \cdot \log^4 \frac{1}{\gamma(n)}\cdot \log^6 \frac{1}{\epsilon}\right)$ rounds with probability at least $1-\epsilon$;
\item 
upon termination, $\Alg{B}$ returns an output $\vec{Y}_{\Inst{I}}^{\Alg{B}}$ that is identically distributed as $(\vec{Y}_\Inst{I}^{\Alg{A}}\mid \vec{F}_\Inst{I}^{\Alg{A}}=\vec{0})$,
which stands for the output of $\Alg{A}$ on the same instance $\Inst{I}$ conditioned on that none of the nodes fails.
\end{itemize}
%
%
%
\end{theorem}

\begin{proof}
%
%
Let  $\Inst{I}=(G,\vec{x})\in\Class{C}$ be the instance of the \LOCAL{} algorithm, where $G=(V,E)$ is a network with $n=|V|$ nodes and the vector $\vec{x}=(x_v)_{v\in V}$ specifies the local inputs.
For each $v\in V$, let $X_v$ denote the local random bits at node $v$ used by algorithm $\Alg{A}$.

%
%

%
Since $\Alg{A}$ is a $t(n)$-round Las Vegas \LOCAL{} algorithm, 
at any $v\in V$, the algorithm $\Alg{A}$ deterministically maps the inputs $\vec{x}_{B}=(x_v)_{v\in B}$ and the random bits $\vec{X}_{B}=(X_v)_{v\in B}$ within the $t(n)$-ball $B=B_{t(n)}(v)$, to the local output $\left(Y_v^{\Alg{A}},F_v^{\Alg{A}}\right)$, where $F_v^{\Alg{A}}\in\{0,1\}$ indicates the failure at $v$. 
%
This defines a bad event $A_v$ for every $v\in V$, on the random variables $X_u$ for $u\in  B_{t(n)}(v)$, i.e.~$\vbl(v)= B_{t(n)}(v)$, by
\[
A_v: F_v^{\Alg{A}}(\vec{X}_{\vbl(v)})=1.
\]
Together, this defines an LLL instance $I= (\{X_v\}_{v\in V},\{A_v\}_{v\in V})$, which is $\gamma(n)$-satisfiable because the probability that $\Alg{A}$ has no failure everywhere  is at least $\gamma(n)$. 

%

We simulate the sampling algorithm in Theorem~\ref{thm: sample-gibbs} (which we call the \emph{LLL sampler}) on this LLL instance $I$.
Rather than executing it on the dependency graph $D_I$ as in Theorem~\ref{thm: sample-gibbs}, here we simulate the LLL sampler on the network $G=(V, E)$, where each node $v\in V$ holds an independent random variable $X_v$ and a bad event $A_v$. 
%
%
%
Note that any 1-round communication in the dependency graph $D_I$
can be simulated by $O(t(n))$-round communications in this network $G=(V,E)$.
Also note that at each $v\in V$, the values of $t(n)$ and $\gamma(n)$ can be computed locally by knowing $n=|V|$ and enumerating all network instances $\Inst{I}\in\Class{C}$ with $n$ nodes.
The LLL sampler can thus be simulated with $O(t(n))$-multiplicative overhead. 
In the end it outputs an $\vec{X}^*=(X_v^*)_{v\in V}\sim\mu_I$, which is identically distributed as $\left(\vec{X}\mid \vec{F}^{\Alg{A}}_{\Inst{I}}=\vec{0}\right)$, i.e.~the random bits used in the algorithm $\Alg{A}$ conditioned on no failure.
%
%
The final output $\vec{Y}^*=(Y_v^*)_{v\in V}$ is computed by simulating $\Alg{A}$ within $t(n)$ locality deterministically using $\vec{X}^*=(X_v^*)_{v\in V}$ as random bits.

The LLL sampler in Theorem~\ref{thm: sample-gibbs} is a Las Vegas algorithm with random terminations. 
Each node $v\in V$ can continue updating $Y_v^*$ using the current random bits $\vec{X}^*_B$ it has collected within its $t(n)$-local neighborhood $B=B_{t(n)}(v)$.   
Once the LLL sampler for generating $\vec{X}^*$ terminates at all nodes, 
the updating of $\vec{Y}^*$ will stabilize within additional $t(n)$ rounds.  
%
And this final $\vec{Y}^*$ is identically distributed as $\left(\vec{Y}_{\Inst{I}}^{\Alg{A}}\mid \vec{F}^{\Alg{A}}_{\Inst{I}}=\vec{0}\right)$.
This gives us the zero-error Las Vegas \LOCAL{} algorithm $\Alg{B}$ as claimed in Theorem~\ref{thm:main-formal}. 
\end{proof}

\section{Preliminary}\label{sec:preliminary}

\subsection{Graph and LLL notations}\label{sec:graph-notations}
Let $G = (V, E)$ be an undirected graph. The following notations are used throughout.
\begin{itemize}
    \item \emph{Neighborhoods}:
    $N_G(v)\triangleq\{u\in V\mid \{u,v\}\in E\}$ and inclusive neighborhood $N^+_G(v)\triangleq N(v)\cup\{v\}$.
    \item \emph{Distances}:
    $\dist_G(u, v)$ represents the shortest path distance between $u$ and $v$ in $G$, and 
    $\dist_G(S, T)\triangleq \min_{u\in S, v\in T} \dist_G(u,v)$.
    The diameter of $\Lambda\subseteq V$ in $G$ is given by $\mathrm{diam}_G(\Lambda)\triangleq \dist_G(\Lambda, \Lambda)$.
    \item \emph{Balls}:
    $B^G_{r}(v)\triangleq\{u\in V\mid \dist_G(u,v)\le r \}$ 
    and $B^G_{r}(\Lambda)\triangleq\{u\in V\mid \dist_G(u,\Lambda)\le r\}$ for $\Lambda\subseteq V$. 
    \item \emph{Shells/Spheres}:
    $S^G_{[\ell,r]}(v)\triangleq B^G_{r}(v) \setminus B^G_{\ell-1}(v)$ 
    and $S^G_{[\ell,r]}(\Lambda)\triangleq B^G_{r}(\Lambda) \setminus B^G_{\ell-1}(\Lambda)$ for $\Lambda\subseteq V$.
\end{itemize}
%
%
Let  $I=\left(\{X_i\}_{i\in U},\{A_v\}_{v\in V}\right)$ be a LLL instance with dependency graph $D_I$.
%
%
The following defines a notion of rings of variables that enclose a nonempty subset $\Lambda\subseteq V$ of bad events.
\begin{itemize}
    \item \emph{Rings}:
    $\ring{I}{r}{\Lambda}\triangleq\vbl\left(B_{r}^{D_I}(\Lambda)\right)\setminus \vbl\left(B_{r-1}^{D_I}(\Lambda)\right)$
    and in particular, $\ring{I}{0}{\Lambda}\triangleq\vbl(\Lambda)$. 
    Furthermore,
    \begin{align}
    \ring{I}{[i,j]}{\Lambda}\triangleq \bigcup_{i\leq r \leq  j} \ring{I}{r}{\Lambda}.\label{eq:def-variable-ring}
    \end{align}
\end{itemize}
%
We also define the rings of bad events {intersected} or {contained} by a ring of variables:
\begin{align}
\eventaround{I}{i}{j}{\Lambda}
\triangleq\left\{v\in V\mid \vbl(v)\cap \ring{I}{[i,j]}{\Lambda} \neq \emptyset \right\}
\quad\text{ and }\quad
\eventin{I}{i}{j}{\Lambda}
\triangleq\left\{v\in V\mid \vbl(v)\subseteq \ring{I}{[i,j]}{\Lambda} \right\}.\label{eq:def-event-ring}
\end{align}
In all the above notations, we omit the graph $G$ or the LLL instance $I$ if they are clear in the context.

Inspired by the induced subgraph, we define the sub-instance of $I=\left(\{X_i\}_{i\in U},\{A_v\}_{v\in V}\right)$ induced by a subset $\Lambda\subseteq V$ of bad events:
\begin{align*}
    \vbl(\Lambda)
    \triangleq\bigcup_{v\in \Lambda}\vbl(v)
    \quad\text{ and }\quad
    I(\Lambda)
    \triangleq\left(\{X_i\}_{i\in \vbl(\Lambda)},\{A_v\}_{v\in \Lambda}\right),
\end{align*}
where $I(\Lambda)$ is the LLL sub-instance induced by the bad events $\{A_v\}_{v\in \Lambda}$.

\subsection{Marginal distribution} 

Let $I=\left(\{X_i\}_{i\in U},\{A_v\}_{v\in V}\right)$ be a LLL instance, where each random variable $X_i$ follows the distribution $\nu_i$ over domain $\Sigma_i$. 
%
For $\Lambda\subseteq U$, define $\Sigma_{\Lambda} \triangleq \bigotimes_{i\in \Lambda} \Sigma_i$ and $\nu_{\Lambda} \triangleq \prod_{i\in \Lambda} \nu_i$,
and write  $\nu= \nu_U$ and $\Sigma= \Sigma_U$.

For nonempty $\Lambda \subset U$ and $\tau  \in \Sigma_\Lambda$, define 
\begin{align}\label{eq:def-Omega-tau}
\Omega^{\tau} =\Omega_I^{\tau} 
\triangleq 
\left\{\sigma \in \Sigma 
\mid 
\sigma_\Lambda=\tau
\text{ and }\sigma
\text{ avoids bad events }A_v\text{ for all }v\in V\text{ s.t. }\vbl(v)\not\subseteq\Lambda 
\right\}.    
\end{align}

For disjoint $S,T \subseteq U$,  for $\sigma \in \Sigma_{S}$ and $\tau \in \Sigma_{T}$, 
denote by $\sigma \land \tau$ the direct concatenation of $\sigma$ and $\tau$, that is, 
$\sigma \land \tau \in  \Sigma_{S \cup T}$ satisfying $(\sigma \land \tau)_i=\sigma(i)$ for $i\in S$ and $(\sigma \land \tau)_i=\tau(i)$ for $i\in T$.

The following defines a notion of marginal distribution in the LLL instance. 

\begin{definition}[marginal distribution]
For $\Lambda\subseteq  U$, a $\tau\in \Sigma_\Lambda$  is said to be a \emph{feasible boundary condition}  if $\Omega^\tau\neq\emptyset$, where $\Omega^\tau$ is defined in~\eqref{eq:def-Omega-tau}.
Given $\Lambda\subset  U$ and feasible boundary condition $\tau\in \Sigma_\Lambda$, 
for any nonempty $S\subseteq U\setminus{\Lambda}$,
the \emph{marginal distribution} on $S$ induced by $\mu=\mu_I$  conditioned on $\tau$, denoted by $\mu_S^\tau=\mu_{I,S}^\tau$, is defined as:
\[
\forall \sigma \in \Sigma_S, \quad \mu^{\tau}_{S}(\sigma) \triangleq \Pr_{\vec{X}\sim\nu}(X_S =\sigma \mid \vec{X}\in \Omega^{\tau}).
\]   
\end{definition}

\subsection{Decay of correlation in LLL}
\label{subsec:correlation-decay-LLL}





We consider the following notion of correlation decay in the LLL instance.
\begin{definition}[$\epsilon$-correlated sets]\label{def:correlation}
%
A pair of disjoint $S, T\subset U$ with $S\cup T\neq U$,
is said to be \concept{$\epsilon$-correlated}, if
one of $S,T$ is empty, or
for any $\sigma_1,\sigma_2\in\Sigma_{S}$ and $\tau_1,\tau_2\in\Sigma_{T}$,
\begin{align*}
\nu\left(\Omega^{\sigma_1\land\tau_1}\right) \cdot \nu\left(\Omega^{\sigma_2\land\tau_2}\right) 
\leq (1+\epsilon) \cdot \nu\left(\Omega^{\sigma_1\land\tau_2}\right)\cdot \nu\left(\Omega^{\sigma_2\land\tau_1}\right).
\end{align*}
\end{definition}

To see that this indeed defines a decay of correlation, 
note that it is equivalent to the following property:
For $\vec{X}$ drawn according to the product distribution $\nu$ 
that avoids all bad events $A_v$ satisfying $\vbl(v)\not\subseteq S\cup T$,
\begin{align*}
    &\Pr\left(X_S=\sigma_1\land X_T=\tau_1\right)\cdot
    \Pr\left(X_S=\sigma_2\land X_T=\tau_2\right)\\
    \le 
    &(1+\epsilon)\cdot
    \Pr\left(X_S=\sigma_1\land X_T=\tau_2\right)\cdot
    \Pr\left(X_S=\sigma_2\land X_T=\tau_1\right).
\end{align*}
Recall that we want to bound the correlation between $X_S$ and $X_T$ in a random $\vec{X}$ distributed as $\mu=\mu_I$.
%
Here, \Cref{def:correlation} is slightly different by ignoring the bad events $A_v$ defined on the variables within $S\cup T$.

\subsection{Las Vegas  \SLOCAL{} algorithm}\label{subsec:SLOCAL-LV}

Our main sampling algorithm is described in a sequential local (\SLOCAL) paradigm. 
The \SLOCAL{} model introduced by Ghaffari, Kuhn, and Maus \cite{ghaffari2016complexity} captures the local computations where symmetry breaking is not concerned.
We extend this notion to the algorithms with random locality of computation.

\paragraph{\SLV{} algorithms}
An $N$-scan \SLV{} algorithm runs on a network $G=(V,E)$ with a subset $A\subseteq V$ of \emph{active} nodes.
Each node $v\in V$ maintains a local memory state $M_v$, 
initially storing $v$'s local input, random bits, its ID, and neighbors' IDs.
An arbitrary total order is assumed over~$V$,
such that the relative order between any $u,v\in V$ can be deduced from the contents of $M_u$ and $M_v$.
	%
	%

The algorithm operates in $N\ge 1$ scans.
Within each scan, the active nodes in $A$ are processed one after another in the ordering.
Upon each node $v\in A$ being processed, for $\ell=0,1,2,\ldots$, 
the node $v$ tries to grow an $\ell$-ball $B_{\ell}(v)$ and update the memory states $M_u$ for all $u\in B_{\ell}(v)$ based on the information observed so far by $v$, 
until some stopping condition has been met by the information within the current ball $B_{\ell}(v)$.
Finally, each $v\in V$ outputs a value based on its memory state $M_v$.

	
	%

Compared to the standard (Monte Carlo) \SLOCAL{} algorithm, whose locality is upper bounded by a fixed value, 
in \SLV{} algorithm, the local neighborhoods are randomly constructed in a sequential and local fashion.
The next theorem gives a simulation of \SLV{} algorithms in the \LOCAL{} model.
\begin{proposition}[simulation of \SLV{} in \LOCAL{}]
\label{thm: SLV-local to local}
Let $\Alg{A}$ be an $N$-scan \SLV{} algorithm that assumes an arbitrary ordering of nodes. 
Then there is a randomized \LOCAL{} algorithm $\Alg{B}$, 
such that starting from the same initial memory states $\boldsymbol{M}=(M_v)_{v\in V}$, 
the algorithm $\Alg{B}$ terminates and returns the same output as $\Alg{A}$ within $O\left(|A|\cdot \max_{v\in A,j\in[N]}\ell_{v,j} \right)$ rounds, where $A$ is the set of active nodes and $\ell_{v,j}$ is the radius of the ball accessed by the active node $v$ in the $j$th scan of algorithm $\Alg{A}$, both fully determined by $\boldsymbol{M}$.
\end{proposition}

Compared to the simulation of \SLOCAL{} Monte Carlo algorithm in the \LOCAL{} model proved in~\cite{ghaffari2016complexity}, which relies on the network decomposition to parallelize the \SLOCAL{} procedure,
\Cref{thm: SLV-local to local} provides a rather straightforward simulation that does not parallelize the local computations.
An advantage of such an easy simulation is that it does not require a worst-case complexity upper bound for all scan orders of nodes. 
%
%
%
This translation from \SLV{} to \LOCAL{} algorithm makes it more convenient to describe \LOCAL{} algorithms where there are multiple randomly growing local neighborhoods interfering with each other.

A formal proof of \Cref{thm: SLV-local to local} is given in \Cref{sec: details of SLV-local to local} for completeness. 


\section{Exposition of the Algorithm: Special Case}\label{sec:special-case}

In this section, we expose the key ideas of the main sampling algorithm stated in \Cref{thm: sample-gibbs} within a special case. 
The detailed and formal construction of the algorithm will be presented in the next section.

As stated in \Cref{thm: sample-gibbs}, the instance of the algorithm is a LLL instance $I=(\{X_i\}_{i\in U}, \{A_v\}_{v\in V})$ with $n=|V|$ bad events. We assume that $I$ is $\gamma$-satisfiable, and let $D_I$ be its dependency graph.

Our sampling algorithm follows a natural two-step framework:
\begin{itemize}
    \item Initially, the algorithm generates a random assignment $\vec{Y}$ according to the product distribution $\vec{\nu}$. This can be easily achieved by having each node independently generate its own variables.
    \item 
    If none of the bad events in $\{A_v\}_{v\in V}$ occurs, then $\vec{Y}$ follows the correct distribution $\mu_I$. Otherwise, if some bad events $A_v$ have occurred on $\vec{Y}$, the algorithm locally applies some corrections to $\vec{Y}$ to ensure that the corrected assignment $\vec{Y}'$ follows the correct distribution~$\mu_I$.
\end{itemize}


Obviously, the nontrivial part of the algorithm is the above second step, in which the algorithm locally modifies a random assignment $\vec{Y}\sim\vec{\nu}$ to a new assignment $\vec{Y}'\sim\mu_I$ when some bad events $A_v$ occur on~$\vec{Y}$.


To simplify our exposition, we start by making the following idealized assumption. 
\begin{assumption}\label{ass:oneAv}
    There is only one node $v\in V$ such that its corresponding bad event $A_v$ occurs on $\vec{Y}$. 
\end{assumption}

Next, we will expose the main idea of our sampling algorithm 
by explaining how to resolve the sampling problem under this idealized assumption. 
%
The rest of this section is organized as follows:
\begin{enumerate}
    \item 
    We begin by approaching the sampling problem with modest objectives. 
    Our first goal is to sample with a bounded \emph{expected} complexity under \Cref{ass:oneAv}, 
    while also assuming a \emph{correlation decay} property. 
    With these assumptions, we present a sampling algorithm. (\Cref{sec:warm-up-with-cd})

    \item 
    The correlation decay assumed above may not always hold. 
    To address this, we introduce an \emph{augmentation} of the LLL instance, inducing desirable correlation decay by properly augmenting the LLL instance. 
    This enables us to remove the correlation decay assumption from the earlier mentioned algorithm, ensuring correct sampling from $\mu_I$ despite augmenting the LLL instance $I$. (\Cref{sec:warm-up-without-cd})

    \item 
    Finally, we introduce a \emph{recursive} sampling framework,
    which upgrades the aforementioned sampling algorithms with bounded expected complexity to algorithms with \emph{exponentially convergent} running time.
    This proves \Cref{thm: sample-gibbs} under \Cref{ass:oneAv}.
    (\Cref{sec:warm-up-recursive})


\end{enumerate}
The reliance on the idealized assumption \Cref{ass:oneAv} will be eventually eliminated in the next section (\Cref{sec:algorithm}), where we will formally prove the general result stated in \Cref{thm: sample-gibbs}.

\subsection{Warm-up: Sampling with idealized correlation decay}\label{sec:warm-up-with-cd}
In this part, we assume the correlation decay property as formally defined in~\Cref{def:correlation}.
This idealized correlation decay property is assumed for exposition purpose, and reliance on it will be eliminated later.
\begin{assumption}\label{ass:decay}
Any pair of disjoint $S,T\subseteq U$ are $\frac{1}{2n^{3}}$-correlated.
\end{assumption}

Let $\vec{Y}$ be generated according to the product distribution $\vec{\nu}$.
Under \Cref{ass:oneAv}, there is only one node $v\in V$ having $A_v$ occur on $\vec{Y}$. 
Let $S\triangleq\vbl(v)$ and $T\triangleq U\setminus \vbl(B_1(v))$. 
Notice that we have
\begin{align}
Y_T\sim \mu^{Y_S}_{T},\label{eq:marginal-distribution-oneAv}
\end{align}
because  $\vec{Y}$ is distributed as $\vec{\nu}$ and all bad events except $A_v$ do not occur on $\vec{Y}$.

%
%

\paragraph{An attempt to correct resampling.}
Our goal is to locally fix the random assignment $\vec{Y}$ around $v$ to make it distributed as $\mu_I$. A natural attempt is to resample evaluation of $Y_{U \backslash T}$ according to the correct marginal distribution $\mu^{Y_T}_{U\backslash T}$. 
This would produce a new assignment $\vec{Y}'$ which is distributed as: 
\[\forall \sigma\in \Sigma, \quad \Pr[\vec{Y}'=\sigma]=\Pr[\vec{Y}'_T=\sigma_T]\cdot \Pr[\vec{Y}'_{U\backslash T}=\sigma_{U\backslash T}\mid\vec{Y}'_T=\sigma_T]=\mu^{Y_S}_T(\sigma_T)\cdot\mu^{\sigma_T}_{U\backslash T}(\sigma_{U\backslash T}); \]
whereas, our goal is that each $\sigma\in \Sigma$ is sampled with probability $\mu_I(\sigma)=\mu_T(\sigma_T)\cdot\mu^{\sigma_T}_{U\backslash T}(\sigma_{U\backslash T})$. 
%

To remedy this, we apply the \emph{Bayes filters} introduced in~\cite{feng2020perfect}. 
A Bayes filter $\mathcal{F}=\mathcal{F}(\vec{Y}, S, T)$ is a trial whose success or failure is determined by $\vec{Y}, S, T$. 
The probability that $\mathcal{F}$ succeeds satisfies 
\begin{align}\label{eq:warm-up-bayes-filter}
  \Pr[\,\mathcal{F} \text{ succeeds}\,\mid \vec{Y}=\sigma] \propto \frac{\mu_{T}(\sigma_T)} {\mu^{Y_S}_{T}(\sigma_T)}, 
\end{align}
where $\propto$ are taken over all possible $\sigma\in\Sigma$ with $\mu^{Y_S}_{T}(\sigma_T)>0$. 
%
%
%
%

Now given a Bayes filter $\mathcal{F}$ satisfying \eqref{eq:warm-up-bayes-filter}, we conduct an experiment of $\mathcal{F}$,
and resample  $Y_{U\setminus T}\sim\mu^{Y_T}_{U\setminus T}$ if $\mathcal{F}$ succeeds. 
This will produce a $\vec{Y}'\sim \mu_I$ due to the Bayes law derived as follows:
\begin{equation}\label{eq:f-succeed-correct-distribution}
\begin{aligned}
\Pr[\vec{Y}'=\sigma\mid \mathcal{F}\text{ succeeds}\,]&=\frac{\Pr[\vec{Y}'=\sigma]\cdot\Pr[\mathcal{F}\text{ succeeds}\, \mid \vec{Y}=\sigma]}{\Pr[\mathcal{F}\text{ succeeds}]}\\
&\propto \left(\mu^{T_S}_T(\sigma_T)\cdot\mu^{\sigma_T}_{U\backslash T}(\sigma_{U\backslash T})\right)\cdot\left(\frac{\mu_{T}(\sigma_T)} {\mu^{Y_S}_{T}(\sigma_T)}\right) \propto \mu_I(\sigma).  
\end{aligned}\end{equation}
Otherwise, if $\mathcal{F}$ fails, we trivially resample the entire $\vec{Y}'\sim \mu_I$, which still ensures the correctness of sampling but uses global information. 
Overall, the above procedure correctly produces a $\vec{Y}'\sim \mu_I$.

It remains to construct a good Bayes filter using local generation and succeeding with high probability.
\paragraph{Construction of the Bayes filter $\mathcal{F}$.} Condition~\eqref{eq:warm-up-bayes-filter} of the Bayes filter can be expressed as
\begin{align}
  \Pr[\,\mathcal{F} \text{ succeeds}\,\mid\vec{Y}=\sigma] 
  &\propto \frac{\mu_{T}(\sigma_T)} {\mu^{Y_S}_{T}(\sigma_T)} \notag\\
(\text{the Bayes law})\quad &= \frac{\nu\left(\Omega^{\sigma_T}\right)\cdot \nu\left(\Omega^{Y_S}\right)}{\nu\left(\Omega\right) \cdot \nu\left(\Omega^{Y_S\wedge \sigma_T}\right)}\notag\\
&\propto  \frac{\nu\left(\Omega^{\sigma_T}\right)}{\nu\left(\Omega^{Y_S\wedge \sigma_T}\right)}\notag\\
&\triangleq f(\sigma_T).   \label{eq:bayes-filter-f-definition}
\end{align}
Note that $f(\sigma_T)\triangleq \frac{\nu\left(\Omega^{\sigma_T}\right)}{\nu\left(\Omega^{Y_S\wedge \sigma_T}\right)}$ can be computed locally from $B_2(v)$.
It is thus natural to define $\mathcal{F}$ as 
\[
\Pr[\,\mathcal{F} \text{ succeeds}\,]=\frac{f(Y_T)}{\max f},    
\]
where $\max f$ denotes the maximum value of $f(\sigma_T)$ taken over all all possible $\sigma_T\in\Sigma_T$ with $\mu^{Y_S}_{T}(\sigma_T)>0$.

It is obvious to see that for the Bayes filter $\mathcal{F}$ constructed as above,
an experiment of $\mathcal{F}$  can be conducted and observed locally from $B_2(v)$.
Furthermore, due to the correlation decay assumed in \Cref{ass:decay}, 
 $\mathcal{F}$~succeeds with high probability.
Let $\Sigma_S'$ be the set of possible assignments on the set $S$ of variables:
\[\Sigma_S'\triangleq\{\rho\in \Sigma_S\mid \text{ $\rho$ avoids bad events $A_v$ for all $v\in V$ s.t. $\vbl(v)\subseteq S$}\}.\]
Let $\tau^*$ be the assignment on $T$ that achieve the maximum $f(\tau^*)$. It holds that 
\begin{align}
    \Pr[\mathcal{F}\text{ succeeds}\,\mid\vec{Y}=\sigma]
    &=\frac{f(\sigma_T)}{f(\tau^*)}=\frac{\nu(\Omega^{\sigma_T})\cdot\nu(\Omega^{Y_S\wedge\tau^*})}{\nu(\Omega^{\tau^*})\cdot \nu(\Omega^{Y_S\wedge \sigma_T})}\notag\\
    &=\frac{\sum_{\rho\in \Sigma_S'}\nu(\Omega^{\rho\wedge \sigma_T})\cdot\nu(\Omega^{Y_S\wedge\tau^*})}{\sum_{\rho\in \Sigma_S'}\nu(\Omega^{Y_S\wedge \sigma_T})\cdot \nu(\Omega^{\rho\wedge \tau^*})}
    \geq 1-\frac{1}{2n^3}, \label{eq:bayes-filter-whp}
\end{align}
where the last inequality is derived from that $S$ and $T$ are  $\frac{1}{2n^3}$-correlated, guaranteed by \Cref{ass:decay}.

%
%

This simple sampling algorithm under \Cref{ass:oneAv} and \Cref{ass:decay} is described in \Cref{alg:sampling-with-cd}. 
By \eqref{eq:bayes-filter-whp} and the locality of the Bayes filter, the algorithm terminates within $O(1)$ rounds in expectation. 

\begin{algorithm}
\caption{\emph{Sampling-with-decay}($\vec{Y}$; $I,v$)}
\label{alg:sampling-with-cd}
\SetKwInOut{Input}{Input}
\SetKwInOut{Data}{Data}

\Input{LLL instance $I=\left(\{X_i\}_{i\in U}, \{A_v\}_{v\in V}\right)$,  node $v\subseteq V$;}
\Data{assignment $\vec{Y}=(Y_i)_{i\in U}$ stored globally that can be updated by the algorithm;}

\SetKwIF{withprob}{}{}{with probability}{do}{}{}{}
    
        define $S\triangleq \vbl(v)$ and $T\triangleq U\setminus \vbl(B_1(v))$\;
        
        \withprob{$\frac{f\left(Y_{T}\right)}{\max f}$, where $f$ is defined as in \eqref{eq:bayes-filter-f-definition}\label{line:resample-ideal-1-bayes-filter}}{ 
            update $\vec{Y}$ by resampling $Y_{U\setminus T}\sim \mu_{U\setminus T}^{Y_{T}}$\; \label{line:resample-ideal-1-resample}
            
            \tcp{\small  $\frac{f\left(Y_{T}\right)}{\max f}$ and $\mu_{U\setminus T }^{Y_{T}}$ can be evaluated locally within $B_{2}(v)$.}
        } 
        \Else{
            resample $\vec{Y} \sim \mu_I$\; 
            \tcp{\small Evaluating $\mu_I$  requires global information.}
        }
        
        \Return{}\;
\end{algorithm}

\subsection{Sampling without correlation decay via LLL augmentation}\label{sec:warm-up-without-cd}
The correlation decay asked by \Cref{ass:decay} may not always hold for general $\gamma$-satisfiable LLL instances.
A key idea is then to properly ``augment'' the LLL instance by including new bad events to enforce desirable decay of correlation.
%
%
%
This is highlighted by the following LLL augmentation lemma.

Recall the notions of balls  $B_\cdot(\cdot)$, shells $S_{[\cdot,\cdot]}(\cdot)$, and rings $\ring{}{[\cdot,\cdot]}{\cdot}$ formally defined in \Cref{sec:graph-notations}.

\begin{lemma}[LLL augmentation] \label{lem:augmenting-informal}
Let $I=\left(\{X_i\}_{i\in U},\{A_v\}_{v\in V}\right)$ be a LLL instance.
Let $\epsilon, \gamma\in (0,1)$, $\delta\in (0,\frac{\gamma}{2})$ and $\ell\geq \ell_0(\epsilon,\gamma,\delta)$ where $\ell_0(\epsilon,\gamma,\delta) = \Tilde{O}(\log \frac{1}{\epsilon}\log \frac{1}{\gamma}\log \frac{1}{\delta})$.
For any nonempty $\Lambda\subseteq V$, 
if the sub-instance $I\left(S_{[1,\ell](\Lambda)}\right)$ is $\gamma$-satisfiable,
%
then a new bad event $A_{\lambda}$ with $\lambda\not\in V$ can be constructed such that:
\begin{enumerate}
    \item $A_{\lambda}$ is defined on $\vbl(\lambda)=\ring{}{[1,\ell]}{\Lambda}$, and is constructed using local information within $B_{\ell+1}(\Lambda)$; 
    
    \item $A_{\lambda}$ occurs with probability at most $\delta$ on independent random variables $\{X_i\}_{i\in U}$, i.e.~$\nu(A_{\lambda})\le \delta$;
    
    \item  the variable sets $S=\vbl(\Lambda)$ and $T=U \setminus \vbl(B_{\ell}(\Lambda))$ are $\epsilon$-correlated after including $A_{\lambda}$ into $I$. 
\end{enumerate}
\end{lemma}


\begin{definition}[the augmenting event]\label{def:augmenting-event}
Let $I=\left(\{X_i\}_{i\in U},\{A_v\}_{v\in V}\right)$ be a LLL instance.
Let $\epsilon,\gamma\in(0,1)$, $\delta\in\left(0,\frac{\gamma}{2}\right)$, and $\ell= \ell_0(\epsilon, \gamma, \delta)$. 
Fix any nonempty $\Lambda\subseteq V$ that $I(S_{[1,\ell](\Lambda)})$ is $\gamma$-satisfiable.
We use 
$\aug{I}{\epsilon, \gamma, \delta}{\Lambda}$
to denote the new bad event $A_\lambda$ constructed in \Cref{lem:augmenting-informal}.  
\end{definition}

\begin{remark}
\Cref{lem:augmenting-informal} constitutes a critical component of  our sampling approach.
Beyond the algorithmic implications, it provides new insights into  correlation between variables constrained by local constraints.

This lemma aims to establish a decay of correlation between two distant regions  $S=\vbl(\Lambda)=\ring{}{0}{\Lambda}$ and $T=U \setminus \vbl(B_{\ell}(\Lambda))=\ring{}{[\ell,\infty]}{\Lambda}$.
However, it only assumes that all local constraints sandwiched between $S$ and $T$ are collectively easy to satisfy. 
Merely knowing this fact does not prevent $S$ and $T$ from being strongly correlated.
Remarkably, \Cref{lem:augmenting-informal} shows that the only obstacle to achieving correlation decay between $S$ and $T$ in this case lies in rarely occurring `bad' assignments between them. 
Consequently, $S$ and $T$ can be effectively de-correlated by introducing a new locally defined bad event over the region between $S$ and $T$ to prohibit those bad assignments that cause significant correlation between them.

This argument will be formalized in \Cref{sec:analysis-correlation-decay}, 
where a formal restatement of \Cref{lem:augmenting-informal} (referred to as \Cref{lem:augmenting}) will be introduced, and the lemma will be rigorously proved.


\end{remark}





\paragraph{Sampling using augmented LLL.}
Now we show how to utilize this LLL augmentation stated in \Cref{lem:augmenting-informal} to get rid of \Cref{ass:decay} in \Cref{alg:sampling-with-cd}, while still assuming  \Cref{ass:oneAv}.
%
%

%
%
As before, let $\vec{Y}$ be generated according to the product distribution $\vec{\nu}$;
and let $S\triangleq\vbl(v)$ and $T\triangleq U\setminus \vbl(B_1(v))$,
where  $v\in V$ is the only node (\Cref{ass:oneAv}) with the bad event $A_v$ occurring on $\vec{Y}$. 

Our goal is to utilize  \Cref{lem:augmenting-informal} to induce the necessary decay of correlation required for sampling.
However, augmenting LLL would inevitably alter the distribution  $\mu_I$. 
The saving grace is the following observation, suggesting that sampling may not be affected by certain local changes to the distribution.

\begin{observation}\label{observation:correct-efficient-sampling}
The sampling in \Cref{alg:sampling-with-cd} is correct and efficient as long as:
    \begin{itemize}
    \item (correctness) $Y_T$ follows the marginal distribution  $\mu^{Y_S}_{I,T}$;
    \item (efficiency) $S$ and $T$ are  $\frac{1}{2n^3}$-correlated.
\end{itemize}
\end{observation}

Using this observation, 
we can apply the LLL augmentation described in \Cref{lem:augmenting-informal} to ensure correlation decay between $S$ and $T$, 
while preserving the marginal distribution $\mu^{Y_S}_{T}$. 
Consequently, the resulting sampling process is both correct and efficient without relying on \Cref{ass:decay}.


Recall the new bad event $A_{\lambda}=\aug{I}{\epsilon, \gamma, \delta}{\Lambda}$ in \Cref{def:augmenting-event},
and define its complement $A_{\bar{\lambda}}$:
\begin{align}\label{eq:warm-up-def-event-A-lambda}
    A_{\lambda}\triangleq\aug{I}{\epsilon, \gamma, \delta}{\Lambda}, \text{ where }\lambda\not\in V,
    \quad\text{ and }\quad 
    A_{\bar{\lambda}}\triangleq\overline{A_{\lambda}}, \text{ where }\bar{\lambda}\not\in V.
\end{align}
Correspondingly, define the following two augmented LLL instances:
\begin{align}\label{eq:warm-up-def-hat-I}
\hat{I} = \left(\{X_i\}_{i\in U}, \{A_v\}_{v\in V\cup\{\lambda\}}\right)
\quad\text{ and }\quad 
\hat{I}'=\left(\{X_i\}_{i\in U}, \{A_v\}_{v\in V\cup \left\{\bar{\lambda}\right\}}\right).
\end{align}
Let $0<\zeta_0<1$ be a sufficiently small constant. 
We define the choice of parameter
\begin{align} \label{eq:def-parameter-without-cd-expected-complexity}(\epsilon_0,\gamma_0,\delta_0)\triangleq \left(\frac{1}{2n^3}, \frac{\gamma}{8}, \frac{\zeta_0\gamma}{24n^3}\right).
\end{align}
In the following,  let $A_{\lambda},  A_{\bar{\lambda}}, \hat{I}, \hat{I}'$ be defined as above on $\Lambda=\{ v\}$ with parameter $(\epsilon_0,\gamma_0,\delta_0)$.
Let $S=\vbl(\Lambda)$ and $ T=U\backslash\vbl(B_{\ell}(\Lambda))$, where $\ell=\ell_0(\epsilon_0,\gamma_0,\delta_0)=\Tilde{O}(\log^2 n\log^2\frac{1}{\gamma})$. 

By the same argument for~\eqref{eq:marginal-distribution-oneAv}, we still have $Y_T\sim \mu^{Y_S}_{I,T}$. 
Notice that $\hat{I}$ is identical $I$ except for the new bad event $A_{\lambda}$. 
Then, conditioned on $A_{\bar{\lambda}}$, we have $Y_T\sim\mu^{Y_S}_{\hat{I}, T}$ in the augmented instance $\hat{I}$.
By \Cref{lem:augmenting-informal} on our choice of parameter, $S$ and $T$ are $\frac{1}{2n^3}$-correlated in $\hat{I}$ . Thus, a calling to \emph{Sampling-with-decay}($\vec{Y};\hat{I},v$) (\Cref{alg:sampling-with-cd}) will produce a random assignment $\vec{Y}\sim \mu_{\hat{I}}$ within $O(1)$ rounds in expectation.

\paragraph{An idealized sampling algorithm.}
Then an idealized sampling procedure may proceed as follows. Our goal is to modify the input $\vec{Y}$ to a new $\vec{Y}\sim\mu_I$. This is achieved differently depending on whether the new bad event $A_{\lambda}$ occurs on $\vec{Y}$. If $A_{\lambda}$ occurs on $\vec{Y}$ (which happens with a small probability), we generate an entire new $\vec{Y}\sim\mu_I$ using global information. Otherwise, if $A_{\lambda}$ does not occur on $\vec{Y}$, we can use the following cleverer approach to produce a $\vec{Y}\sim\mu_I$:
\begin{itemize}
    \item With probability $P$, generate a $\vec{Y}\sim \mu_{\hat{I}}$, where $P$ is defined as
\begin{align}
P\triangleq \Pr_{\vec{X}\sim \mu_I}[\vec{X} \text{ avoids }A_{\lambda}].\label{eq:probability-P}    
\end{align}
Since $\vec{Y}$ avoids $A_{\lambda}$, this can be achieved by calling \emph{Sampling-with-decay}($\vec{Y};\hat{I},v$) to produce $\vec{Y}\sim \mu_{\hat{I}}$. 
\item Otherwise, generate entire $\vec{Y}\sim\mu_{\hat{I}'}$ using global information.
\end{itemize}
Overall, this correctly generates a $\vec{Y}\sim \mu_I=P\cdot\mu_{\hat{I}}+(1-P)\cdot \mu_{\hat{I}'}$
This procedure for sampling is idealized because it uses a value $P$ as defined in \eqref{eq:probability-P}  that is not easy to compute locally.

\paragraph{Bootstrpping the unknown threshold $P$.}
The probability $P$ in \eqref{eq:probability-P}  can be lower bounded. 
By \Cref{lem:augmenting-informal}, $A_\lambda$ occurs with probability at most  $\delta_0$. With the parameter specified in \eqref{eq:def-parameter-without-cd-expected-complexity}, we have 
\[P=\frac{\nu(\Omega_{\hat{I}})}{\nu(\Omega_{I})} = \frac{\nu(\Omega_{I})-\nu(\Omega_{\hat{I}'})}{\nu(\Omega_{I})}\geq 1 - \frac{\nu(A_{\lambda})}{\nu(\Omega_I) }\geq 1-\frac{\delta_0}{\gamma_0} \geq 1-\frac{1}{2n^3}.\] 
Therefore, when $A_{\lambda}$ does not occur on $\vec{Y}$, 
the above idealized sampling procedure can be realized as: 
first call the subroutine \emph{Sampling-with-decay}($\vec{Y};\hat{I},v$) defined in \Cref{alg:sampling-with-cd} to produce $\vec{Y}\sim \mu_{\hat{I}}$, 
and then with probability $\frac{1}{2n^3}$, compute the value of $P$ (which uses global information) and generate the entire $\vec{Y}\sim \mu_{\hat{I}'}$ with probability $2(1-P)\cdot n^3$.
The resulting algorithm is described in \Cref{alg:sampling-without-cd}.

\begin{algorithm}
\caption{\emph{Sampling-without-decay}($\vec{Y}$; $I,v$)}
\label{alg:sampling-without-cd}
\SetKwInOut{Input}{Input}
\SetKwInOut{Data}{Data}

\Input{LLL instance $I=\left(\{X_i\}_{i\in U}, \{A_v\}_{v\in V}\right)$,  $v\subseteq V$;}
\Data{assignment $\vec{Y}=(Y_i)_{i\in U}$ stored globally that can be updated by the algorithm;}

\tcp{\small Throughout the algorithm, $(\epsilon,\gamma,\delta)$ are defined as in \eqref{eq:def-parameter-without-cd-expected-complexity}. }
\tcp{\small $A_{\lambda},  A_{\bar{\lambda}}, \hat{I}, \hat{I}'$ are defined as in \eqref{eq:warm-up-def-event-A-lambda} and~\eqref{eq:warm-up-def-hat-I} with $\Lambda=\{ v\}$ and $(\epsilon_0, \gamma_0,\delta_0)$.}

\SetKwIF{withprob}{}{}{with probability}{do}{}{}{}

        define $\ell\triangleq \ell_0(\epsilon,\gamma,\delta)$, $S\triangleq \vbl(\Lambda)$ and $T\triangleq U\setminus \vbl(B_{\ell}(\Lambda))$\;

        \If{$\vec{Y}$ avoids $A_\lambda$}{
            call \emph{Sampling-with-decay}($\vec{Y};\hat{I}, v$), which is defined in \Cref{alg:sampling-with-cd}\; \label{line:without-decay-call-decay}

            

            \tcp{\small Now we have $\vec{Y}\sim\mu_{\hat{I}}$. }

            \withprob{$\frac{1}{2n^3}$}{
                evaluate $P\triangleq \Pr_{\vec{X}\sim \mu_I}[\vec{X} \text{ avoids }A_{\lambda}]$ and resample $\vec{Y}\sim \mu_{\hat{I}'} $ with probability $2(1-P)\cdot n^3$ \; \label{line:without-decay-globale-1}
            }
            \tcp{\small Now we have $\vec{Y}\sim \mu_I$. }
        } \Else{
            resample $\vec{Y} \sim \mu_{I}$\;\label{line:without-decay-globale-2}
            \tcp{\small Evaluating $P$, $\mu_I$, $\mu_{\hat{I}'}$ requires global information.}
        }
        
        \Return{}\;
\end{algorithm}

\Cref{alg:sampling-without-cd} may use global information, 
particularly in \Cref{line:without-decay-globale-1} with probability $\frac{1}{2n^3}$, and  in \Cref{line:without-decay-globale-2} when $A_{\lambda}$ occurs on $\vec{Y}$, which happens with probability at most $\delta_0\le \frac{1}{n^3}$ due to \Cref{lem:augmenting-informal}.\footnote{Technically, \Cref{ass:oneAv} may  bias the probability of $A_{\lambda}$ on $\vec{Y}$. 
However, \Cref{alg:sampling-without-cd} is solely used for exposition purposes,
and such bias caused by \Cref{ass:oneAv} will no longer be an issue beyond in \Cref{alg:sampling-without-cd}.}
%
As we discussed, the call to \emph{Sampling-with-decay}($\vec{Y};\hat{I},v$) at \Cref{line:without-decay-call-decay} returns within $O(1)$ rounds in expectation. 
At last, the construction of the augmented instance $\hat{I}$ takes $\ell=\Tilde{O}(\log^2 n\log^2\frac{1}{\gamma})$ rounds.
Overall, \Cref{alg:sampling-without-cd} returns a $\vec{Y}\sim\mu_{I}$ within $\Tilde{O}(\log^2 n\log^2\frac{1}{\gamma})$ rounds in expectation without relying on \Cref{ass:decay}.


\subsection{Recursive sampling with exponential convergence}
\label{sec:warm-up-recursive}
Both \Cref{alg:sampling-with-cd} and \Cref{alg:sampling-without-cd} occasionally rely on global information with a probability of $O(1/n^3)$. 
While the expected time complexity remains well bounded, achieving exponential convergence, as outlined in \Cref{thm: sample-gibbs}, would be more desirable in randomized running time.

We introduce a recursive sampling framework to achieve such exponential convergence in the time complexity.
And more importantly, this framework is also crucial for finally getting rid of \Cref{ass:oneAv}.


First, we adapt \Cref{alg:sampling-with-cd} to the recursive sampling framework. This adaptation is direct and simple. Then, we achieve the recursive sampling without \Cref{ass:decay}. Although part of this has already been embodied in \Cref{alg:sampling-without-cd}, there are still some substantial difficulties that need to be overcome.

\paragraph{Recursive sampling assuming correlation decay.} 
At first, we still assume \Cref{ass:decay}. 
The recursive sampling algorithm in this scenario is described in \Cref{alg:recursive-sampling-with-cd}.
There are two differences between this algorithm and \Cref{alg:sampling-with-cd}:
%
First, instead of taking just a node $v\in V$ as input in \Cref{alg:sampling-with-cd}, now \Cref{alg:recursive-sampling-with-cd} takes a subset $\Lambda\subseteq V$ as input. 
Second and more importantly, the original step of sampling using global information in \Cref{alg:sampling-with-cd},  is now replaced by a recursive call in \Cref{alg:recursive-sampling-with-cd}. 

\begin{algorithm}
\caption{\emph{RecursiveSampling-with-decay}($\vec{Y}$; $I,\Lambda$)}
\label{alg:recursive-sampling-with-cd}
\SetKwInOut{Input}{Input}
\SetKwInOut{Data}{Data}

\Input{LLL instance $I=\left(\{X_i\}_{i\in U}, \{A_v\}_{v\in V}\right)$,  subset $\Lambda\subseteq V$;}
\Data{assignment $\vec{Y}=(Y_i)_{i\in U}$ stored globally that can be updated by the algorithm;}

\SetKwIF{withprob}{}{}{with probability}{do}{}{}{}
    
        define $S\triangleq \vbl(\Lambda)$ and $T\triangleq U\setminus \vbl(B_1(\Lambda))$\;
        
        \withprob{$\frac{f\left(Y_{T}\right)}{\max f}$, where $f$ is defined as in \eqref{eq:bayes-filter-f-definition} \label{line:resample-ideal-3-bayes-filter}}{ 
            update $\vec{Y}$ by resampling $Y_{U\setminus T}\sim \mu_{U\setminus T}^{Y_{T}}$\; \label{line:resample-ideal-3-resample}
            
            \tcp{\small  $\frac{f\left(Y_{T}\right)}{\max f}$ and $\mu_{U\setminus T }^{Y_{T}}$ can be evaluated locally within $B_{2}(\Lambda)$.}
        } 
        \Else{
            \emph{RecursiveSampling-with-decay}($\vec{Y}$; $I,B_2(\Lambda)$)\;
            
            \tcp{\small Recursively sample $\vec{Y} \sim \mu_I$.}
        }
        
        \Return{}\;
\end{algorithm}

The correctness of recursive sampling relies on the following notion of \emph{conditional Gibbs} property, which is adapted from the same property introduced in \cite{feng2021dynamic,feng2020perfect} for Gibbs distributions.

\begin{definition}[conditional Gibbs]
\label{def:conditional-gibbs}
    Let $I=(\{X_i\}_{i\in U}, \{A_v\}_{v\in V})$ be a LLL instance. The random pair $(\vec{Y}, \RV{R})$,  where $\vec{Y}\in\Sigma$ is an assignment and $\RV{R}\subseteq V$ is a subset of events, is said to satisfy \emph{conditional Gibbs property} on instance $I$, if for any $R\subseteq V$ and $\sigma\in \Sigma_{\vbl(\Lambda)}$ with $\Pr[\RV{R}=R \wedge Y_{\vbl(R)}=\sigma]>0$, conditioned on $\RV{R}=R\wedge Y_{\vbl(R)}=\sigma$, it holds that $Y_{U\setminus \vbl(R)}\sim \mu^{\sigma}_{I, U\setminus \vbl(R)}$. 
\end{definition}

Intuitively, the random pair $(\vec{Y}, \RV{R})$ represents a ``partially correct'' sample $\vec{Y}$ with its problematic part coverred by $\RV{R}$.
The conditional Gibbs property guarantees that except for this problematic part, the sample is always distributed correctly.
This gives a key invariant condition for the correctness of recursive sampling.
%

The following is easy to verify for the assignment $\vec{Y}$ satisfying \Cref{ass:oneAv} with $\Lambda=\{v\}$. 

\begin{condition}\label{cond:warm-up-recursive-with-cd}
    $(\vec{Y}, \Lambda)$ satisfies conditional Gibbs property on instance $I$.
\end{condition}

The following  can be verified routinely  by a structural induction: 
As long as \Cref{cond:warm-up-recursive-with-cd} is satisfied by the input,
\Cref{alg:recursive-sampling-with-cd} terminates with probability 1 and returns a $\vec{Y}$ that is  distributed as $\mu_I$.  
This guarantees the correctness of \Cref{alg:recursive-sampling-with-cd} as a sampling algorithm.


To further bound the complexity of  \Cref{alg:recursive-sampling-with-cd}, we need to assume \Cref{ass:decay}.
With such assumption, by the same analysis as in \eqref{eq:bayes-filter-whp}, the Bayes filter in \Cref{line:resample-ideal-3-bayes-filter} succeeds with probability at least $1-\frac{1}{2n^3}$. Hence, for any $0<\epsilon<1$,  \Cref{alg:recursive-sampling-with-cd} terminates in $O(\log \frac{1}{\epsilon} / \log n)$ rounds with probability $1-\epsilon$. 


\paragraph{Recursive sampling with LLL augmentation.} 
Our goal here is to adopt \Cref{alg:sampling-without-cd} for sampling without the additional assumption about correlation decay,
into the recursive sampling framework.
The resulting algorithm will solve the problem in \Cref{thm: sample-gibbs} under \Cref{ass:oneAv}.


As in \Cref{alg:sampling-without-cd}, we need LLL augmentation.
%
The following notion of \emph{augmented conditional Gibbs property} is a variant of the conditional Gibbs property in \Cref{def:conditional-gibbs}, tailored with LLL augmentation.

\begin{definition}[augmented conditional Gibbs]
\label{def:aug-conditional-gibbs}
    Let  $\epsilon,\gamma\in(0,1)$, $\delta\in\left(0,\frac{\gamma}{2}\right)$, and $\ell=\ell_0(\epsilon, \gamma, \delta)$, where $\ell_0(\epsilon, \gamma, \delta)=\Tilde{O}(\log \frac{1}{\epsilon}\log \frac{1}{\gamma}\log \frac{1}{\delta})$ is defined as in \Cref{lem:augmenting-informal}.
    Let $I=(\{X_i\}_{i\in U}, \{A_v\}_{v\in V})$ be a LLL instance. The random pair $(\vec{Y}, \RV{R})$,  where $\vec{Y}\in\Sigma$ and $\RV{R}\subseteq V$,
    is said to satisfy \emph{augmented conditional Gibbs property} on instance $I$ with parameter $(\epsilon, \gamma, \delta)$, if for any $R\subseteq V$ with $\Pr[\RV{R}=R]>0$:
    \begin{enumerate}
        \item the sub-instance $I\left(S_{[1,\ell]}(R)\right)$ is $\gamma$-satisfiable;
        \item  for $S\triangleq \vbl(R)$, $T\triangleq U\setminus \vbl(B_{\ell}(R))$, for any $\sigma\in \Sigma_{\vbl(\Lambda)}$ with $\Pr[\RV{R}=R \wedge Y_{\vbl(R)}=\sigma]>0$, conditioned on that $\RV{R}=R\wedge Y_{\vbl(R)}=\sigma$, the assignment $Y_T$ follows the distribution $\mu^{\sigma}_{\hat{I}, T}$, i.e., \[
        \forall \tau\in\Sigma_T,\quad \Pr\left(Y_{T}=\tau\mid \RV{R}=\mathcal{R}\wedge Y_{S}=\sigma\right)=\mu^{\sigma}_{\widehat{I},T}(\tau),
        \]
        where $\hat{I}$ stands for the augmented LLL instance  $\hat{I}\triangleq\left(\{X_i\}_{i\in U}, \{A_v\}_{v\in V}\cup \{\aug{I}{R}{\epsilon,\gamma,\delta}\}\right)$ with $\aug{I}{R}{\epsilon,\gamma,\delta}$ as in \Cref{def:augmenting-event}. 
    \end{enumerate}
\end{definition}

Recall the choice of parameter $(\epsilon_0,\gamma_0,\delta_0)$ in \eqref{eq:def-parameter-without-cd-expected-complexity} %
and let $\ell\triangleq \ell_0(\epsilon_0,\gamma_0,\delta_0)$. 
Let $r\in \mathbb{N}$ be the minimal integer that $\vec{Y}$ avoids the bad event $\aug{I}{B_{r\cdot \ell}(v)}{\epsilon_0,\gamma_0,\delta_0}$. 
The following can be routinely verified on the instance $I$, the random assignment $\vec{Y}$, $\Lambda= B_{r\cdot \ell}(v)$, along with the parameter $(\epsilon,\gamma,\delta,\alpha)=(\epsilon_0,\gamma_0,\delta_0, \gamma_0)$.

\begin{condition}\label{cond:warm-up-recursive-without-cd}
The following hold:
    \begin{itemize}
        \item $0<\epsilon\leq \frac{1}{2}$, $0<\alpha \leq \gamma<1$ and $0<\delta<\zeta_0 \cdot \alpha$; 
        
        \item the LLL instance $I$ is $\alpha$-satisfiable and the sub-instance $I(V\setminus \Lambda)$ is $\gamma$-satisfiable;
        
        \item $(\vec{Y}, \Lambda)$ satisfies the augmented conditional Gibbs property on instance $I$ with parameter $(\epsilon, \gamma, \delta)$.
    \end{itemize}
\end{condition}

Our goal is then to modify the random assignment $\vec{Y}$ around the region $\Lambda$ to make it follow the correct distribution $\mu_I$, 
as long as \Cref{cond:warm-up-recursive-without-cd} is satisfied.
Adopting the definitions of $A_{\lambda},  A_{\bar{\lambda}}, \hat{I}, \hat{I}'$ and $P$ as in \cref{eq:warm-up-def-event-A-lambda,eq:warm-up-def-hat-I,eq:probability-P} with $\Lambda=B_{r\cdot \ell}(v)$ and $(\epsilon,\gamma,\delta)=(\epsilon_0,\gamma_0,\delta_0)$.
By our choice of $r$, the bad event $A_{\lambda}$ may never occur on $\vec{Y}$.
Then, \Cref{alg:sampling-without-cd} can be simiplied to the case where $\vec{Y}$ always aoids $A_{\lambda}$.

\paragraph{Recursive bootstrapping of marginal probabilities.}
A challenge, as we mentioned in \Cref{sec:warm-up-without-cd}, is that the true value of the marginal probability $P$ as defined in \eqref{eq:probability-P}  is hard to compute locally.
This is now resolved by a more sophisticated bootstrapping of the threshold $P$, 
by recursively estimating it within a progressively more accurate interval $[L, R]$.
The estimation is formally stated by the following lemma, which is proved in \Cref{sec:proof-estimate-augmenting}. 
It shows that the threshold $P$ can be estimated  exponentially more accurately as the radius of local information grows. 
%
%
\begin{lemma}\label{lem:estimate-augmenting}
    Let $I=\left(\{X_i\}_{i\in U}, \{A_v\}_{v\in V}\right)$ be a LLL instance. 
    Let $0<\epsilon<\frac{1}{2}$,
    $k\in\mathbb{N}^+$, 
    $0<\alpha_1\leq \alpha_2<1$
    and
    $\ell=\ell_0(\epsilon^k, \alpha_2, \alpha_1\cdot \epsilon^k)$, where the function $\ell_0$ is as defined in \Cref{lem:augmenting-informal}. 
    For any nonempty $\Lambda\subseteq V$ and an arbitrary event $A_{\lambda}$ defined on the variables in $\vbl(\Lambda)$, 
    assuming that $I$ is $\alpha_1$-satisfiable and  $I(V\setminus \Lambda )$ is $\alpha_2$-satisfiable, 
    there is a $\hat{P}\in(0,1)$ determined only by  $\Lambda$, $A_{\lambda}$, and $I(B_{\ell+1}(\Lambda))$, 
    such that
    $$
        P\triangleq \Pr_{\vec{X}\sim \mu_I}[\,\vec{X}\text{ avoids }A_{\lambda}\,] \in\left[\hat{P}-2\epsilon^k, \hat{P}+2\epsilon^k\right].
    $$
\end{lemma}

\begin{algorithm}
\caption{\RecursiveSampling($\vec{Y}$; $I, \Lambda, \epsilon, \gamma, \delta, \alpha$)}
\label{alg:recursive-sampling-without-cd}
\SetKwInOut{Input}{Input}
\SetKwInOut{Data}{Data}

\Input{LLL instance $I=\left(\{X_i\}_{i\in U}, \{A_v\}_{v\in V}\right)$,  subset $\Lambda\subseteq V$, parameter  $(\epsilon, \gamma, \delta, \alpha)$;}
\Data{assignment $\vec{Y}=(Y_i)_{i\in U}$ stored globally that can be updated by the algorithm;}

\tcp{\small Throughout the algorithm, $A_{\lambda}, A_{\bar{\lambda}}, \hat{I}, \hat{I}'$ are defined as in \eqref{eq:warm-up-def-event-A-lambda} and~\eqref{eq:warm-up-def-hat-I}.}

\SetKwIF{withprob}{}{}{with probability}{do}{}{}{}

    
initialize $i\gets 1$, and define $\ell_0\triangleq \ell_0(\epsilon,\gamma, \delta)$\; 

draw $\rho\in[0,1)$ uniformly at random\; \label{line:draw-rho}

\While{true\label{line:recursive-termination}}{


    ${\ell}_i \gets \ell_0\left(\zeta_0^i, \gamma, \alpha\cdot \zeta_0^i\right)$\;

    compute the smallest interval $[{L},{R}]$ containing $P\triangleq \Pr_{\vec{X}\sim \mu_I}[\,\vec{X}\text{ avoids }A_{\lambda}\,]$
    based on $\Lambda$, $A_{\lambda}$, $I\left(B_{\ell_0+{\ell}_i+1}(\Lambda)\right)$, assuming that $I$ is $\alpha$-satisfiable and $I\left(V\setminus B_{\ell_0+{\ell}_i+1}(\Lambda)\right)$ is $\gamma$-satisfiable\; \label{line:estimate}
    \tcp{\small By \Cref{lem:estimate-augmenting}, such interval $[{L},{R}]$ exists and satisfies ${R}-{L} \leq 4\zeta_0^i$.}



    \If{$\rho< {L}$ \label{line:first-if}}{
    \tcp{\small Enters the zone $[0,{L})\subseteq[0,P)$ for generating $\vec{Y}\sim\mu_{\hat{I}}$.}

            
                
        
        define $T\triangleq U\setminus \vbl(B_{\ell_0}(\Lambda))$\;

        {Define $f(\tau)\triangleq \frac{\nu\left(\Omega_{\hat{I}}^{\tau}\right)}{\nu\left(\Omega_{\hat{I}}^{Y_S\wedge \tau}\right)}$ for all $\tau\in \Sigma_T$ with $\nu\left(\Omega_{\hat{I}}^{Y_S\wedge \tau}\right)>0$, similar to \eqref{eq:bayes-filter-f-definition}}\;
        \withprob{$\frac{f\left(Y_{T}\right)}{\max f}$  \label{line:filter-soundness}}{ 
            update $\vec{Y}$ by redrawing $Y_{U\setminus T}\sim \mu_{\hat{I}, U\setminus T}^{Y_{T}}$; \label{line:resample}
            
            \tcp{\small  $\frac{f\left(Y_{T}\right)}{\max f}$ and $\mu_{\hat{I}, U\setminus T}^{Y_{T}}$ can be evaluated locally within $B_{\ell+1}(\Lambda)$.}
        } 
        \Else{
            initialize  $r\gets \ell_0+1$;

            
            \While{$\vec{Y}$ does not avoid the bad event $\aug{\hat{I}}{{1}/{2},\, \gamma,\, {\zeta_0\alpha}/{2}}{B_{r}(\Lambda)}$\label{line:violate1}}{
                grow the ball:  $r\gets r+\ell_0\left(\frac{1}{2}, \gamma,\frac{\zeta_0\alpha}{2}\right)$; \label{line:add-first-radius}
            }            
            
            
            \RecursiveSampling$\left(\vec{Y}; \hat{I}, B_{r}(\Lambda)\cup\{\lambda\}, \frac{1}{2}, \gamma, \frac{\zeta_0 \alpha}{2}, \frac{\alpha}{2}\right)$; \label{line:first-recursion}
        }
        
            \Return{}\;\label{line:first-break}
        }
        \ElseIf{$\rho\ge {R}$ \label{line:second-if}}{
        \tcp{\small Enters the zone $[{R},1)\subseteq[P,1)$ for generating $\vec{Y}\sim\mu_{\hat{I}'}$.}
        
            initialize $s\gets \ell_0 +1$\;

    
            \While{$\vec{Y}$ does not avoid the bad event $\aug{\hat{I}'}{1/2,\,\gamma,\, {\zeta_0\alpha  (1-R)}/{2}}{B_{s}(\Lambda)}$\label{line:violate2}}{
                grow the ball: $s\gets s+\ell_0\left(\frac{1}{2},\gamma, \frac{\zeta_0\alpha (1-R)}{2}\right)$; \label{line:add-second-radius}
            }
            \RecursiveSampling$\left(\vec{Y}; \hat{I}', B_{s}(\Lambda)\cup \left\{\overline{\lambda}\right\},  \frac{1}{2}, \gamma,  \frac{\zeta_0 \alpha  (1-R)}{2}, \alpha  (1-R) \right)$;\label{line:second-recursion}
            
            \Return{}\;\label{line:second-break}

        }
        \Else{
        \tcp{\small Enters the zone $[{L},{R})$ of indecision.}
        
                   enter the next iteration (and refine the estimation of $P$): $i\gets i+1$;  \label{line:add-estimate-radius}
        }
    }
\end{algorithm}

The high-level strategy for sampling can be outlined as follows. Given the current estimate interval $[L,R]$ of the threshold $P\in[L,R]$,  as provided by \Cref{lem:estimate-augmenting}, the algorithm operates as follows:
\begin{itemize}
    \item 
    With probability $L$, the algorithm is determined to enter the branch for sampling $\vec{Y}\sim \mu_{\hat{I}}$. 
    This can be resolved efficiently as in \Cref{alg:recursive-sampling-with-cd}, because of the correlation decay in the augmented instance~$\hat{I}$.
    \item 
    With probability $1-R$, the algorithm is determined to enter the branch for sampling $\vec{Y}\sim \mu_{\hat{I}'}$. 
    This can be resolved recursively with an appropriately enlarged neighborhood containing $\vbl(\lambda)$, 
    ensuring that \Cref{cond:warm-up-recursive-without-cd} is still satisfied invariantly.
    \item 
    Otherwise, the algorithm enters the so-called `zone of indecision', where it enlarges the local radius to gather more information in order to obtain a more accurate estimate $[L, R]$ of the threshold~$P$.
\end{itemize}
%
%
Overall, the above procedure generates a random assignment $\vec{Y}\sim \mu_I$. 
Throughout the recursive calls, we ensure to maintain the invariant \Cref{cond:warm-up-recursive-without-cd}. 
%
%
The procedure \RecursiveSampling($\vec{Y}$; $I, \Lambda, \epsilon, \gamma, \delta, \alpha$) is detailed in \Cref{alg:recursive-sampling-without-cd}.

Our use of `zones of indecision' for recursive bootstrapping of marginal probabilities is inspired from the Anand-Jerrum algorithm introduced in \cite{anand2022perfect} for solving perfect simulation of Gibbs distributions.%

The correctness of \Cref{alg:recursive-sampling-without-cd} follows from a structural induction.
The complexity of \Cref{alg:recursive-sampling-without-cd} is challenging to analyze due to the random recursion, nevertheless, we apply a potential method to bound it.
%
\begin{lemma}\label{lem:recursive-sample-correctness-complexity}
Assume that \Cref{cond:warm-up-recursive-without-cd} is satisfied by the input of
\RecursiveSampling$(\vec{Y}; I,\Lambda,\epsilon,\gamma,\delta,\alpha)$.
\begin{enumerate}
    \item \label{lem:item:recursive-sample-correctness}
    After \RecursiveSampling$(\vec{Y}; I,\Lambda,\epsilon,\gamma,\delta,\alpha)$ returns, $\vec{Y}$ follows the distribution $\mu_I$.
    \item \label{lem:item:recursive-sample-complexity}
    For any $0<\eta<1$, with probability $1-\eta$, \RecursiveSampling$(\vec{Y};I,\Lambda,\epsilon,\gamma,\delta,\alpha)$ accesses $I(B_{r}(\Lambda))$ and updates $Y_{\vbl(B_{r}(\Lambda))}$, where $r=\ell_0(\epsilon,\gamma, \delta)+\tO\left(\log\frac{1}{\gamma}\cdot \log^4 \frac{1}{\eta} + \log\frac{1}{\gamma}\cdot \log^2 \frac{1}{\eta}\cdot \log \frac{1}{\alpha}\right)$.
\end{enumerate}
%
\end{lemma}

\begin{remark}
It is important to note that \Cref{ass:oneAv} is not assumed in the statement of \Cref{lem:recursive-sample-correctness-complexity}.
Instead, the lemma holds as long as that \Cref{cond:warm-up-recursive-without-cd} is satisfied. 
The scenario described in \Cref{ass:oneAv} can in fact be incorporated into \Cref{cond:warm-up-recursive-without-cd} as a special case where $\Lambda=B_{r\cdot\ell}(v)$.

\Cref{lem:recursive-sample-correctness-complexity} will be formally proved in \Cref{sec:analysis-resampling}.
\end{remark}

\section{The Main Algorithm: General Case} \label{sec:algorithm}
In this section, we formally present the main sampling algorithm in \Cref{thm: sample-gibbs}.
Recall that the input instance is a $\gamma$-satisfiable LLL instance $I=\left(\{X_i\}_{i\in U},\{A_v\}_{v\in V}\right)$ with $n=|V|$ bad events. 
The network $G=D_I$ is its dependency graph.
Our goal is to draw a sample according to the LLL distribution $\mu_I$.


As discussed in \Cref{sec:special-case}: 
the sampling algorithm initially generates a random assignment $\vec{Y}$ according to the  product distribution $\nu$; 
and if there are some bad events occurring on $\vec{Y}$, 
the algorithm then tries to locally fix this prudctly generated $\vec{Y}$ to follow the the correct joint distribution $\mu_I$.

Through a series of expositions in the last section, the aforementioned problem was partially resolved in \Cref{sec:warm-up-recursive} under an idealized assumption (\Cref{ass:oneAv}), 
which posits that there is only one bad event $A_v$ that may occur on $\vec{Y}$ drawn from the product distribution $\nu$ initially.

To remove the dependency on \Cref{ass:oneAv}, 
the algorithm for the general case additionally clusters the erred bad events into local balls that are suitably far apart from each other.
These erred balls are then fixed one-by-one using \Cref{alg:recursive-sampling-without-cd} presented in the last section, 
using the sequential local paradigm (\SLV{}) described in \Cref{subsec:SLOCAL-LV}. 
Finally, according to \Cref{thm: SLV-local to local}, this sequential local procedure is converted into a \LOCAL{} algorithm.
Compared to \Cref{sec:special-case}, the new elements introduced in this section, namely clustering and \SLOCAL{} paradigm, are technically more conventional to local computation, 
while their inclusion does indeed introduce additional complexity to the presentation of the algorithm.


The algorithm for the general case consists of three phases: \emph{initialization}, \emph{clustering}, and \emph{resampling}. 
The first phase runs in the \LOCAL{} model and terminates within fixed $\Tilde{O}( \log^3 n)$ rounds.
\begin{itemize}
    \item \textbf{Initialization}: produce a random assignment $\vec{Y}=(Y_i)_{i\in U}$ distributed as the product distribution $\nu$,
    and a random set $\RV{R}\subseteq V$ of clustering centers  used by the next phase. (\Cref{subsec:algorithm-initialization})
\end{itemize}
Then this random $(\vec{Y},\RV{R})$ is then passed to a $2$-scan \SLV{} algorithm that runs on the same network, 
where the two scans correspond to the next two phases of the algorithm, respectively. 
\begin{itemize}
    \item \textbf{Clustering}: cluster the bad events that occur on $\vec{Y}$ into sufficiently far apart balls, ensuring that \Cref{cond:warm-up-recursive-without-cd} (indeed, a slightly stronger condition, \Cref{cond:clustering}) is satisfied invariantly. (\Cref{subsec:algorithm-clustering})
    \item \textbf{Resampling}:  properly fix the assignments on the balls to make $\vec{Y}$ follow the correct distribution $\mu_I$, which utilizes the \RecursiveSampling{} described in \Cref{alg:recursive-sampling-without-cd} as a black box.  (\Cref{subsec:algorithm-resampling})
\end{itemize}
Finally, by \Cref{thm: SLV-local to local}, the $2$-scan \SLV{} algorithm can be transformed to a \LOCAL{} algorithm, which altogether with the initialization phase, 
gives us the \LOCAL{} algorithm claimed in \Cref{thm: sample-gibbs}. 

At last, we wrap up the proof of \Cref{thm: sample-gibbs} in \Cref{subsec:wrpping-up}.

\subsection{The \emph{Initialization} phase}\label{subsec:algorithm-initialization}


The goal of this phase is to generate a random assignment $\vec{Y}=(Y_i)_{i\in U}$  and a random set $\RV{R}\subseteq V$ satisfying:

\begin{condition}\label{cond:initialization-stage}
The followings hold for $\vec{Y}=(Y_i)_{i\in U}$ and $\RV{R}\subseteq V$: 
    \begin{enumerate}
        \item $\vec{Y}$ follows the product distribution $\nu$. 
        \item For any $v\in V$, if $\vec{Y}$ does not avoid $A_v$ , 
        then  $\dist(u,v)\le {d\cdot \log n\cdot \log \log \log n}$ for some $u\in \RV{R}$, where $d$ is a large enough universal constant to be specified later.
        \item For any $0<\epsilon<1$, we have $|\RV{R}|=O\left(\log n \cdot \log \log n \cdot \log \frac{1}{\gamma} \cdot\log \frac{1}{\epsilon}\right)$ with probability at least $1-\epsilon$.
    \end{enumerate}
\end{condition}

Intuitively, this asks for generating a random vector $\vec{Y}=(Y_i)_{i\in U}$ obeying the marginal distributions $\nu_i$, 
and a small set $\RV{R}\subseteq V$ of ``centers'', 
such that all bad events that are made occur by $\vec{Y}$ are not far away from some center.
In the \LOCAL{} model, this can be achieved rather straightforwardly with the help of network decomposition,
which is a major building block for distributed algorithms.

\begin{definition}[network decomposition]\label{def:network-decomposition}
A weak $(c, d)$-network decomposition of $G=(V,E)$ is a pair $(\mathcal{S}, \mathcal{C})$, where $\mathcal{S}=\{S_1,S_2,\ldots\}$ is a partition of $V$ into vertex subsets, each with diameter at most $d$;
and $\mathcal{C}:\mathcal{S}\to [c]$ is a proper coloring of $\mathcal{S}$ such that $\mathcal{C}(S_1)\neq \mathcal{C}(S_2)$ for any $S_1,S_2\in \mathcal{S}$ with $\dist_G(S_1,S_2) = 1$.
\end{definition}



We adopt the recent bound for deterministic network decomposition from \cite{ghaffari2023improved}.

\begin{lemma}[\cite{ghaffari2023improved}]\label{lem:time-complexity-of-network-decomposition}
    There is a deterministic \LOCAL{} algorithm that, on any network $G$ with $n$ nodes, computes a weak $( O(\log n), O(\log n\cdot \log \log \log n))$-network decomposition of $G$ in { $\Tilde{O}(\log^3n)$} rounds. 
\end{lemma}

Note that each node $v\in V$ corresponds to a bad event $A_v$ defined on the variables in $\vbl(v)$, 
and one variable may appear in the $\vbl(v)$ for multiple $v$'s.
Alternatively, one can assign each variable to a unique node through the following partition:
\[
\forall v\in V,\quad U_v\triangleq\{i\in\vbl(v)\mid v\text{ has the smallest ID among all $v'\in V$ with }i\in\vbl(v')\}.
\]
It is obvious that $\{U_v\mid v\in V\}$ is a partition of $U$, assuming that $U=\vbl(V)$ is the set of variables ever appearing in any bad events. 
Also, each node $v\in V$ can compute $U_v$ within one round.

The algorithm for producing the desirable $(\vec{Y},\RV{R})$ is as follows.
At first, each node $v\in V$ locally draws $Y_i\sim\nu_i$ independently for all $i\in U_v$.
This gives the random assignment $\vec{Y}=(Y_i)_{i\in U}$, which obviously follows the marginal distributions $\nu_i$.
Next, we compute the $\RV{R}\subseteq V$. 
First, each node $v\in V$ checks if $A_v$ is avoided by $\vec{Y}$, which costs one round since $\vbl(v) \subseteq \bigcup_{w\in N^+(v)} U_{w}$. 
Then, we construct a weak $(c\cdot \log n, d\cdot  \log n\cdot \log \log \log n)$-network decomposition $(\mathcal{S}, \mathcal{C})$ 
for some suitable constant $c,d\in \mathbb{N}$, 
which according to \Cref{lem:time-complexity-of-network-decomposition}, can be done within $\tO(\log^3 n)$ rounds. 
After that, within each cluster $S\in \mathcal{S}$, 
each node $v\in S$ checks whether the union bad event $A_S\triangleq \bigcup_{v\in S} A_v$ is avoided by $\vec{Y}$,
and marks $v$ in $\RV{R}$ iff $A_S$ occurs and $v$ has the smallest ID within the cluster $S$.
This guarantees that for any bad event $A_v$ that occurs on $\vec{Y}$, the node  $v\in V$ must be $d\cdot\log n\cdot\log\log\log n$-close to a node in $\RV{R}$,
and the construction of such $\RV{R}$ takes $\tO(\log n)$ rounds given the network decomposition $(\mathcal{S}, \mathcal{C})$ , because the diameter of each cluster $S\in \mathcal{S}$ is at most $d\cdot\log n\cdot\log\log\log n$.
Together, the following is achieved.

\begin{lemma}\label{lem:initialization-stage}
        The initialization phase outputs $(\vec{Y}, \RV{R})$ satisfying \Cref{cond:initialization-stage} within $\Tilde{O}( \log^3 n)$ rounds.
\end{lemma}

It only remains to verify that $|\RV{R}|=O\left(\log n \cdot \log \log n \cdot \log \frac{1}{\gamma} \cdot\log \frac{1}{\epsilon}\right)$ with probability at least $1-\epsilon$, which follows from the Chernoff bound.
A formal proof is included in \Cref{sec:proof-initialization-stage} for completeness. 

\subsection{The \emph{Clustering} phase} \label{subsec:algorithm-clustering}
The random assignment $\vec{Y}$ and node set $\RV{R}\subseteq V$ constructed in the initialization phase, is passed to a \SLV{} algorithm,
which runs on the same network $G=D_I$ and takes $\RV{R}$ as the set of active nodes.
The \SLV{} algorithm runs in two scans, where the first scan is the \emph{Clustering} phase.

Each node $v\in V$ maintains a local memory $M_v$, which initially stores its UID $\id(v)$, the private random bits, 
the part of the random assignment $(Y_i)_{i\in U_v}$ and the indicator of whether $v\in \RV{R}$.
%
%
The total order assumed on $V$ is a natural one: $u <  v$ iff $\id(u) <  \id(v)$ for any $u,v \in V$. 
%

The \SLV{} algorithm scans the active nodes  $v\in \RV{R}$ in order, 
and computes two parameters $p_v\in V\cup \{\perp\}$ and $r_v\in \mathbb{N}\cup\{\perp\}$ for each $v\in \RV{R}$,
%
which defines a collection of balls in the network $G$: 
\begin{align}\label{eq:def-clustering-balls}
        \RV{B}\triangleq \left\{B_{r_v}(p_v) \mid v\in \RV{R}\land p_v\neq \perp\land r_v\neq \perp\right\}.
\end{align}
%
The goal is to construct a collection $\RV{B}$ of far-apart and reasonably small balls, 
which together with the random assignment $\vec{Y}$, 
satisfy a clustered version of the augmented conditional Gibbs property.


The following notion of \emph{clustered conditional Gibbs property} is a refined variant of the augmented conditional Gibbs property defined in \Cref{def:aug-conditional-gibbs}, adapting to random clustering.

\begin{definition}[{clustered conditional Gibbs}]
\label{def:clustered-conditional-gibbs}
    Let  $\epsilon,\gamma\in(0,1)$, $\delta\in\left(0,\frac{\gamma}{2}\right)$, and $\ell=\ell_0(\epsilon, \gamma, \delta)$, where $\ell_0(\epsilon, \gamma, \delta)$ is defined as in \Cref{lem:augmenting-informal}.
    Let $I=\left(\{X_i\}_{i\in U},\{A_v\}_{v\in V}\right)$ be a LLL instance. 
    A random pair $(\vec{Y},\RV{\mathcal{B}})$, where $\vec{Y}=(Y_i)_{i\in U}\in\Sigma$ and $\RV{B}\subseteq {2^V}$,
    is said to satisfy the \emph{clustered conditional Gibbs property} on instance $I$ with parameter $(\epsilon,\gamma,\delta)$, if for any $\mathcal{B}\subseteq {2^V}$ with $\Pr(\RV{B}=\mathcal{B})>0$:
    \begin{enumerate}
        \item 
        the sub-instance $I\left(B_{\ell}(\Lambda)\setminus\Lambda\right)$ is $\gamma$-satisfiable for every $\Lambda \in \mathcal{B}$;
        
        \item for $S\triangleq \bigcup_{\Lambda \in \mathcal{B}} \vbl(\Lambda)$, $T\triangleq U\setminus \bigcup_{\Lambda\in \mathcal{B}} \vbl(B_{\ell}(\Lambda))$,
        for any $\sigma\in \Sigma_{S}$ with $\Pr\left(\RV{B}=\mathcal{B}\wedge Y_{S}=\sigma\right)>0$, 
        conditioned on that $\RV{B}=\mathcal{B}\land Y_{S}=\sigma$, the assignment $Y_{T}$ follows the marginal distribution $\mu^{\sigma}_{\widehat{I},T}$, i.e.
        \[
        \forall \tau\in\Sigma_T,\quad \Pr(Y_{T}=\tau\mid \RV{B}=\mathcal{B}\wedge Y_{S}=\sigma)=\mu^{\sigma}_{\widehat{I},T}(\tau),
        \]
        where $\widehat{I}$ stands for the LLL instance defined by $\widehat{I}=\left(\{X_i\}_{i\in U}, \{A_v\}_{v\in V} \cup \left\{\aug{I}{\epsilon, \gamma,\delta}{\Lambda} \mid \Lambda\in \mathcal{B}\right\}\right)$ 
        and recall that $\aug{I}{\epsilon, \gamma,\delta}{\Lambda}$ represents the bad event constructed in \Cref{lem:augmenting-informal}.
    \end{enumerate}
    All the above balls $B_\cdot(\cdot)$ are defined in the dependency graph $D_I$ of the LLL instance $I$.
\end{definition}

Recall the choice of parameter $(\epsilon_0,\gamma_0,\delta_0)$ in \eqref{eq:def-parameter-without-cd-expected-complexity} and let $\ell=\ell_0(\epsilon_0, \gamma_0,\delta_0)$, which is defined as in \Cref{lem:augmenting-informal} and formally in~\eqref{eq:def-ell-0}.
%
Our goal is to construct $\RV{B}$ so that $(\vec{Y},\RV{B})$ satisfies the following condition.

\begin{condition}\label{cond:clustering}
The random assignment $\vec{Y}=(Y_i)_{i\in U}$
and random collection of node sets $\RV{B}\subseteq {2^V}$,
satisfy
\begin{enumerate}
    \item $\dist_G(\mathcal{B}_1, \mathcal{B}_2) \geq 2(\ell+2)$ for any distinct $\mathcal{B}_1,\mathcal{B}_2\in \RV{B}$; 
        \item $(\vec{Y},\RV{B})$ satisfies the clustered conditional Gibbs property on instance $I$ with parameter $(\epsilon_0, \gamma_0,\delta_0)$.
\end{enumerate}
\end{condition}

The following is the \SLV{} algorithm for constructing such $\RV{B}\subseteq {2^V}$.
Each node $v\in \RV{R}$ maintains a variable $b_v\in V\cup \{\perp\}$ in its local memory $M_v$, which is initialized to $\perp$.
%
The \SLV{} algorithm sequentially processes each active node $v\in \RV{R}$ in order.
The pseudocode is given in \Cref{alg:clustering}.

After all active nodes $v\in \RV{R}$ have been processed, the 1st scan of the \SLV{} algorithm terminates, and the collection $\RV{B}$ of balls are constructed as in~\eqref{eq:def-clustering-balls} from the centers $p_v$ and radius $r_v$ computed for $v\in \RV{R}$.
Formally, the following are guaranteed for the clustering $\RV{B}$ and the random assignment $\vec{Y}$.

\begin{algorithm}[H]
\caption{The \SLV{} algorithm for Clustering at node $v\in \RV{R}$} \label{alg:clustering}

set $p_v\gets v$ and $r_v\gets 1+d\cdot \log n\cdot \log \log \log n$\;

\While{true}{
    \If{there exist $u\in \RV{R}\setminus\{v\}$ and $w\in B_{2(\ell+2)+r_v}(p_v)$ such that $b_w=u$\label{alg:clustering-line:if-1}}{
    let $c\in V$ be the node with the smallest $\id(c)$ satisfying  \hspace{200pt}
    $\dist(p_v,c)\leq r_u+\ell+2$ and $\dist(p_u,c)\leq r_v+\ell+2$\; 
    \tcp{\small Such $c\in V$ must exist since $\dist(p_u,p_v)\leq r_u+r_v+2(\ell+2)$.}
    set $p_v\gets c$  and  $r_v\gets r_u+r_v+2\cdot (\ell+2)$\; \label{line:6}
    
    $\forall w\in B_{r_u}(p_u)$:\,\, set $b_w\gets \perp$\; 
    
    set $p_u\gets \perp$ and $r_u\gets \perp$\; \label{line:7}
    
}

    \ElseIf{$\vec{Y}$ makes the bad event $\aug{I}{\epsilon_0, \gamma_0,\delta_0}{ B_{r_v}(p_v)}$ occur\label{alg:clustering-line:if-2}}
    {
        update the radius of the ball responsible for node $v$ to  $r_v\gets r_v+\ell$\; \label{line:9}

    }
    \Else{\label{alg:clustering-line:else}
        $\forall u\in B_{r_v}(p_v)$:\,\, set $b_u\gets v$\;
        
        \Return{};
    }
}

\end{algorithm}

\begin{lemma}\label{lem:clustering-correctness}
Assume \Cref{cond:initialization-stage}.
The followings hold after \Cref{alg:clustering} is sequentially executed on all node $v\in\RV{R}$ in the ascending order of $\id(v)$, 
which computes the values of $p_v\in V\cup \{\perp\}$ and $r_v\in \mathbb{N}\cup\{\perp\}$ for each $v\in \RV{R}$ and thereby constructs $\RV{B}\triangleq \left\{B_{r_v}(p_v) \mid v\in \RV{R}\land p_v\neq \perp\land r_v\neq \perp\right\}$.
\begin{enumerate}
    \item \label{lem:item:clustering-well-define}
    For any distinct $u,v\in \RV{R}$, if $p_u,p_v,r_u,r_v\not\in\{\perp\}$ then $B_{r_u}(p_u)\cap B_{r_v}(p_v)=\emptyset$, which guarantees that
    each ball $B=B_{r_v}(p_v)\in \RV{B}$ is uniquely identified by some node $v\in\RV{R}$ with $p_v,r_v\not\in\{\perp\}$.
    \item \label{lem:item:clustering-correctness}
    $(\vec{Y}, \RV{B})$ satisfies \Cref{cond:clustering}. 
    \item \label{lem:item:clustering-complexity}
    For any $0<\eta<1$, with probability at least $1-\eta$, the sum of radii of all balls in $\RV{B}$ is bounded as
    \[
    \mathcal{D}\triangleq\sum_{\substack{v\in \RV{R}\\p_v\neq \perp,r_v\neq \perp}} r_v=\tO\left(|\RV{R}| \cdot \log^2 n \cdot\log^2  \frac{1}{\gamma}\cdot \log\frac{1}{\eta}\right).
    \]
\end{enumerate}
\end{lemma}


\Cref{lem:clustering-correctness} is formally proved in \Cref{sec:proof-clustering-correctness}.
The following corollary follows easily since $\mathcal{D}$ is monotonically increasing during the scan and upper bounds the radii of the \SLV{} algorithm.

\begin{corollary}\label{lem:clustering-complexity-2}
For any $0<\eta<1$, with probability at least $1-\eta$, 
\Cref{alg:clustering} returns at every node $v\in \RV{R}$ within radius $\tO\left(|\RV{R}| \cdot \log^2 n \cdot \log^2 \frac{1}{\gamma}\cdot\log\frac{1}{\eta}\right)$.
\end{corollary}

\subsection{The \emph{Resampling} phase}\label{subsec:algorithm-resampling}
The second scan of the \SLV{} algorithm comprises the \emph{Resampling} phase.
The algorithm scans the active nodes $v\in\RV{R}$ sequentially in the order of their IDs,
and properly resamples part of the random assignment $\vec{Y}$ generated in the \emph{Initialization} phase, 
over the balls in $\RV{B}$ constructed in the \emph{Clustering} phase,
to finally obtain the random satisfying assignment distributed correctly as the LLL distribution $\mu_I$.

According to \Cref{lem:recursive-sample-correctness-complexity}, the task of resampling mentioned above can be resolved by \Cref{alg:recursive-sampling-without-cd} as long as \Cref{cond:warm-up-recursive-without-cd} is satisfied invariantly during the sequential and local execution of \Cref{alg:recursive-sampling-without-cd}.

To adapt \Cref{alg:recursive-sampling-without-cd}, designed for resampling a single local neighborhood, to the more general scenario where multiple balls need to be fixed one by one,
we introduce the \emph{Substituting} trick. 
This abstraction is formally described in the following technical lemma, the proof of which can be found in \Cref{sec:proof-substituting}.

\begin{lemma} \label{lem:substituting}
    Let  $\epsilon,\gamma\in(0,1)$, $\alpha\in(0,\gamma]$, $\delta\in\left(0,\frac{\gamma}{2}\right)$, and $\ell=\ell_0(\epsilon, \gamma, \delta)$ which is defined as in~\eqref{eq:def-ell-0}.
    Let $I=\left(\{X_i\}_{i\in U},\{A_v\}_{v\in V}\right)$ be a LLL instance. 
    For any nonempty set $\Lambda \subseteq V$ of bad events, 
    any $\sigma\in \Sigma_{\vbl(\Lambda)}$,
    if $I$ is $\alpha$-satisfiable and $I(B_{\ell}(\Lambda)\setminus \Lambda)$ is $\gamma$-satisfiable, 
    then we can construct a new random variable $X_{\beta}$ (with domain $\Sigma_{\beta}$ and distribution $\nu_{\beta}$) with $\beta\notin U$ and a new bad event $A_{\kappa}$ with $\kappa\notin V$ satisfying the followings: 
    \begin{enumerate}
        \item \label{lem:substituting-item-1} The constructions of $X_{\beta}$ and $A_{\kappa}$ depend only on the specifications of $\Lambda$, $\sigma$, $(\epsilon, \gamma,\delta)$  and  $I(B_{\ell+1}(\Lambda))$.
        The event $A_{\kappa}$ is defined on the random variables in $\vbl(\kappa)=\vbl(B_{\ell+1}(\Lambda))\setminus \vbl(B_{\ell}(\Lambda))\cup \{\beta\}$.  
        
        \item \label{lem:substituting-item-2}Let $T=U\setminus \vbl(B_{\ell}(\Lambda))$. 
        For any $W\subseteq T$, any $\omega \in \Sigma_{W}$ and $\overline{\omega}\in \Sigma_{\overline{W}}$, where $\overline{W}=T\setminus W$, we have 
        \[\mu^{\omega \wedge \sigma}_{\hat{I},\overline{W}} (\overline{\omega}) = \mu^{\omega}_{I_{\sigma},\overline{W}} (\overline{\omega}), \]
        where $\hat{I}$ stands for the LLL instance defined by $\hat{I}=\left(\{X_i\}_{i\in U}, \{A_v\}_{v\in V}\cup \left\{\aug{I}{\epsilon,\gamma,\delta}{\Lambda}\right\}\right)$, where $\aug{I}{\epsilon, \gamma,\delta}{\Lambda}$ represents the bad event constructed in \Cref{lem:augmenting},
        and $I_{\sigma}$ stands for the LLL instance defined by $I_{\sigma}=\left(\{X_i\}_{i\in (U\setminus \vbl(B_{\ell}(\Lambda)))\cup \{\beta\}}, \{A_v\}_{v\in (V\setminus B_{\ell+1}(\Lambda)) \cup \{\kappa\}}\right)$.
        
        \item \label{lem:substituting-item-3} The LLL instance $I_{\sigma}$ is $(1-\epsilon)\cdot (\alpha-\delta)$-satisfiable. 
    \end{enumerate}
    All balls $B_\cdot(\cdot)$ above are defined in the dependency graph $D_I$ of the original LLL instance $I$.
\end{lemma}

The core idea behind the above abstraction is to simulate each ball using a single variable. 
This allows us to translate the multi-ball invariant \Cref{cond:clustering} into \Cref{cond:warm-up-recursive-without-cd},
thus enabling us to reuse  \Cref{alg:recursive-sampling-without-cd} to tackle the challenge of multi-ball resampling.
It is worth noting that this step is introduced primarily for the convenience of reusing \Cref{alg:recursive-sampling-without-cd} with the current formulation.
However, even without this step, the sampling problem can still be resolved, albeit with a more intricate formulation in \Cref{alg:recursive-sampling-without-cd}.

%
With this abstraction, the \SLV{} algorithm for the \emph{Resampling} phase can be described as follows.
For the exposition, we write $V=\{v_1,v_2,...,v_{n}\}$, where the nodes are sorted in ascending order according to their IDs. 
The \SLV{} algorithm scans the nodes in $\RV{R}$ in this order.
Suppose that a node $v=v_k$ with $p_v,r_v\not\in\{\perp\}$ is being processed, while the local algorithm is oblivious to its rank $k$.
We are going to define the substituted instance $I'$ and $\vec{Y}'$ on which the \RecursiveSampling{} procedure is actually applied.
%

First, define 
    $N_k\triangleq\{j\mid v_j\in \RV{R} \land j>i \land p_{v_j} \neq \perp \land r_{v_j}\neq \perp\}$.
%
Recall the choice of parameter $(\epsilon_0,\gamma_0,\delta_0)$ in \eqref{eq:def-parameter-without-cd-expected-complexity}. 
For $j\in N_k$, let $X_{\beta_j}$ and $A_{\lambda_j}$ denote the respective random variable and bad event constructed in \Cref{lem:substituting} under parameter $\Lambda_j\triangleq B_{r_{v_j}}(p_{v_j})$, $Y_{\vbl(\Lambda_j)}$, $(\epsilon_0,\gamma_0,\delta_0)$, $I(B_{\ell_0(\epsilon_0,\gamma_0,\delta_0)+1}(\Lambda_j))$. 
Define 
\[
    U' \triangleq U\setminus \bigcup_{j\in N_k} \vbl(B_{\ell_0(\epsilon_0,\gamma_0,\delta_0)} (\Lambda_j))
    \quad\text{ and }\quad
    V' \triangleq V\setminus \bigcup_{j\in N_k} B_{\ell_0(\epsilon_0,\gamma_0,\delta_0)+1} (\Lambda_j).
\]
Then the substituted instance $I'$ is defined as follows:
\begin{align}
    I' &\triangleq \left(\{X_i\}_{i\in U'}\cup \{X_{\beta_j}\}_{j\in N_k}, \{A_v\}_{v\in V'} \cup \{A_{\lambda_j}\}_{j\in N_k} \right).\label{eq:I-substitute}
\end{align}
%
Let $\vec{Y}$ be the current assignment right before the \SLV{} algorithm starts at $v=v_k$.
For each $j\in N_k$, let $Y_{\beta_j}$ be drawn independently from the marginal distribution $\mu^{Y_{\vbl(\lambda_j) \setminus\{\beta_j\}}}_{I',\beta_j}$.
This is well-defined because $Y_{\vbl(\lambda_j) \setminus\{\beta_j\}}\subseteq U'$.
Then the substituted assignment $\vec{Y}'$ is the concatenation of $Y_{ U'}$ and $(Y_{\beta_j})_{j\in N_k}$, i.e.
\begin{align}
\vec{Y}'\triangleq Y_{ U'}\land (Y_{\beta_j})_{j\in N_k}.
    \label{eq:Y-substitute}
\end{align}
%

The \SLV{} algorithm at $v$ just calls \RecursiveSampling$(\vec{Y}'; I',B_{r_{v}}(p_{v}),\epsilon_0,\gamma_0,\delta_0,\gamma_0)$. 
Observe that although the definitions of $I'$ and $\vec{Y}'$ involve the global rank $k$ of the node $v$,
the actual constructions of $I'$ and $\vec{Y}'$ can be implicit during the recursion of \RecursiveSampling{},
so that the substituted parts of $I$ and $\vec{Y}$ are locally constructed when being accessed. 
This can be realized by local computation on the original $I$ and $\vec{Y}$, 
after extending the radius of the local algorithm by an additional $2(\mathcal{D}+|\RV{R}|\cdot \ell_0(\epsilon_0,\gamma_0,\delta_0)+1)$.
This implementation is formally explained in \Cref{sec:analysis-resampling}.

%
%

%

\begin{algorithm}
\caption{The \SLV{} algorithm for Resampling at node $v\in \RV{R}$}

\label{alg:resampling}
    


        
            

            

            
            
        

    
            \RecursiveSampling$(\vec{Y}'; I', B_{r_v}(p_v) ,\epsilon_0,\gamma_0,\delta_0,\gamma_0)$; \label{line:callrecursivesampling}

\tcp{\small The  $I'$ and $\vec{Y}'$  respectively defined in \eqref{eq:I-substitute}, \eqref{eq:Y-substitute} are implicitly given, where the substituted parts are realized at the time being accessed.}

            update $Y_{U'}\gets Y'_{U'}$; \label{line:resample-updating}

            \tcp{\small Only the updated variables need to be copied.}

\end{algorithm}

The following lemma states the correctness and efficiency of this algorithm and is proved in \Cref{sec:proof-resample-correctness-complexity}.

\begin{lemma}\label{lem:resample-correctness-complexity}
The followings hold after \Cref{alg:resampling} is sequentially executed on all node $v\in\RV{R}$, 
assuming that the input $\vec{Y}$ and $p_v,r_v$ for $v\in\RV{R}$ satisfy the properties asserted by~\Cref{lem:clustering-correctness}.
\begin{enumerate}
    \item \label{lem:item-resample-correctness}
    $\vec{Y}$ follows the distribution $\mu_I$.
    \item \label{lem:item-resample-complexity}
    For any $0<\eta<1$, with probability $1-\eta$, \Cref{alg:resampling} returns at every node $v\in \RV{R}$ within radius 
    $$\tO\left(|\RV{R}| \cdot \log^2 n \cdot\log^2 \frac{1}{\gamma}\cdot\log\frac{1}{\eta}\right)+\tO\left(\log^4 n \cdot \log^2\frac{1}{\gamma} \cdot\log^4 \frac{1}{\eta}\right).$$ 
\end{enumerate}
\end{lemma}


\subsection{Wrapping up (Proof of \Cref{thm: sample-gibbs})}\label{subsec:wrpping-up}
    
    First, we prove the correctness of sampling. 
    %
Let the LLL instance $I=(\{A_v\}_{v\in V},\{X_i\}_{i\in U})$ be  $\gamma$-satisfiable.
    By \Cref{lem:initialization-stage}, the \emph{Initialization} phase outputs $(\vec{Y},\RV{R})$ that satisfies \Cref{cond:initialization-stage}.
    Then, they are passed to a $2$-scan \SLOCAL{} algorithm on the set $\RV{R}$ of active nodes.
    %
    %
    Since $(\vec{Y}, \RV{R})$ satisfies \Cref{cond:initialization-stage}, by \Cref{lem:clustering-correctness}, after the 1st scan, 
    the assignment $\vec{Y}$, along with the values of $p_v\in V\cup \{\perp\}$ and $r_v\in \mathbb{N}\cup\{\perp\}$ for $v\in \RV{R}$ computed in this scan, 
    satisfy the the condition of \Cref{lem:resample-correctness-complexity}.
    Then by \Cref{lem:resample-correctness-complexity}, after the 2nd scan, the \SLV{} algorithm terminates and computes a random assignment $\vec{Y}\sim \mu_I$.
    By \Cref{thm: SLV-local to local}, the \SLV{} algorithm is faithfully simulated by a \LOCAL{} algorithm.

    Then, we bound the round complexity of the algorithm. 
    By \Cref{lem:initialization-stage}, the \emph{Initialization} phase takes fixed $\tO(\log^3 n)$ rounds in the \LOCAL{} model.
    By \Cref{thm: SLV-local to local}, the round complexity of the \LOCAL{} algorithm that simulates a \SLV{} algorithm is the product of the number of active nodes and the maximum radius of the \SLOCAL{} algorithm.
    For $i\in\{1,2\}$, let $R_i$ be the random variable that represents the maximum radius of the $i$-th scan of the  \SLV{} algorithm. 
    %
    %
    %
    %
    Let $\epsilon\in(0,1)$ be arbitrary.
    By \Cref{lem:initialization-stage}, we have 
    $|\RV{R}| =\tO\left( \log n   \cdot  \log \frac{1}{\gamma} \cdot \log \frac{1}{\epsilon} \right)$ 
    with probability at least $1-\frac{\epsilon}{3}$;
    by \Cref{lem:clustering-complexity-2}, we have 
    $R_1 =\tO\left( |\RV{R}|\cdot  \log^2 n \cdot \log^2 \frac{1}{\gamma} \cdot \log \frac{1}{\epsilon}\right)$
    with probability at least $1-\frac{\epsilon}{3}$;
    and by \Cref{lem:resample-correctness-complexity}, we have 
    $R_2 =\tO\left( |\RV{R}| \cdot  \log^4 n \cdot \log^2 \frac{1}{\gamma} \cdot \log^4\frac{1}{\epsilon} \right)$
    with probability at least $1-\frac{\epsilon}{3}$.
%
%
%
    Altogether, by union bound, with probability at least $1-\epsilon$, the round complexity of the \LOCAL{} algorithm 
    is bounded by
    \[|\RV{R}|\cdot \max(R_1, R_2) = \tO\left(\log^6 n\cdot\log^4 \frac{1}{\gamma}\cdot \log^6 \frac{1}{\epsilon}\right).\qedhere
    \]


\section{Related Work and Discussions}

The perfect simulation of Las Vegas algorithms is a fundamental problem.
In the celebrated work of Luby, Sinclaire, and Zuckerman~\cite{luby1993optimal}, an optimal strategy was given for speeding up Las Vegas algorithms.
Their approach was based on stochastic resetting, which requires global coordination and works for the Las Vegas algorithms with deterministic outputs,
or the interruptible random outputs.

The distribution of satisfying solutions of the Lov\'{a}sz local lemma (LLL) has drawn much attention, e.g.~in~\cite{GJL19,harris2020new}. 
Its perfect simulation was studied in~\cite{GJL19,he2022sampling,he2021perfect,feng2022towards},
where several key approaches for perfect sampling were applied, including partial rejection sampling (PRS) \cite{GJL19}, ``lazy depth-first'' method of Anand and Jerrum (\emph{a.k.a.}~the AJ algorithm)~\cite{anand2022perfect}, coupling from the past (CFTP)~\cite{PW96}, and coupling towards the past (CTTP)~\cite{feng2022towards}.

In the \LOCAL{} model,
Ghaffari, Harris and Kuhn~\cite{ghaffari2018derandomizing} showed that for distributed graph problems, 
where the goal is to construct feasible graph configurations, 
the fixed-round Las Vegas algorithms and the zero-error Las Vegas algorithms are equivalent to polylogarithmic rounds.
Their approach was based on a derandomization by conditional expectations, 
and hence was especially suitable for the tasks where the support of the output distribution, instead of the output distribution itself, is concerned,
such as the searching problems for constructing feasible solutions.


The \LOCAL{} algorithms and Gibbs distributions are intrinsically related.
For example, the distributions of the random bits on which a fixed-round Las Vegas  \LOCAL{} algorithm successfully returns are Gibbs distributions,
where the certifiers of local failures are the local constraints defining the Gibbs distribution.
In~\cite{feng2018local}, Feng and Yin gave a \LOCAL{} sampler with local failures for the Gibbs distributions with strong spatial mixing by
parallelizing the JVV sampler~\cite{jerrum1986random} using the network decomposition~\cite{linial1993low}.

\paragraph{Discussion of the current result.}
In this paper, we show that for local computation, 
the successful output of any fixed-round Las Vegas computation, where failures are reported locally, 
can be perfectly simulated with polylogarithmic overheads.

As by-products, this  gives perfect simulations, via efficient (polylogarithmic-round) local computation, 
for several fundamental classes of high-dimensional joint distributions, including: 
\begin{itemize}
    \item random satisfying solutions of Lov\'{a}sz local lemma with non-negligible satisfiability;
    \item uniform locally checkable labelings (LCLs) with non-negligible feasibility;
    \item Gibbs distributions satisfying the strong spatial mixing with exponential decay.
\end{itemize}

We develop a novel approach for augmenting Lov\'{a}sz local lemma (LLL) instances  
by introducing locally-defined new bad events, to create the desirable decay of correlation.
We also give a recursive local sampling procedure,
which utilizes the correlation decay to accelerate the sampling process, 
and meanwhile still keeps the sampling result correct, without being biased by the change to the LLL instance.

At first glance, 
it almost looks like we are creating mixing conditions out of nothing.
In particular, the approach seems to bypass the local-lemma-type conditions for sampling (e.g.~the one assumed in~\cite{he2022sampling}).
But indeed, our augmentation of LLL instances relies on that the LLL instances are fairly satisfiable 
(or at least the separator between the regions that we want to de-correlate should be enough satisfiable).
Sampling in such instances might already be tractable in conventional computation models, e.g.~in polynomial-time Turing machine, 
but the problem remains highly nontrivial for local computation.

This new approach for perfect simulation works especially well in the models where the locality is the sole concern.
A fundamental question is how this could be extended to the models where the computation and/or communication costs are also concerned, e.g.~$\mathsf{CONGEST}$ model or $\mathsf{PRAM}$ model.

\section{Analysis of Correlation Decay}\label{sec:analysis-correlation-decay}
In this section, we prove \Cref{lem:augmenting-informal}, the LLL augmentation lemma.
Recall the notions of balls  $B_\cdot(\cdot)$, shells $S_{[\cdot,\cdot]}(\cdot)$, and rings $\ring{}{[\cdot,\cdot]}{\cdot}$ defined in \Cref{sec:graph-notations}.
The lemma is formally restated as follows.

\begin{lemma}[formal restatement of \Cref{lem:augmenting-informal}]\label{lem:augmenting}
There is a universal constant $C_0>0$ such that the followings hold for any $\epsilon,\gamma\in(0,1)$, $\delta\in\left(0,\frac{\gamma}{2}\right)$, and for all $\ell\geq \ell_0(\epsilon, \gamma, \delta)$, where 
\begin{align}\label{eq:def-ell-0}
\ell_0(\epsilon,\gamma,\delta) 
\triangleq 
\left\lceil C_0 \cdot \log\frac{2}{\epsilon}\cdot\log\frac{2}{\gamma} \cdot \log\frac{1}{\delta} \cdot \log\left(2\log\frac{2}{\epsilon} \cdot \log\frac{2}{\gamma} \cdot\log\frac{1}{\delta}\right)\right\rceil.
\end{align}
%
Let  $I=\left(\{X_i\}_{i\in U},\{A_v\}_{v\in V}\right)$ be a LLL instance whose dependency graph is $D_I$.
For any nonempty $\Lambda\subseteq V$, 
if the sub-instance $I\left(S_{[1,\ell]}(\Lambda)\right)$ is $\gamma$-satisfiable,
then there exists a new bad event $A_{\lambda}$ with $\lambda\not\in V$ such that:
\begin{enumerate}
    \item \label{item:a-lamba-property} $A_{\lambda}$ is an event defined on the random variables in $\vbl(\lambda)=\ring{I}{[1,\ell]}{\Lambda}$, 
    and the construction of $A_{\lambda}$ depends only on the specifications of $\Lambda$, $\{A_v\}_{v\in B_{\ell+1}(\Lambda)}$, $\{X_i\}_{i\in \vbl(B_{\ell+1}(\Lambda))}$ and $(\epsilon, \gamma, \delta)$;
    \item \label{item:a-lambda-error-rate} $A_{\lambda}$ occurs with probability at most $\delta$ on independent random variables $\{X_i\}_{i\in U}$, i.e.~$\nu(A_{\lambda})\le \delta$;
    \item \label{item:augmented-cd} the variable sets $S=\ring{I}{0}{\Lambda}=\vbl(\Lambda)$ and $T=\ring{I}{[\ell+1,\infty]}{\Lambda}=U \setminus \vbl(B_{\ell}(\Lambda))$ are $\epsilon$-correlated in the augmented LLL  instance 
    \[\hat{I}=\left(\{X_i\}_{i\in U},\{A_v\}_{v\in V}\cup \{A_{\lambda}\}\right).\] 
\end{enumerate}
All balls $B_\cdot(\cdot)$ and shells $S_{[\cdot,\cdot]}(\cdot)$ in above are defined in the dependency graph $D_I$ of the LLL instance $I$.
\end{lemma}

\Cref{lem:augmenting} is crucial for proving our main technical result (\Cref{thm: sample-gibbs}).
The lemma basically says that in any LLL instance,
if two regions $S$ and $T$ are sufficiently distant apart and their separator is easy to satisfy, 
then a new rare bad event $A_\lambda$ can always be locally constructed depending only on the separator,
so that $S$ and $T$ are sufficiently de-correlated in the new LLL instance after being augmented with $A_\lambda$.

\subsection{Proof outline for the LLL augmentation lemma}\label{subsec:partial-correlated-rings}
The proof of \Cref{lem:augmenting} is fairly technical. Unlike most existing correlation decay results, the proof of \Cref{lem:augmenting} does not use potential analyses or coupling arguments.
Instead, our proof of \Cref{lem:augmenting} is more combinatorial and constructive, which is distinct from traditional analyses of correlation decays. 

Fix a subset $\Lambda\subseteq V$ of bad events in the LLL instance  $I=\left(\{X_i\}_{i\in U},\{A_v\}_{v\in V}\right)$.
Recall the rings ${\ring{}{r}{\Lambda}}$ and ${\ring{}{[i,j]}{\Lambda}}$ of variables defined in~\eqref{eq:def-variable-ring},
and rings $\eventaround{}{i}{j}{\Lambda}$ and $\eventin{}{j}{j}{\Lambda}$ of events defined in \eqref{eq:def-event-ring}.
\Cref{lem:augmenting} establishes a correlation decay between $S=\ring{I}{0}{\Lambda}$ and $T=\ring{I}{[\ell+1,\infty]}{\Lambda}$. 
%
%
Fix any $0\le i<j$, and boundary conditions $\sigma\in \Sigma_{\ring{}{i}{\Lambda}}$ and $\tau \in \Sigma_{\ring{}{j}{\Lambda}}$ respectively specified on ${\ring{}{i}{\Lambda}}$ and ${\ring{}{j}{\Lambda}}$.
Define:
\begin{align}
\psat{i}{\sigma}{j}{\tau}
&=\partialsat{\Lambda}{i}{\sigma}{j}{\tau}{I}\notag\\
&\triangleq \Pr_{\vec{X}\sim\nu}\left( \vec{X} \text{ avoids $A_v$ for all }v\in \eventaround{I}{i+1}{j-1}{\Lambda} \mid  X_{\ring{I}{i}{\Lambda}}=\sigma \wedge X_{\ring{I}{j}{\Lambda}}=\tau \right).   \label{eq:def-partialsat}
\end{align}
%
Intuitively, $\partialsat{\Lambda}{i}{\sigma}{j}{\tau}{I}$ represents the marginal probability for avoiding all the bad events intersected by the region sandwiched between the rings $\ring{}{i}{\Lambda}$ and $\ring{}{j}{\Lambda}$, given the boundary conditions $\sigma$ and $\tau$ respectively specified on these two rings.

Using this quantity, we define the following variant of the correlation decay property for rings.

\begin{definition}[partial $\epsilon$-correlated rings]\label{def:partial-correlated}
Fix any nonempty subset $\Lambda\subseteq V$ of bad events, and fix any $0\leq i<j$. 
    %
    The two rings $\ring{}{i}{\Lambda}$ and $\ring{}{j}{\Lambda}$ are said to be \concept{partiall $\epsilon$-correlated} 
    if for any $\sigma_1,\sigma_2\in \ring{}{i}{\Lambda}$ that avoid all bad events in $\eventin{}{i}{i}{\Lambda}$ and any $\tau_1,\tau_2\in \ring{}{j}{\Lambda}$ that avoid all bad events in $\eventin{}{j}{j}{\Lambda}$,
    \[
    \psat{i}{\sigma_1}{j}{\tau_1}\cdot\psat{i}{\sigma_2}{j}{\tau_2}\le (1+\epsilon)\cdot \psat{i}{\sigma_1}{j}{\tau_2}\cdot\psat{i}{\sigma_2}{j}{\tau_1}.
    \]
\end{definition}

\begin{remark}
This notion of partial $\epsilon$-correlated rings can be seen as a local variant of the $\epsilon$-correlated sets in \Cref{def:correlation}
since it only considers the bad events sandwiched between the two rings. 
\end{remark}

The following states a monotonicity property for the partial $\epsilon$-correlated rings.
\begin{lemma}[monotonicity of partial $\epsilon$-correlated rings]\label{lem:maintenance}
For any $0\leq i'\le i<j\le j'$, any $\epsilon>0$, if $\ring{}{i}{\Lambda}$ and $\ring{}{j}{\Lambda}$ are partial $\epsilon$-correlated, then $\ring{}{i'}{\Lambda}$ and $\ring{}{j'}{\Lambda}$ are partial $\epsilon$-correlated.
\end{lemma}

Meanwhile, the notion of partial $\epsilon$-correlated rings is strong enough to imply the $\epsilon$-correlated sets.
%
%
\begin{lemma}[$\epsilon$-correlated regions from partial $\epsilon$-correlated rings]\label{lem:partial-cd-to-cd}
    For $1\leq i<j$, if $\ring{}{i}{\Lambda}$ and $\ring{}{j}{\Lambda}$ are partial $\epsilon$-correlated, 
    then $\ring{}{[0,i-1]}{\Lambda}$ and $\ring{}{[j+1,\infty]}{\Lambda}$ are $\epsilon$-correlated. 
\end{lemma}

The proofs of \Cref{lem:maintenance} and \Cref{lem:partial-cd-to-cd} are due to routine verification of the definitions of partial $\epsilon$-correlated rings, which are deferred to \Cref{sec:partial-cd-proofs}.

Our goal is then to ensure that the rings  $\ring{}{1}{\Lambda}$ and $\ring{}{\ell}{\Lambda}$ are partial $\epsilon$-correlated.
By \Cref{lem:partial-cd-to-cd}, this would directly imply that  $S=\ring{}{0}{\Lambda}$ and $T=\ring{}{[\ell+1,\infty]}{\Lambda}$ are $\epsilon$-correlated.
However, such partial $\epsilon$-correlated rings $\ring{}{1}{\Lambda}$ and $\ring{}{\ell}{\Lambda}$ may not always hold in a LLL instance with $\gamma$-satisfiable $I\left(S_{[1,\ell](\Lambda)}\right)$.

To see this, consider the following counterexample.
\begin{example}[strong correlation enforced by rare bad assignments]
In this LLL instance, suppose there exists a sequence of assignments $\sigma_1\in \Sigma_{\ring{}{1}{\Lambda}}$, $\dots$, $\sigma_{\ell}\in \Sigma_{\ring{}{\ell}{\Lambda}}$ such that for any $1\leq i< \ell$, if $\ring{}{i}{\Lambda}$ is assigned $\sigma_i$, then $\ring{}{i+1}{\Lambda}$ must be assigned $\sigma_{i+1}$. This condition can be enforced by defining a bad event on each $\ring{}{i}{\Lambda}\cup\ring{}{i+1}{\Lambda}$, which occurs if and only if $\ring{}{i}{\Lambda}$ is assigned $\sigma_i$ but $\ring{}{i+1}{\Lambda}$ is not assigned $\sigma_{i+1}$. Apart from these bad events, the LLL instance does not contain any other bad events.\footnote{Note that such a LLL instance with rings $\ring{}{i}{\Lambda}$ is well-defined: For the bad events $\{A_{v_i}\}_{1\le i< \ell}$, it holds that $\vbl(v_i)=\ring{}{i}{\Lambda}\cup\ring{}{i+1}{\Lambda}$ and $\vbl(v_i)\cap\vbl(v_{i+1})=\ring{}{i}{\Lambda}$. It can be verified that the dependency graph gives the rings $\ring{}{i}{\Lambda}$.}
With sufficiently small $\nu_{\ring{}{i}{\Lambda}}(\sigma_i)$ for $1\leq i\leq \ell-1$, the entire LLL instance can be made sufficiently $\gamma$-satisfiable. This is because a bad event occurs only if a specific $\sigma_i$ is assigned to some ring $\ring{}{i}{\Lambda}$. However, $\ring{}{1}{\Lambda}$ and $\ring{}{\ell}{\Lambda}$ are strongly correlated: if $\ring{}{1}{\Lambda}$ is assigned $\sigma_1$, then $\ring{}{\ell}{\Lambda}$ is constrained to be assigned~$\sigma_{\ell}$.

Nevertheless, the marginal probability of $\ring{}{1}{\Lambda}$ being assigned $\sigma_1$ is very small, as it enforces $\ring{}{i}{\Lambda}$ to be assigned $\sigma_i$ for all $1<i\leq \ell$. Consequently, by introducing a new bad event that prevents $\ring{}{1}{\Lambda}$ from being assigned $\sigma_1$, which occurs with small probability, one can make $\ring{}{1}{\Lambda}$ and $\ring{}{\ell}{\Lambda}$ decorrelated.
\end{example}

Remarkably, the strategy used to address the extreme counterexample described above is also effective in general.
For a general LLL instance, we set a reasonably large threshold $D$, dependent on the radius $\ell$ and the probability $\delta$. 
It appears that, for any $1\leq i,j\leq \ell$ with $|i-j|>D$,
those `bad' assignments on $\ring{}{i}{\Lambda}$ which may induce significant correlation between $\ring{}{i}{\Lambda}$ and $\ring{}{j}{\Lambda}$ must have small marginal probabilities.
We further prove that, by additionally avoiding a new bad event $A_{\lambda}$ which includes all such `bad' assignments with small marginal probabilities, the correlation between $\ring{}{1}{\Lambda}$ and $\ring{}{\ell}{\Lambda}$ can be substantially reduced.

\subsection{Construction of the local rare bad event $A_{\lambda}$} \label{sec:bad-event-construction}
We now describe the construction of the new bad event $A_{\lambda}$.
Fix
$\epsilon,\gamma\in(0,1)$, $\delta\in\left(0,\frac{\gamma}{2}\right)$, $\ell\geq \ell_0(\epsilon, \gamma, \delta)$,
and a nonempty subset $\Lambda\subseteq V$ of bad events.
Suppose that the sub-instance $I(\ring{}{[1,\ell]}{\Lambda})$ is $\gamma$-satisfiable.
Let $\varepsilon_0$ be a sufficiently small constant to be fixed later, and define $D\triangleq \frac{1}{\varepsilon_0}\log\frac{\ell}{\delta}$.

The new bad event $A_\lambda$ claimed in \Cref{lem:augmenting} is constructed by the procedure described in \Cref{alg:augmenting}.

\begin{algorithm}
\caption{Construction of the bad event $A_{\lambda}$.}
    \label{alg:augmenting}
    \For{$1\le i\le \ell$}{    
    initialize $A_{\lambda_i}$ $\gets$ the trivial event defined on the variables in $\vbl(\lambda_i)\triangleq \ring{I}{i}{\Lambda}$ that never occurs\;}
    define $J\triangleq \left(\{X_i\}_{i\in U},\{A_v\}_{v\in V}\cup\{A_{\lambda_i}\}_{1\le i\le \ell}\right)$ and $D\triangleq \frac{1}{\varepsilon_0}\log\frac{\ell}{\delta}$\;
\Repeat{nothing has changed to $\{A_{\lambda_i}\}_{1\leq i\leq \ell}$}{
\If{there exist $1\leq i<j\leq \ell$ with $|i-j|>D$ and $\sigma\in{\Sigma_{\vbl(\lambda_i)}}$ avoiding $A_{\lambda_i}$  s.t. \hspace{200pt}
$\mathop{\mathbb{E}}_{\tau}\left[\partialsat{\Lambda}{i}{\sigma}{j}{\tau}{J} \right]<\frac{\delta}{2\ell}$ where $\tau\sim \nu_{\vbl(\lambda_j)}$
}{
 update the bad event $A_{\lambda_i}$ to further include $\sigma$; 
}
\If{there exist $1\leq i<j\leq \ell$ with $|i-j|>D$  and $\tau\in{\Sigma_{\vbl(\lambda_j)}}$ avoiding $A_{\lambda_j}$ s.t. \hspace{200pt}
 and $\mathop{\mathbb{E}}_{\sigma}\left[\partialsat{\Lambda}{i}{\sigma}{j}{\tau}{J} \right]<\frac{\delta}{2\ell}$ where $\sigma\sim \nu_{\vbl(\lambda_i)}$
}{
        update the bad event $A_{\lambda_j}$ to further include $\tau$; 
}
}    
let $A_\lambda$ be the event defined on the variables in $\vbl(\lambda)\triangleq \vbl(\lambda_1)\uplus\vbl(\lambda_2)\uplus\cdots\uplus\vbl(\lambda_{\ell})$ such that
$$A_\lambda =\left(\bigcap_{v\in \eventin{I}{1}{\ell}{\Lambda}}\overline{A_v}\right)\cap \left(\bigcup_{1\le i\le \ell}A_{\lambda_i}\right);$$
\end{algorithm}


%
%

\begin{remark}[consistency of the rings]
   Note that \Cref{alg:augmenting} iteratively updates the sequence of bad events  $\{A_{\lambda_i}\}_{1\leq i\leq \ell}$ specified respectively on the rings $\vbl(\lambda_i)\triangleq \ring{I}{i}{\Lambda}$.
   It is easy to see that throughout the process, $\vbl(\lambda_i)$ remains the same and hence $\ring{J}{i}{\Lambda}=\ring{I}{i}{\Lambda}$, where $J$ denotes the updated LLL instance.
\end{remark}

\begin{remark}[locality of \Cref{alg:augmenting}]\label{remark:locality-augment}
It is easy to see that the construction of the bad event $A_{\lambda}$ depends only on $\Lambda$, $\{A_v\}_{v\in B_{\ell+1}(\Lambda)}$, $\{X_i\}_{i\in \vbl(B_{\ell+1}(\Lambda))}$ and $(\epsilon, \gamma, \delta)$.
\Cref{item:a-lamba-property} of \Cref{lem:augmenting} is satisfied.
\end{remark}

Furthermore, the bad event $A_{\lambda}$ constructed in \Cref{alg:augmenting} occurs with a small probability. 
This is because all the bad assignments included by $A_{\lambda}$ have small marginal probabilities.
\Cref{item:a-lambda-error-rate} of \Cref{lem:augmenting} then follows. 

\begin{lemma}[rarity of the new bad event] \label{lem:A-lambda-error-rate}
    $A_{\lambda}$ occurs with probability at most $\delta$.
\end{lemma}
\begin{proof}
    %
    For $k\ge 0$, let $J^{(k)}$ denote the LLL instance $J$ after $k$ iterations of the \textbf{repeat} loop in~\Cref{alg:augmenting}. 
    Let $\left\{A_{\lambda_i}^{(k)}\right\}_{1\leq i\leq \ell}$ denote the $\{A_{\lambda_i}\}_{1\leq i\leq \ell}$ after $k$ iterations. 
    Further suppose that in the $k$-th iteration, the bad event $A_{\lambda_{i_k}}$ is picked to update 
    and is made  occur on the $\tau^{(k)}\in \Sigma_{\vbl(\lambda_{i_k})}$.
%

    For $\vec{X}$ drawn from the product distribution $\nu$, the probability of $A_{\lambda}$ can be calculated as:
\begin{align*}
    \Pr_{\vec{X}\sim\nu}(A_{\lambda})
    &=\Pr_{\vec{X}\sim \nu}\left(\left(\bigcap_{v\in \eventin{I}{1}{\ell}{\Lambda}}\overline{A_v}\right)\cap \left(\bigcup_{1\le i\le \ell}A_{\lambda_i}\right)\right) \\
    &= \Pr_{\vec{X}\sim \nu}\left(  \left(\bigcap_{v\in \eventin{I}{1}{\ell}{\Lambda}}\overline{A_v}\right) \cap \left(\exists k\geq 1:   X_{\ring{J}{i_k}{\Lambda}}= \tau^{(k)}\right)  \right) \\
    &=\sum_{k\ge 1} \Pr_{\vec{X}\sim \nu}\left(\left(\bigcap_{v\in \eventin{I}{1}{\ell}{\Lambda}}\overline{A_v} \right)
        \cap \left(\forall 1\le j<k: X_{\ring{J}{i_j}{\Lambda}}\neq \tau^{(j)}\right)
        \cap \left(X_{\ring{J}{i_k}{\Lambda}}= \tau^{(k)}\right)\right)\\
    &=\sum_{k\ge 1} \Pr_{\vec{X}\sim \nu}\left(\left(\bigcap_{v\in \eventin{I}{1}{\ell}{\Lambda}}\overline{A_v} \right)
        \cap \left(\bigcap_{1\leq i\leq \ell} \overline{A_{\lambda_i}^{(k-1)}}\right)
        \cap \left( X_{\ring{J}{i_k}{\Lambda}}= \tau^{(k)}\right)\right)
\end{align*}
Due to the definition of $J^{(k)}$, we have
\[
\left(\bigcap_{v\in \eventin{I}{1}{\ell}{\Lambda}}\overline{A_v} \right)
        \cap \left(\bigcap_{1\leq i\leq \ell} \overline{A_{\lambda_i}^{(k-1)}}\right)
        = \bigcap_{v\in \eventin{J^{(k-1)}}{1}{\ell}{\Lambda}}\overline{A_v}.
\]
And in \Cref{alg:augmenting}, for each iteration $k\geq 1$, there exists a $1\leq j_k\leq \ell$ with $|i_k-j_k|>D$ such that
\begin{itemize}
    \item if $i_k<j_k$, then $\partialsat{\Lambda}{i_k}{\tau^{(k)}}{j_k}{*}{J^{(k-1)}} 
    = \Pr_{\vec{X}\sim \nu}\left(\bigcap_{v\in \eventaround{J^{(k-1)}}{i_k+1}{j_k-1}{\Lambda}}\overline{A_v} ~\middle|~ X_{\ring{J}{i_k}{\Lambda}}= \tau^{(k)}\right) 
    <\frac{\delta}{2\ell}$;
    \item if $i_k>j_k$, then $\partialsat{\Lambda}{j_k}{*}{i_k}{\tau^{(k)}}{J^{(k-1)}}
    = \Pr_{\vec{X}\sim \nu}\left(\bigcap_{v\in \eventaround{J^{(k-1)}}{j_k+1}{i_k-1}{\Lambda}}\overline{A_v} ~\middle|~ X_{\ring{J}{i_k}{\Lambda}}= \tau^{(k)}\right)<\frac{\delta}{2\ell}$.
\end{itemize}
Thus, we have
\begin{align*}
    \Pr_{\vec{X}\sim\nu}(A_{\lambda}) =&\sum_{k\ge 1} \Pr_{\vec{X}\sim \nu}\left(\left(\bigcap_{v\in \eventin{J^{(k-1)}}{1}{\ell}{\Lambda}}\overline{A_v} \right) 
        \cap \left(X_{\ring{J}{i_k}{\Lambda}}= \tau^{(k)}\right) \right)\\
    =&\sum_{k\ge 1} \Pr_{\vec{X}\sim \nu}\left(\bigcap_{v\in \eventin{J^{(k-1)}}{1}{\ell}{\Lambda}}\overline{A_v} ~\middle|~ X_{\ring{J}{i_k}{\Lambda}}= \tau^{(k)}\right) \cdot \nu_{\ring{J}{i_k}{\Lambda}}(\tau^{(k)}) \\
    \leq& \sum_{k\ge 1} \Pr_{\vec{X}\sim \nu}\left(\bigcap_{v\in \eventaround{J^{(k-1)}}{\min(i_k,j_k)+1}{\max(i_k,j_k)-1}{\Lambda}}\overline{A_v} ~\middle|~ X_{\ring{J}{i_k}{\Lambda}}= \tau^{(k)}\right) \cdot \nu_{\ring{J}{i_k}{\Lambda}}(\tau^{(k)}) \\
    \leq& \sum_{k\ge 1} \frac{\delta}{2\ell} \cdot \nu_{\ring{J}{i_k}{\Lambda}}(\tau^{(k)})
    \leq \frac{\delta}{2\ell}\sum_{1\leq i\leq \ell} \sum_{\sigma\in \Sigma_{\ring{J}{i}{\Lambda}}} \nu_{\ring{J}{i}{\Lambda}} (\sigma)
    \leq \delta.\qedhere
\end{align*}
\end{proof}



{

\subsection{Decay of correlation in augmented LLL}



It remains to prove \Cref{item:augmented-cd} of \Cref{lem:augmenting}, the decay of correlation between $S\triangleq\ring{I}{0}{\Lambda}$ and $T\triangleq\ring{I}{[\ell+1,\infty]}{\Lambda}$ in the augmented LLL instance $\hat{I}\triangleq \left(\{X_i\}_{i\in U}, \{A_v\}_{v\in V} \cup \{A_{\lambda}\}\right)$. 
%


Let $J^*$ denote the LLL instance $J$ when \Cref{alg:augmenting} terminates.  
%
%
%
The following is easy to verify.
\begin{proposition}
For any $\sigma\in \Sigma_S$ and $\tau\in \Sigma_T$, it holds that $\nu(\Omega_{J^*}^{\sigma\land\tau}) = \nu(\Omega_{\hat{I}}^{\sigma\land \tau})$.
\end{proposition}
Thus, to guarantee that $S$ and $T$ are $\epsilon$-correlated in $\hat{I}$, it suffices to prove that they are $\epsilon$-correlated in $J^*$.
%
%
By \Cref{lem:partial-cd-to-cd}, it suffices to prove that $\ring{}{1}{\Lambda}$ and $\ring{}{\ell}{\Lambda}$ are partial $\epsilon$-correlated in $J^*$.

In the following, we omit the LLL instance $J^*$ and the subset $\Lambda$ of variables in our notations and write:
\begin{align*}
\psat{i}{\sigma}{j}{\tau}
\triangleq\partialsat{\Lambda}{i}{\sigma}{j}{\tau}{J^*}, \quad
\ri{r}
\triangleq \ring{J^*}{r}{\Lambda}, \quad
\ein{i}{j}
\triangleq \eventin{J^*}{i}{j}{\Lambda} \quad\text{and}\quad
\earound{i}{j}
\triangleq\eventaround{J^*}{i}{j}{\Lambda}.
\end{align*}
And when we say $\epsilon$-correlated, it always means the $\epsilon$-correlated in $J^*$. 
Furthermore, we use  $V^*$ to denote the set of the bad events in the LLL instance $J^*$, that is, $V^*\triangleq V\cup\{\lambda_i\mid 1\leq i\leq \ell\}$.  

To establish a partial correlation decay between  $\ri{1}$ and $\ri{\ell}$, we consider the \concept{good} rings, 
the bad events intersected by which happen with small probability. Here, $\varepsilon_0$ is a constant to be fixed later.
\begin{definition}
\label{def:good-rings}
    For $i\geq 0$, the $\ri{i}$ is called a \concept{good ring}, if $\bigcup_{v\in \earound{i}{i}} A_v$ occurs with probability at most $\varepsilon_0$.
\end{definition}

First, we observe that a local cluster of good rings may create certain partial correlation decays. 
\begin{lemma}\label{lem:initialization}
    For $1\leq i\leq \ell-3D$, if $\ri{i}$, $\ri{i+1}$, $\dots$, $\ri{i+3D}$ are good rings, then $\ri{i}$ and $\ri{i+3D}$ are partial $\frac{1}{\varepsilon_0}$-correlated. 
\end{lemma}

Next, we observe that consecutive good rings may enhance the partial correlation decay.

\begin{lemma}\label{lem:iteration}
For any $1< i<i+2 D<j\leq \ell-2D$, any $\epsilon>0$,  if $\ri{i}$ and $\ri{j}$ are partial $\epsilon$-correlated and $\ri{i},\ri{i+1},\ldots,\ri{j+2D}$ are \gr rings, then $\ri{i}$ and $\ri{j+2D}$ are partial \co{$(1-\varepsilon_0)\epsilon$}.
\end{lemma}

Using \Cref{lem:initialization,lem:iteration}, we can prove \Cref{lem:augmenting}, the LLL augmentation lemma.

\begin{proof}[Proof of \Cref{lem:augmenting}]  

    \Cref{item:a-lamba-property} follows from \Cref{remark:locality-augment}.
    \Cref{item:a-lambda-error-rate} follows from \Cref{lem:A-lambda-error-rate}. 
    
    To prove \Cref{item:augmented-cd}, it is sufficient to prove the existence of a long enough sequence of consecutive good rings.
%
    %
    By the assumption, the sub-instance $I\left(S_{[1,\ell]}(\Lambda)\right)$ is $\gamma$-satisfiable;
    and by \Cref{lem:A-lambda-error-rate}, the new bad event $A_{\lambda}$ defined on $\ri{[1,\ell]}$ occurs with small probability $\delta<\frac{\gamma}{2}$.
    Therefore, the sub-instance $J^*(S_{[1,\ell]}(\Lambda))$ is $\frac{\gamma}{2}$-satisfiable.
    %
    %
    
    Let  $\ri{i_1}, \ri{i_2}, \ldots,\ri{i_k}$ be the longest possble sub-sequence of non-good rings satisfying that $|i_j-i_{j+1}|>2$ for all $1\le j<k$. 
    Observe that the corresponding sub-sequence of events $\bigcup_{v\in \earound{i_j}{i_j}} A_v$, where $1\le j\le k$, must be mutually independent since they cannot share any variables because $|i_j-i_{j+1}|>2$.
    %
    Due to such mutual independence, the satisfying probability of the sub-instance $J^*(S_{[1,\ell]}(\Lambda))$ is upper bounded by the product of the probabilities of the non-occurrence for $\bigcup_{v\in \earound{i_j}{i_j}} A_v$ over all $1\le j\le k$, which is further bounded by $(1-\epsilon_0)^k$, because each $\ri{i_j}$ is non-good ring.
    Recall that $J^*(S_{[1,\ell]}(\Lambda))$ is $\frac{\gamma}{2}$-satisfiable.
    Thus, the length of this sub-sequence is bounded as $k\le \log_{1-\epsilon_0} \frac{\gamma}{2}$. 
    Note that this sub-sequence $\ri{i_1}, \ri{i_2},...,\ri{i_k}$ of non-good rings partitions the $\ri{[1,\ell]}$ into $k+1$ parts, with the first part being $\ri{[1, i_1-1]}$, the last part being $\ri{[i_k,\ell]}$, and the $j$-th part being $\ri{[i_{j-1}, i_j-1]}$ for $1<j<k+1$. 
    And apparently, no part can contain more than 5 non-good rings. Otherwise, the sub-sequence of non-good rings $\ri{i_1}, \ri{i_2}, \ldots,\ri{i_k}$ can be made longer and still satisfying $|i_j-i_{j+1}|>2$, which contradicts to the assumption that it is the longest such sequence.
    Thus, there are at most $5 \cdot \left(\log_{1-\epsilon_o} \frac{\gamma}{2}+1\right)$ non-good rings among $\ri{1},\ri{2},\ldots,\ri{\ell}$.
    If 
    \begin{equation} \label{eq:continues-good-rings}
        \ell> 2\cdot \left(5\cdot \left(\log_{1-\epsilon_o} \frac{\gamma}{2}+1\right) \right) \cdot (3D+2D\cdot \lceil \log_{1-\varepsilon_0} (\varepsilon_0\epsilon)\rceil ),
    \end{equation}
    then there must exist a sequence of consecutive good rings $\ri{i},\ri{i+1},\ldots , \ri{j}$ such that $1\leq i \le j\leq \ell$ and $j-i\geq 3D+2D\cdot \lceil\log_{1-\varepsilon_0}(\varepsilon_0\epsilon)\rceil$.
    Recall that $D=\frac{1}{\varepsilon_0}\log \frac{\ell}{\delta}$. By setting $C_0$ to be sufficient large, all $\ell\geq \ell_0(\epsilon,\gamma,\delta)$ 
    satisfy~\eqref{eq:continues-good-rings}. 

    With this long enough sequence of consecutive good rings, we are able to prove \Cref{item:augmented-cd} as follows.
    By \Cref{lem:initialization}, $\ri{i}$ and $\ri{i+3D}$ are partial $\frac{1}{\varepsilon_0}$-correlated.
    By \Cref{lem:iteration}, for $1\leq k\leq \lceil \log_{1-\varepsilon_0} (\varepsilon_0\epsilon)\rceil$, $\ri{i}$ and $\ri{i+3D+k\cdot 2D}$ are partial $\frac{(1-\varepsilon_0)^k}{\varepsilon_0}$-correlated.  
    Thus, $R_i$ and $R_j$ are partial $\epsilon$-correlated. 
    By \Cref{lem:maintenance}, $\ri{1}$ and $\ri{\ell}$ are partial $\epsilon$-correlated. 
    By \Cref{lem:partial-cd-to-cd}, $S=R_0$ and $T=R_{[\ell+1,\infty]}$ are $\epsilon$-correlated. 
\end{proof}
}

\subsection{Correlation decay between good rings}
It remains to prove \Cref{lem:initialization,lem:iteration}, which shows that consecutive good rings may create and enhance partial correlation decays.

We extend the definition of probability $\psat{i}{\sigma}{j}{\tau}$ in \eqref{eq:def-partialsat} to the classes $C_i\subseteq \Sigma_{\ri{i}}$ and $C_j\subseteq \Sigma_{\ri{j}}$ of assignments on the rings $\ri{i}$ and $\ri{j}$:
 \begin{align*}
     \psat{i}{C_i}{j}{\tau}
     \triangleq \sum_{\sigma\in C_i} \nu_{\ri{i}}(\sigma) \cdot \psat{i}{\sigma}{j}{\tau}
     \quad\text{ and }\quad
     \psat{i}{\sigma}{j}{C_j}
     \triangleq \sum_{\tau\in C_j} \nu_{\ri{j}}(\tau) \cdot \psat{i}{\sigma}{j}{\tau}.
 \end{align*}
In particular, we observe that
\begin{align*}
    \psat{i}{\sigma}{j}{\Sigma_{\ri{j}}}
    =
    \mathop{\mathbb{E}}_{\tau\sim \nu_{\ri{j}}}\left[\psat{i}{\sigma}{j}{\tau} \right] 
    =\Pr_{\vec{X}\sim\nu}\left( \vec{X} \text{ avoids $A_v$ for all }v\in \earound{i+1}{j-1} \mid  X_{\ri{i}}=\sigma \right),
\end{align*}
and $\psat{i}{\Sigma_{\ri{i}}}{j}{\tau}=\mathop{\mathbb{E}}_{\sigma}\left[\psat{i}{\sigma}{j}{\tau} \right]$ is symmetrically satisfied.

\begin{definition}
Let $1\leq i<j \leq \ell$.
For any $\tau\in \Sigma_{\ri{j}}$, we say that $\ri{i}$ is \concept{well-distributed based on} $(j,\tau)$, 
if $\tau$ avoids $A_{\lambda_j}$ and for any $C\subseteq  \Sigma_{\ri{i}}$ satisfying $\nu_{\ri{i}}(C)>\frac{1}{4}$, 
we have $\frac{\psat{i}{C}{j}{\tau}}{\psat{i}{\Sigma_{\ri{i}}}{j}{\tau}}>\frac{\nu_{\ri{i}}(C)}{2}$.

And symmetrically, for any  $\sigma\in \Sigma_{\ri{i}}$ we say that $\ri{j}$ is \emph{well-distributed based on} $(i,\sigma)$ 
if $\sigma$ avoids $A_{\lambda_i}$ and for any $C\subseteq  \Sigma_{\ri{j}}$ satisfying $\nu_{\ri{j}}(C)>\frac{1}{4}$, 
we have $\frac{\psat{i}{\sigma}{j}{C}}{\psat{i}{\sigma}{j}{\Sigma_{\ri{j}}}}>\frac{\nu_{\ri{j}}(C)}{2}$.

\end{definition} 

The following lemma states the good properties for being well-distributed and holds generally. 

\begin{lemma}\label{lem:well-distributed-property}
Let $1<i<j\leq \ell$. For any $\tau\in \Sigma_{\ri{j}}$, if $\ri{i}$ is well-distributed based on $(j,\tau)$ and $\ri{i}$ is a good ring, then the followings hold:
\begin{enumerate}
\item $\ri{i-1}$ is well-distributed based on $(j,\tau)$;
\item for any $\sigma\in\Sigma_{\ri{i-1}}$: $\psat{i-1}{\sigma}{j}{\tau}\le 3\nu_{\ri{i-1}}(\sigma) \cdot \psat{i-1}{\Sigma_{\ri{i-1}}}{j}{\tau}$;
\item $\nu_{\ri{i}}(S)\geq \frac{3}{4}$ for $S=\left\{\pi\in \Sigma_{\ri{i}} ~\middle|~ \frac{\psat{i}{\pi}{j}{\tau}}{\psat{i}{\Sigma_{\ri{i}}}{j}{\tau}}\geq \frac{1}{8}\right\}$. 
\end{enumerate}
And the symmetric holds for the good ring $\ri{j}$ that is well-distributed based on $(i,\sigma)$ for any $\sigma\in \Sigma_{\ri{i}}$:
\begin{enumerate}
\item $\ri{j+1}$ is well-distributed based on $(i,\sigma)$;
\item for any $\tau\in\Sigma_{\ri{j+1}}$: $\psat{i}{\sigma}{j+1}{\tau}\le 3\nu_{\ri{j+1}}(\tau) \cdot \psat{i}{\sigma}{j+1}{\Sigma_{\ri{j+1}}}$;
\item $\nu_{\ri{j}}(S)\geq \frac{3}{4}$ for $S=\left\{\pi\in \Sigma_{\ri{j}} ~\middle|~ \frac{\psat{i}{\sigma}{j}{\pi}}{\psat{i}{\sigma}{j}{\Sigma_{\ri{j}}}}\geq \frac{1}{8}\right\}$. 
\end{enumerate}
\end{lemma}

\begin{proof}
    We prove these properties one by one.
\begin{enumerate}
    \item Note that for any $v\in \earound{i}{i} \setminus \earound{i+1}{i+1}$, we have $\vbl(v)\subseteq \ri{i-1}\cup \ri{i}$. 
    For any $\sigma\in \Sigma_{\ri{i-1}}$, define 
    \[
    K_{\sigma}\triangleq \left\{\rho \in \Sigma_{\ri{i}} \mid \sigma \land \rho \text{ avoids all bad events $A_v$ for  $v\in \earound{i}{i} \setminus \earound{i+1}{i+1}$}\right\}.\]
    For any $C\subseteq \Sigma_{\ri{i-1}}$ with $\nu_{\ri{i-1}}(C)>\frac{1}{4}$, we define $C'\triangleq \{\sigma\in C \mid \nu_{\ri{i}}(K_{\sigma}) >1-\frac{1}{9}\} $. 
    Then, for any $\sigma'\in C\setminus C'$, and $\vec{X}$ drawn from product distribution ${\nu}$, we have 
    \[ \Pr\left( \text{$\vec{X}$ avoids all bad events $A_v$ for $v\in  \earound{i}{i} \setminus \earound{i+1}{i+1}$} \mid X_{\ri{i-1}}=\sigma'\right)\leq \frac{8}{9}.\]
    Since $\ri{i}$ is a good ring, $\bigcup_{v\in \earound{i}{i} \setminus \earound{i+1}{i+1}} A_v$  occurs with probability at most $\varepsilon_0$. 
    We have $\nu_{\ri{i-1}}(C\setminus C')\le 9\varepsilon_0$, which means $\nu_{\ri{i-1}}(C') \geq \nu_{\ri{i-1}}(C)-9\varepsilon_0$.
    Since $\ri{i}$ is well-distributed based on $(j,\tau)$, for any $\sigma\in \Sigma_{\ri{i-1}}$ with $\nu_{\ri{i}}(K_\sigma)\ge \frac{8}{9}$, we have
    \[\psat{i}{K_{\sigma}}{j}{\tau} = \frac{\psat{i}{K_{\sigma}}{j}{\tau}}{\psat{i}{\Sigma_{\ri{i}}}{j}{\tau}} \cdot \psat{i}{\Sigma_{\ri{i}}}{j}{\tau} > \frac{\nu_{\ri{i}}(K_{\sigma})}{2}\cdot \psat{i}{\Sigma_{\ri{i}}}{j}{\tau} \geq  \frac{4}{9} \cdot  \psat{i}{\Sigma_{\ri{i}}}{j}{\tau}.\]
    Thus, for any $C\subseteq \Sigma_{\ri{i-1}}$ with $\nu_{\ri{i-1}}(C)>\frac{1}{4}$, we have 
    \[\psat{i-1}{C}{j}{\tau} \geq \sum_{\sigma\in C'}\nu_{\ri{i-1}}(\sigma)\psat{i}{K_\sigma}{j}{\tau}  > \frac{4}{9}\left(\nu_{\ri{i-1}}(S)-9\varepsilon_0\right)\psat{i}{\Sigma_{\ri{i}}}{j}{\tau}.\]
    Note that $\psat{i-1}{\Sigma_{R_{i-1}}\setminus C}{j}{\tau}\le (1-\nu_{\ri{i-1}}(C))\cdot \psat{i}{\Sigma_{\ri{i}}}{j}{\tau}=(1-\nu_{\ri{i-1}}(C))\cdot \psat{i}{\Sigma_{\ri{i}}}{j}{\tau}$, hence
    \[\frac{\psat{i-1}{C}{j}{\tau}}{\psat{i-1}{\Sigma_{\ri{i-1}}\setminus C}{j}{\tau}} > \frac{4}{9}\cdot \frac{\nu_{\ri{i-1}}(C)-9\varepsilon_0}{1-\nu_{\ri{i-1}}(C)} \geq \frac{\nu_{\ri{i-1}}(C)/2}{1-\nu_{\ri{i-1}}(C) /2}. \]
    The last inequality holds for $\nu_{\ri{i-1}}(C) \geq \frac{1}{4}$ and sufficiently small $\varepsilon_0$. Thus, we obtain that
    $$\frac{\psat{i-1}{C}{j}{\tau}}{\psat{i-1}{*}{j}{\tau}}=\frac{\psat{i-1}{C}{j}{\tau}}{\psat{i-1}{\Sigma_{\ri{i-1}}\setminus C}{j}{\tau}+\psat{i-1}{C}{j}{\tau}}>\frac{\nu_{\ri{i-1}}(C)}{2}.$$
    
    \item We define $C''\triangleq \{\sigma\in \Sigma_{\ri{i-1}}\mid \nu_{\ri{i}}(K_\sigma)\geq 1-\frac{1}{4}\}$. 
    By the same argument as before, we have $\nu_{\ri{i-1}}(C'')\ge 1-4\varepsilon_0$ and 
    \[\psat{i-1}{\Sigma_{\ri{i-1}}}{j}{\tau} \geq \psat{i-1}{C''}{j}{\tau} \geq \frac{3}{8}\left(1-4\varepsilon_0\right)\cdot \psat{i}{\Sigma_{\ri{i}}}{j}{\tau}.\]
    Since $\psat{i-1}{\sigma}{j}{\tau}\leq \nu_{\ri{i-1}}(\sigma)\cdot\psat{i}{\Sigma_{\ri{i}}}{j}{\tau}$ for any $\sigma\in \Sigma_{\ri{i-1}}$. For sufficient small $\varepsilon_0$, we have 
    \[\psat{i-1}{\sigma}{j}{\tau}\le 3\nu_{\ri{i-1}}(\sigma)\cdot\psat{i-1}{\Sigma_{\ri{i-1}}}{j}{\tau}.\]

    \item Sort  $\Sigma_{\ri{i}}=\{\sigma_1,\sigma_2,...,\sigma_{|\Sigma_{\ri{i}}|}\}$ in the non-decreasing order according to the value of $\psat{i}{\sigma}{j}{\tau}$. 
    Let $k\in \mathbb{N}$ be the smallest number such that $\sum_{l\le k}\nu_{\ri{i}}(\sigma_l)>\frac{1}{4}$. 
    Then, we have $\frac{\psat{i}{\sigma_k}{j}{\tau}}{ \psat{i}{\Sigma_{\ri{i}}}{j}{\tau}}=\frac{1}{\nu_{\ri{i}}(\sigma_k)} \cdot \frac{\psat{i}{\{\sigma_k\}}{j}{\tau}}{ \psat{i}{\Sigma_{\ri{i}}}{j}{\tau}}\ge \frac{1}{8}$,
    because otherwise $\frac{\psat{i}{\{\sigma_1,\ldots,\sigma_k\}}{j}{\tau}}{\psat{i}{\Sigma_{\ri{i}}}{j}{\tau}}<\sum_{l\le k}\frac{1}{8}\nu_{R_i}(\sigma_l)\leq \frac{1}{8}$, contradicting that $\ri{i}$ is well-distributed based on $(j,\tau)$. 
    Therefore, we have $\{\sigma_k,\sigma_{k+1},...,\sigma_{|\Sigma_{\ri{i}}|}\} \subseteq S$ and 
    \[
    \nu_{\ri{i}}(S)\geq \nu_{\ri{i}}(\{\sigma_k, \sigma_{k+1}, ..., \sigma_{|\Sigma_{\ri{i}}|}\}) \geq \frac{3}{4}.
    \]
\end{enumerate}
The symmetric case that $\ri{j}$ is well-distributed based on $(i,\sigma)$ follows by symmetry.
\end{proof}

The next lemma guarantees the existence of a well-distributed ring in the $J^*$ produced by \Cref{alg:augmenting}.
\begin{lemma}\label{lem:construct-well-distributed}
In the $J^*$ produced by \Cref{alg:augmenting},
    for $1\leq i,j\leq\ell$ with $|i-j|>D$, if $\ri{k}$ is a good  ring for all $\min(i,j) \leq k\leq \max(i,j)$, 
    then $\ri{i}$ is well-distributed based on $(j,\tau)$ for every $\tau\in \Sigma_{\ri{j}}$ avoiding $A_{\lambda_j}$.
\end{lemma}

\begin{proof}
    We prove for the case that $i<j$. The case with $i>j$ follows by symmetry.

    By contradiction, assume that there is some $\tau\in \Sigma_{\ri{j}}$  avoiding $A_{\lambda_j}$ such that  $\ri{i}$ is not well-distributed based on $(j,\tau)$.
    %
    By \Cref{lem:well-distributed-property}, $\ri{k}$ is not well-distributed based on $(j,\tau)$ for all $i\leq k < j$. 
    Thus, for all such $k$, there exists a $C\subseteq \Sigma_{\ri{k}}$ such that $\nu_{\ri{k}}(C)>\frac{1}{4}$ and $\frac{\psat{k}{C}{j}{\tau)}}{ \psat{k}{\Sigma_{\ri{k}}}{j}{\tau
    }}\leq \frac{\nu_{\ri{k}}(C)}{2}.$
    When $k=j-1$, we have $\psat{k}{\Sigma_{\ri{k}}}{j}{\tau}=1$.
    For $i\leq k<j-1$, it can be verified by the law of total probability that $\psat{k}{\Sigma_{\ri{k}}\setminus C}{j}{\tau}\le (1-\nu_{\ri{k}}(C))\cdot \psat{k+1}{\Sigma_{\ri{k+1}}}{j}{\tau}$, hence 
    \[\psat{k}{\Sigma_{\ri{k}}}{j}{\tau} =\frac{\psat{k}{\Sigma_{\ri{k}}\setminus C}{j}{\tau}}{1-\frac{\psat{k}{C}{j}{\tau}}{ \psat{k}{\Sigma_{\ri{k}}}{j}{\tau}}} \le \frac{(1-\nu_{\ri{k}}(C))\cdot \psat{k+1}{\Sigma_{\ri{k+1}}}{j}{\tau}}{1-{\nu_{\ri{k}}(C)}/{2}} \le \frac{6}{7}\cdot \psat{k+1}{\Sigma_{\ri{k+1}}}{j}{\tau}, \]
    which implies that  $\psat{i}{\Sigma_{\ri{i}}}{j}{\tau}\leq \left({\frac{6}{7}}\right)^{D}$. 
    Thus, we have $\psat{i}{\Sigma_{\ri{i}}}{j}{\tau}<\frac{\delta}{2\ell}$ for all sufficiently small constant $\varepsilon_0$, contradicting the termination condition of the \textbf{repeat} loop in \Cref{alg:augmenting}.
\end{proof}


{
Now we are ready to prove \Cref{lem:initialization,lem:iteration} for partial correlation decay between good rings.

\begin{proof}[Proof of \Cref{lem:initialization}]
    Let $j\triangleq i+3D$ and $k\triangleq \lfloor\frac{i+j}{2}\rfloor$.
    By \Cref{lem:construct-well-distributed}, $\ri{k-1}$, $\ri{k}$, $\ri{k+1}$ are well-distributed based on both $(i,\sigma)$ and $(j,\tau)$ for any $\sigma\in \Sigma_{\ri{i}}$ avoiding $A_{\lambda_i}$  and any $\tau\in \Sigma_{\ri{j}}$ avoiding $A_{\lambda_j}$. 
    %
    %
    According to \Cref{lem:well-distributed-property}, for any $\sigma\in \Sigma_{\ri{i}}$ avoiding bad events in $\ein{i}{i}$  and any $\tau\in \Sigma_{\ri{j}}$ avoiding bad events in $\ein{j}{j}$, 
    \begin{align*}
        \psat{i}{\sigma}{j}{\tau}&=\sum_{\rho\in \Sigma_{\ri{k}} }\nu_{\ri{k}}(\rho) \cdot \prod_{v\in \ein{k}{k}}I[\rho\text{ avoids }A_v] \cdot \psat{i}{\sigma}{k}{\rho}\cdot \psat{k}{\rho}{j}{\tau}\\
        &>{\left(\frac{1}{2}-\varepsilon_0\right) }\cdot \left(\frac{1}{8}\right)^{2}\cdot \psat{i}{\sigma}{k}{\Sigma_{\ri{k}}}\cdot \psat{k}{\Sigma_{\ri{k}}}{j}{\tau};\\
    %
\text{and }\quad        \psat{i}{\sigma}{j}{\tau}&=\sum_{\rho\in \Sigma_{\ri{k}}} \nu_{\ri{k}}(\rho) \cdot \prod_{v\in \ein{k}{k}}I[\rho\text{ avoids }A_v] \cdot  \psat{i}{\sigma}{k}{\rho}\cdot \psat{k}{\rho}{j}{\tau} \\
        &\le 3^2\cdot \psat{i}{\sigma}{k}{\Sigma_{\ri{k}}}\cdot \psat{k}{\Sigma_{\ri{k}}}{j}{\tau},
    \end{align*}
    Thus, for any $\sigma_1,\sigma_2\in \Sigma_{\ri{i}}$ avoiding bad events in $\ein{i}{i}$ and any $\tau_1,\tau_2\in \Sigma_{\ri{j}}$ avoiding bad events in $\ein{j}{j}$, 
    \[\frac{\psat{i}{\sigma_1}{j}{\tau_1}\cdot \psat{i}{\sigma_2}{j}{\tau_2}}{\psat{i}{\sigma_2}{j}{\tau_1}\cdot \psat{i}{\sigma_1}{j}{\tau_2}}\le \frac{3^4}{\left(\left(\frac{1}{2} -\varepsilon_0 \right)\cdot \left(\frac{1}{8}\right)^2\right)^2}\le  1+\frac{1}{\varepsilon_0}, \]
    for all sufficiently small constant $\varepsilon_0\le 7\times10^{-7}$.
\end{proof}

\begin{proof}[Proof of \Cref{lem:iteration}]
    Let $j'\triangleq j+2D$. 
    %
    %
    %
    For any $\sigma_1,\sigma_2\in \Sigma_{\ri{i}}$ avoiding bad events in $\ein{i}{i}$  and any $\tau_1,\tau_2\in \Sigma_{\ri{j'}}$ avoiding bad events in $\ein{j'}{j'}$, 
    we consider the following two cases:
    \begin{itemize}
        \item Case I: $P(\sigma_1,\tau_1)\cdot P(\sigma_2,\tau_2)=0$. Then, $P(\sigma_1,\tau_1)\cdot P(\sigma_2,\tau_2) \leq (1+\epsilon) \cdot P(\sigma_1,\tau_2)\cdot P(\sigma_2,\tau_1)$.   
        
        \item Case II: $P(\sigma_1,\tau_1)\cdot P(\sigma_2,\tau_2)>0$. For $\rho\in \Sigma_{\ri{j}}$, define $W(\rho)\triangleq \nu_{R_j}(\rho)\cdot \prod_{v\in \ein{j}{j}} I[\rho \text{ avoids } A_v]$.  Let $S_j\triangleq\{\rho\in \Sigma_{R_j}\mid \text{$\rho$ avoids bad events in $\ein{j}{j}$}\}$. 
        Then,
        \[\begin{aligned}
            P(\sigma_1,\tau_1)\cdot P(\sigma_2,\tau_2)&=\sum_{\rho_1,\rho_2\in \Sigma_{R_j}}W(\rho_1)\cdot W(\rho_2)\cdot P(\sigma_1,\rho_1)\cdot P(\sigma_2,\rho_2) \cdot P(\rho_1,\tau_1)\cdot P(\rho_2,\tau_2)  \\
            &= \sum_{\rho_1,\rho_2\in S_j}W(\rho_1)\cdot W(\rho_2)\cdot P(\sigma_1,\rho_1)\cdot P(\sigma_2,\rho_2) \cdot P(\rho_1,\tau_1)\cdot P(\rho_2,\tau_2)
        \end{aligned}\]
        and 
        \[P(\sigma_1,\tau_2)\cdot P(\sigma_2,\tau_1)=\sum_{\rho_1,\rho_2\in S_j}W(\rho_1)\cdot W(\rho_2)\cdot P(\sigma_1,\rho_2)\cdot P(\sigma_2,\rho_1) \cdot P(\rho_1,\tau_1)\cdot P(\rho_2,\tau_2).\]
        Recall that $\ri{i}$ and $\ri{j}$ are partial \co{$\epsilon$}. 
        According to~\Cref{def:partial-correlated}, for any $\sigma_1,\sigma_2\in \Sigma_{\ri{i}}$ avoiding bad events in $\ein{i}{i}$ and $\rho_1,\rho_2\in \Sigma_{\ri{j}}$ avoiding bad events in $\ein{j}{j}$, it holds that  
        \[P(\sigma_1,\rho_1)\cdot P(\sigma_2,\rho_2) \leq (1+\epsilon)\cdot P(\sigma_1,\rho_2)\cdot P(\sigma_2,\rho_1).\] 
        Then, we have $P(\sigma_1,\tau_2)\cdot P(\sigma_2,\tau_1)>0$. 
        Let \[S_{\sigma_1,\sigma_2}\triangleq \{(\rho_1,\rho_2)\in S_j\times S_j\mid P(\sigma_1,\rho_1)\cdot P(\sigma_2,\rho_2) \leq \cdot P(\sigma_1,\rho_2)\cdot P(\sigma_2,\rho_1).\] 
        Thus, for any $\rho_1,\rho_2\in S_j$,  it holds that $(\rho_1,\rho_2)\in S_{\sigma_1,\sigma_2}$ or $(\rho_2,\rho_1)\in S_{\sigma_1,\sigma_2}$. 
        We define 
        \[P_1\triangleq \sum_{(\rho_1,\rho_2)\in S_{\sigma_1,\sigma_2}}W(\rho_1)\cdot W(\rho_2)\cdot P(\sigma_1,\rho_1)\cdot P(\sigma_2,\rho_2) \cdot P(\rho_1,\tau_1)\cdot P(\rho_2,\tau_2),\]
        \[P_2\triangleq \sum_{(\rho_1,\rho_2)\notin S_{\sigma_1,\sigma_2}}W(\rho_1)\cdot W(\rho_2)\cdot P(\sigma_1,\rho_1)\cdot P(\sigma_2,\rho_2) \cdot P(\rho_1,\tau_1)\cdot P(\rho_2,\tau_2),\]
        \[Q_1\triangleq \sum_{(\rho_1,\rho_2)\in S_{\sigma_1,\sigma_2}}W(\rho_1)\cdot W(\rho_2)\cdot P(\sigma_1,\rho_2)\cdot P(\sigma_2,\rho_1) \cdot P(\rho_1,\tau_1)\cdot P(\rho_2,\tau_2),\]
        \[Q_2\triangleq \sum_{(\rho_1,\rho_2)\notin S_{\sigma_1,\sigma_2}}W(\rho_1)\cdot W(\rho_2)\cdot P(\sigma_1,\rho_2)\cdot P(\sigma_2,\rho_1) \cdot P(\rho_1,\tau_1)\cdot P(\rho_2,\tau_2).\]
        According to \Cref{lem:construct-well-distributed} and \Cref{lem:well-distributed-property}, $\ri{j-1}$, $\ri{j}$, $\ri{j+1}$ are well-distributed based on $(i,\sigma_1)$, $(i,\sigma_2)$,$(j',\tau_1)$ and $(j',\tau_2)$, and it holds that
        \begin{align*}
 Q_1+Q_2
 &\le 3^4\cdot \psat{i}{\sigma_1}{j}{\Sigma_{\ri{j}}}\cdot \psat{i}{\sigma_2}{j}{\Sigma_{\ri{j}}}\cdot  \psat{j}{\Sigma_{\ri{j}}}{j'}{\tau_1}\cdot \psat{j}{\Sigma_{\ri{j}}}{j'}{\tau_2},\\
Q_1
&\ge \frac{1}{2}\cdot \left(\frac{1}{2} -\varepsilon_0 \right)^2\left(\frac{1}{8}\right)^4\cdot \psat{i}{\sigma_1}{j}{\Sigma_{\ri{j}}}\cdot \psat{i}{\sigma_2}{j}{\Sigma_{\ri{j}}}\cdot  \psat{j}{\Sigma_{\ri{j}}}{j'}{\tau_1}\cdot \psat{j}{\Sigma_{\ri{j}}}{j'}{\tau_2}.
        \end{align*}
        Then, we can bound
        \[
        \begin{aligned}
            \frac{P(\sigma_1,\tau_1)\cdot P(\sigma_2,\tau_2)}{P(\sigma_1,\tau_2)\cdot P(\sigma_2,\tau_1)}&=\frac{\sum_{\rho_1,\rho_2\in S_j}W(\rho_1)\cdot W(\rho_2)\cdot P(\sigma_1,\rho_1)\cdot P(\sigma_2,\rho_2) \cdot P(\rho_1,\tau_1)\cdot P(\rho_2,\tau_2)}{\sum_{\rho_1,\rho_2\in S_j}W(\rho_1)\cdot W(\rho_2)\cdot P(\sigma_1,\rho_2)\cdot P(\sigma_2,\rho_1) \cdot P(\rho_1,\tau_1)\cdot P(\rho_2,\tau_2)} \\
            &=\frac{P_1+P_2}{Q_1+Q_2} \leq \frac{Q_1+(1+\epsilon) \cdot Q_2}{Q_1+Q_2} \leq 1+\epsilon\cdot \left(1-\frac{Q_1}{Q_1+Q_2}\right)  \\
            &\le 1+\epsilon\cdot \left(1-{\frac{1}{2}\cdot \left(\frac{1}{2} -\varepsilon_0 \right)^2\left(\frac{1}{8}\right)^4}\cdot {3^{-4}}\right) \leq 1+(1-\varepsilon_0)\cdot \epsilon,
        \end{aligned}
        \]
        for all sufficiently small constant $\varepsilon_0\le 2\times 10^{-5}$.
    \end{itemize}
Altogether, this proves that $R_i$ and $R_{j'}$ are partial $(1-\varepsilon_0)\cdot \epsilon$ correlated. 
\end{proof}
}

\subsection{Omitted proofs of the partial correlation decay}\label{sec:partial-cd-proofs}


\begin{proof}[Proof of \Cref{lem:maintenance}]
    %
    %
    %
    We first prove that $\ring{I}{i}{\Lambda}$ and $\ring{I}{j'}{\Lambda}$ are partial $\epsilon$-correlated. 
    For $\sigma_1,\sigma_2\in \Sigma_{\ring{I}{i}{\Lambda}}$ avoiding bad events in $\eventin{I}{i}{i}{\Lambda}$  and $\tau_1,\tau_2\in \Sigma_{\ring{I}{j'}{\Lambda}}$ avoiding bad events in $\eventin{I}{j'}{j'}{\Lambda}$, consider the following two cases:

    \begin{itemize}
        \item Case I: $P(\sigma_1,\tau_1)\cdot P(\sigma_2,\tau_2)=0$. Then, $P(\sigma_1,\tau_1)\cdot P(\sigma_2,\tau_2) \leq (1+\epsilon) \cdot P(\sigma_1,\tau_2)\cdot P(\sigma_2,\tau_1)$.   
        
        \item Case II: $P(\sigma_1,\tau_1)\cdot P(\sigma_2,\tau_2)>0$. 
        For $\rho\in \Sigma_{\ring{I}{j}{\Lambda}}$, let $W(\rho)\triangleq \nu_{\ring{I}{j}{\Lambda}}(\rho)\cdot \prod_{v\in \eventin{I}{j}{j}{\Lambda}}I [\rho\text{ avoids} A_v]$. 
        Let $S_j\triangleq\{\rho\in \Sigma_{\ring{I}{j}{\Lambda}} \mid \text{$\rho$ avoids bad events in $\eventin{I}{j}{j}{\Lambda}$}\}$.
        Then,
        \[\begin{aligned}
            P(\sigma_1,\tau_1)\cdot P(\sigma_2,\tau_2)&=\sum_{\rho_1,\rho_2\in \Sigma_{R_j}}W(\rho_1)\cdot W(\rho_2)\cdot P(\sigma_1,\rho_1)\cdot P(\sigma_2,\rho_2) \cdot P(\rho_1,\tau_1)\cdot P(\rho_2,\tau_2)  \\
            &= \sum_{\rho_1,\rho_2\in S_j}W(\rho_1)\cdot W(\rho_2)\cdot P(\sigma_1,\rho_1)\cdot P(\sigma_2,\rho_2) \cdot P(\rho_1,\tau_1)\cdot P(\rho_2,\tau_2)
        \end{aligned}\]
        and 
        \[P(\sigma_1,\tau_2)\cdot P(\sigma_2,\tau_1)=\sum_{\rho_1,\rho_2\in S_j}W(\rho_1)\cdot W(\rho_2)\cdot P(\sigma_1,\rho_2)\cdot P(\sigma_2,\rho_1) \cdot P(\rho_1,\tau_1)\cdot P(\rho_2,\tau_2).\]
        Recall that $\ring{I}{i}{\Lambda}$ and $\ring{I}{j}{\Lambda}$ are partial \co{$\epsilon$}. 
        According to~\Cref{def:partial-correlated}, for any $\sigma_1,\sigma_2\in \Sigma_{\ring{I}{i}{\Lambda}}$ avoiding bad events in $\eventin{I}{i}{i}{\Lambda}$ and $\rho_1,\rho_2\in \Sigma_{\ring{I}{j}{\Lambda}}$ avoiding bad events in $\eventin{I}{j}{j}{\Lambda}$, it holds that  
        \[P(\sigma_1,\rho_1)\cdot P(\sigma_2,\rho_2) \leq (1+\epsilon)\cdot P(\sigma_1,\rho_2)\cdot P(\sigma_2,\rho_1).\] 
        Then, we have $P(\sigma_1,\tau_2)\cdot P(\sigma_2,\tau_1)>0$ and

        \[
        \begin{aligned}
            \frac{P(\sigma_1,\tau_1)\cdot P(\sigma_2,\tau_2)}{P(\sigma_1,\tau_2)\cdot P(\sigma_2,\tau_1)}&=\frac{\sum_{\rho_1,\rho_2\in S_j}W(\rho_1)\cdot W(\rho_2)\cdot P(\sigma_1,\rho_1)\cdot P(\sigma_2,\rho_2) \cdot P(\rho_1,\tau_1)\cdot P(\rho_2,\tau_2)}{\sum_{\rho_1,\rho_2\in S_j}W(\rho_1)\cdot W(\rho_2)\cdot P(\sigma_1,\rho_2)\cdot P(\sigma_2,\rho_1) \cdot P(\rho_1,\tau_1)\cdot P(\rho_2,\tau_2)}, \\ &\leq 1+\epsilon.
        \end{aligned}\]
    \end{itemize}
    This proves that $\ring{I}{i}{\Lambda}$ and $\ring{I}{j'}{\Lambda}$ are partial $\epsilon$-correlated. By symmetric argument, it also follows that  $\ring{I}{i'}{\Lambda}$ and $\ring{I}{j}{\Lambda}$ are partial $\epsilon$-correlated. Now treating $i$ and $j'$ as our new ``$i$'' and ``$j$'', applying the old result with $\ring{I}{i'}{\Lambda}$ and $\ring{I}{j}{\Lambda}$ on this new instance, gives us that the partial $\epsilon$-correlation between $\ring{I}{i'}{\Lambda}$ and $\ring{I}{j'}{\Lambda}$ in the original instance $I$.
%
    %
\end{proof}

\begin{proof}[Proof of \Cref{lem:partial-cd-to-cd}]
    For $1\leq i<j$, if $\ring{}{i}{\Lambda}$ and $\ring{}{j}{\Lambda}$ are partial $\epsilon$-correlated, 
    then $\ring{}{[0,i-1]}{\Lambda}$ and $\ring{}{[j+1,\infty]}{\Lambda}$ are $\epsilon$-correlated. 

    Let $S\triangleq \ring{}{[0,i-1]}{\Lambda}$, $T\triangleq \ring{}{[j+1,\infty]}{\Lambda}$.
    For any $\sigma\in \Sigma_S$, $\pi\in \Sigma_{\ring{}{i}{\Lambda}}$, define \[f(\sigma,\pi)\triangleq \nu_{\ring{}{i}{\Lambda}}(\pi)\cdot \prod_{v\in \eventin{I}{i-1}{i}{\Lambda} \cap \eventaround{I}{i}{i}{\Lambda}}I[\sigma \wedge \pi \text{ avoids } A_v].\]
    For any $\tau \in \Sigma_T$, $\rho \in \Sigma_{\ring{}{j}{\Lambda}}$, define \[g(\tau,\rho)\triangleq \nu_{\ring{}{i}{\Lambda}}(\rho)\cdot \prod_{v\in \eventin{I}{j}{j+1}{\Lambda} \cap \eventaround{I}{j}{j}{\Lambda} }I[\tau\wedge \rho \text{ avoids } A_v].\]
    %
    %
    %
    Let $\Sigma_i$ denote the assignments in $\Sigma_{\ring{I}{i}{\Lambda}}$ that avoid all bad events in $\eventin{I}{i}{i}{\Lambda}$. Let $\Sigma_j$ denote the assignments in $\Sigma_{\ring{I}{j}{\Lambda}}$ that avoid all bad events in $\eventin{I}{j}{j}{\Lambda}$. 
    Then, for any $\sigma\in \Sigma_S$, $\tau\in \Sigma_T$, we have 
    \[\nu(\Omega^{\sigma\wedge \tau}) = \nu_S(\sigma) \cdot \nu_T(\tau)\cdot \sum_{\pi\in \Sigma_i, \rho\in \Sigma_j} f(\sigma,\pi)\cdot  g(\tau,\rho) \cdot P(\pi,\rho) .\]
    Recall that $\ring{I}{i}{\Lambda}$ and $\ring{I}{j}{\Lambda}$ are partial $\epsilon$-correlated.
    Thus, for any $\pi_1,\pi_2\in \Sigma_{\ring{I}{i}{\Lambda}}$ and $\rho_1,\rho_2\in \Sigma_{\ring{I}{j}{\Lambda}}$, it holds that 
    \[P(\pi_1,\rho_1)\cdot P(\pi_2,\rho_2)\leq (1+\epsilon)\cdot P(\pi_1,\rho_2)\cdot P(\pi_2,\rho_1)\]
    %
    %
    Thus, for any $\sigma_1, \sigma_2\in \Sigma_S$, $\tau_1,\tau_2\in \Sigma_T$, $\nu(\Omega^{\sigma_1\wedge \tau_1}) \cdot \nu(\Omega^{\sigma_2\wedge \tau_2})$ is equal to 
    \[\nu_S(\sigma_1) \cdot \nu_T(\tau_1)\cdot \nu_S(\sigma_2) \cdot \nu_T(\tau_2) \cdot \sum_{\substack{\pi_1, \pi_2\in \Sigma_i\\ \rho_1,\rho_2\in \Sigma_j}} f(\sigma_1,\pi_1) \cdot g(\tau_1,\rho_1) \cdot  f(\sigma_2,\pi_2) \cdot g(\tau_2,\rho_2) \cdot P(\pi_1,\rho_1)  \cdot P(\pi_2,\rho_2), \]
    which is bounded by $(1+\epsilon)$ multiplies 
    \[\nu_S(\sigma_1) \cdot \nu_T(\tau_1)\cdot \nu_S(\sigma_2) \cdot \nu_T(\tau_2) \cdot \sum_{\substack{\pi_1, \pi_2\in \Sigma_i\\ \rho_1,\rho_2\in \Sigma_j}} f(\sigma_1,\pi_1) \cdot g(\tau_1,\rho_1) \cdot  f(\sigma_2,\pi_2) \cdot g(\tau_2,\rho_2) \cdot P(\pi_1,\rho_2)  \cdot P(\pi_2,\rho_1).\]
    This is exactly $\nu(\Omega^{\sigma_1\wedge \tau_1}) \cdot \nu(\Omega^{\sigma_2\wedge \tau_2})\leq (1+\epsilon)\nu(\Omega^{\sigma_1\wedge \tau_2}) \cdot \nu(\Omega^{\sigma_2\wedge \tau_1})$.
    Thus, $S$ and $T$ are $\epsilon$-correlated.
\end{proof}

\section{Analysis of Initialization and Clustering}\label{sec:clustering}

In this section, we first prove the correctness and efficiency of the \emph{Initialization} phase (\Cref{lem:initialization-stage}),
and then we  prove the correctness and efficiency of the \emph{Clustering} phase (\Cref{lem:clustering-correctness}).

\subsection{Analysis of \emph{Initialization} (Proof of \Cref{lem:initialization-stage})}\label{sec:proof-initialization-stage}

It is sufficient to show that for any $\epsilon\in (0,1)$, the event $|\RV{R}|=O\left(\log n \cdot \log \log n \cdot \log \frac{1}{\gamma} \cdot\log \frac{1}{\epsilon}\right)$ holds with probability at least $1-\epsilon$.
This can be proved by the Chernoff bound.




For each $S\in \mathcal{S}$, let $Y_S\in\{0,1\}$ be the random variable that indicates whether  $A_S$ occurs on $\vec{Y}$.
%
Then $|\RV{R}|=\sum_{S\in\mathcal{S}}Y_S=\sum_{k=1}^{c\log n}\left(\sum_{S\in \mathcal{S},\mathcal{C}(S)=k}Y_S\right)$.
For $1\leq k\leq c\cdot \log n$, let $n_k=\left|\{S\in\mathcal{S}|\mathcal{C}(S)=k\}\right|$, which is the number of clusters with color $k$. 
Then, for each $1\leq k\leq c\cdot \log n$, we have 
\[
\mathbf{E}\left[\sum_{S\in \mathcal{S},\mathcal{C}(S)=k}Y_S\right]=\sum_{S\in \mathcal{S},\mathcal{C}(S)=k}\left(1-\Pr[Y_S=0]\right) \leq n_k-n_k\cdot \left(\prod_{S\in \mathcal{S},\mathcal{C}(S)=k}\Pr[Y_S=0]\right)^{\frac{1}{n_k}}.
\]
Since $\mathcal{C}$ is a proper coloring of $\mathcal{S}$,
$\{Y_S\}_{S\in \mathcal{S}, \mathcal{C}(S)=k}$ are mutually independent random variables.
Recall that the distributed LLL instance $I$ is $\gamma$-satisfiable. 
We have $\prod_{S\in \mathcal{S},\mathcal{C}(S)=k}\Pr[Y_S=0]\ge\gamma$, which implies 
\[\mathbf{E}\left[\sum_{S\in \mathcal{S},\mathcal{C}(S)=k}Y_S\right]\le n_k\cdot \left(1-\gamma^{\frac{1}{n_k}}\right)=n_k\cdot \left(1-\mathrm{e}^{-\frac{1}{n_k}\ln\frac{1}{\gamma}}\right)=O\left(\log\frac{1}{\gamma}\right).\]
Suppose that $c_1$ is a sufficiently large constant.
By Chernoff bound, for any $0<\epsilon<1$, we have
\[\Pr\left[\sum_{S\in \mathcal{S},\mathcal{C}(S)=k}Y_S\ge c_1\cdot (\log\log n)\cdot \log\frac{1}{\gamma}\cdot \log \frac{1}{\epsilon}\right]\le\frac{\epsilon}{c\log n}.\]
By the union bound, we have 
\[\Pr\left[\sum_{S\in \mathcal{S} }Y_S\ge (c\log n) \cdot  c_1\cdot (\log\log n)\cdot\log\frac{1}{\gamma}\cdot\log \frac{1}{\epsilon}\right]\le\frac{\epsilon}{c\log n}\cdot c\log n=\epsilon.\]
Thus, for any $0<\epsilon<1$, we have $|\RV{R}|=O\left(\log n\cdot \log \log n \cdot \log \frac{1}{\gamma}\cdot\log  \frac{1}{\epsilon}\right)$ with probability at least $1-\epsilon$.

\subsection{Analysis of \emph{Clustering} (Proof of \Cref{lem:clustering-correctness})}\label{sec:proof-clustering-correctness}

\paragraph{Balls are uniquely identified and far apart.}

According to the definition of $\RV{B}$, 
if for any distinct $u,v\in \RV{R}$ with $r_v,p_v,r_u,p_u\notin \{\perp\}$ we have $\dist_G(B_{r_v}(p_v), B_{r_u}(p_u)) \geq 2(\ell+2)$, then the followings hold:
\begin{enumerate}
    \item for any distinct $u,v\in \RV{R}$, if $p_u,p_v,r_u,r_v\not\in\{\perp\}$ then $B_{r_u}(p_u)\cap B_{r_v}(p_v)=\emptyset$.
    \item $\dist_G(\mathcal{B}_1,\mathcal{B}_2)\geq 2(\ell+2)$ for any distinct $\mathcal{B}_1, \mathcal{B}_2 \in \RV{B}$.
\end{enumerate}

It then remains to show that we indeed have $\dist_G(B_{r_v}(p_v), B_{r_u}(p_u)) \geq 2(\ell+2)$ for any distinct $u,v\in \RV{R}$ with $r_v,p_v,r_u,p_u\notin \{\perp\}$.
We prove this by induction. 
Suppose that the sequence $\{v_1,v_2,...,v_{|\RV{R}|}\}$ is obtained by sorting $\RV{R}$ in ascending order of IDs. 
\Cref{alg:clustering} is applied on the nodes in $\RV{R}$ in this order.

For the induction basis: initially all $v\in \RV{R}$ set $p_v$ and $r_v$ to $\perp$, which satisfies the hypothesis trivially.
%

Now, suppose the induction hypothesis holds before the execution of \Cref{alg:clustering} on node $v_i$ for some $1\leq i\leq n$.
%
%
After  \Cref{alg:clustering} terminates on $v_i$, it holds that, it holds that $\dist_G(B_{r_{v_i}}(p_{v_i}), B_{r_{v_j}}(p_{v_j}))\geq 2(\ell+2)$ for any  $1\leq j\leq i$, if $p_{v_j}, r_{v_j}\notin \{\perp\}$.
Otherwise, the \textbf{while} loop in \Cref{alg:clustering} would not terminate. 
According to  \Cref{alg:clustering}, for $1\leq j<i$, the only possible modification can be made to $p_{v_j}$ and $r_{v_j}$ during the execution of the algorithm at node $v_i$ is to set them to $\perp$.
Thus, by the induction hypothesis, it still holds that for any distinct $1\leq j,k<i$ such that $p_{v_j}, r_{v,j}, p_{v,k}, r_{v_k}\notin \{\perp\}$, we have $\dist_G(B_{r_{v_j}}(p_{v_j}), B_{r_{v_k}}(p_{v_k}))\geq 2(\ell+2)$.
And for $j>i$, it holds that $p_{v_j}=r_{v_j}=\perp$.
Altogether, it holds that after \Cref{alg:clustering} terminates on node $v_i$, for any distinct $1\leq j,k\leq |\RV{R}|$ such that $p_{v_j}, r_{v,j}, p_{v,k}, r_{v_k}\notin \{\perp\}$, we have $\dist_G(B_{r_{v_j}}(p_{v_j}), B_{r_{v_k}}(p_{v_k}))\geq 2(\ell+2)$.

Thus, after \Cref{alg:clustering} has been sequentially executed on all node $v\in\RV{R}$, for any distinct $u,v\in \RV{R}$ with $r_v,p_v,r_u,p_u\notin \{\perp\}$, it holds that $\dist_G(B_{r_v}(p_v), B_{r_u}(p_u)) \geq 2(\ell+2)$.

\paragraph{Clustered Conditional Gibbs Property. }
Then, we prove that $(\vec{Y},\RV{B})$ satisfies the clustered conditional Gibbs property (as defined in \Cref{def:clustered-conditional-gibbs}) on instance $I$ with parameter $(\epsilon_0, \gamma_0,\delta_0)$ . 

For any $\mathcal{B}\subseteq 2^V$, define 
\begin{align*}
    S(\mathcal{B})
    &\triangleq \bigcup_{\Lambda\in \mathcal{B}} \vbl(\Lambda), 
    &T(\mathcal{B})
    &\triangleq U\setminus \bigcup_{\Lambda\in \mathcal{B}} \vbl(B_{\ell}(\Lambda)), \\
    \Phi(\mathcal{B})
    &\triangleq \left\{A_v \mid v\in V\setminus \bigcup_{\Lambda \in \mathcal{B}} \Lambda \right\},
    &\Phi'(\mathcal{B})
    &\triangleq \{\aug{I}{\epsilon_0, \gamma_0,\delta_0}{\Lambda} \mid  \Lambda\in \mathcal{B} \}.
\end{align*}
%
%
For any $\mathcal{B}\subseteq 2^V$ and $\sigma\in \Sigma_{S(\mathcal{B})}$ with $\Pr[\RV{B}=B\wedge Y_{S(\mathcal{B})} = \sigma]>0$, we define following three events:
\begin{align*}
    \mathcal{A}_1(\mathcal{B}, \sigma):\quad &Y_{S(\mathcal{B})}= \sigma,  \\
    \mathcal{A}_2(\mathcal{B}, \sigma):\quad &\RV{B}=\mathcal{B},  \\
    \mathcal{A}_3(\mathcal{B}, \sigma):\quad &\vec{Y}\text{ avoids all bad events in $\Phi(\mathcal{B})\cup \Phi'(\mathcal{B})$. }
\end{align*}

For any $\mathcal{B}\subseteq 2^V$ and $\sigma\in \Sigma_{S(\mathcal{B})}$ with $\Pr[\RV{B}=\mathcal{B}\wedge Y_{S(\mathcal{B})} = \sigma]>0$, we will prove the equivalence between $\mathcal{A}_1(\mathcal{B}, \sigma) \wedge \mathcal{A}_2(\mathcal{B}, \sigma)$ and $\mathcal{A}_1(\mathcal{B}, \sigma) \wedge \mathcal{A}_3(\mathcal{B}, \sigma)$. 

First, we show \textbf{$\mathcal{A}_1(\mathcal{B}, \sigma) \wedge \mathcal{A}_2(\mathcal{B}, \sigma)\implies \mathcal{A}_1(\mathcal{B}, \sigma) \wedge \mathcal{A}_3(\mathcal{B}, \sigma)$}.
Suppose that $\mathcal{A}_1(\mathcal{B}, \sigma)$ and $\mathcal{A}_2(\mathcal{B}, \sigma)$ happen together. 
If a bad events in $\Phi(\mathcal{B})$ occurs on $\vec{Y}$, according to \Cref{cond:initialization-stage}, 
it must be included in at least one ball, centered on a node in $\RV{R}$ with radius $2d\cdot \log n\log \log n+1$.
And \Cref{alg:clustering} will only combine balls into bigger balls.
Thus, it must hold that $v\in \Lambda$ for some $\Lambda\in \mathcal{B}$, a contradiction. 
If a bad event in $\Phi'(\mathcal{B})$ occurs on $\vec{Y}$, then according to \Cref{alg:clustering}, the algorithm will not terminate with $\RV{B}=\mathcal{B}$, a contradiction.
Therefore, we have the occurrence of $\mathcal{A}_3(\mathcal{B}, \sigma)$.

Next, we show \textbf{$\mathcal{A}_1(\mathcal{B}, \sigma) \wedge \mathcal{A}_3(\mathcal{B}, \sigma)\implies \mathcal{A}_1(\mathcal{B}, \sigma) \wedge \mathcal{A}_2(\mathcal{B}, \sigma)$}. This is proved by induction. 
%
%
%
Define 
\[\mathcal{Z}\triangleq \{Z\in {\Sigma} \mid Z_{S(\mathcal{B})}=\sigma \text{ and $Z$ avoids all bad events in $\Phi(\mathcal{B})\cup \Phi'(\mathcal{B})$ }\}. \]
Since $\Pr[\RV{B}=\mathcal{B}\wedge Y_{S(\mathcal{B})}=\sigma]>0$, there exists at least one assignment $Z\in \mathcal{Z}$ such that $\mathcal{B}$ is output by \Cref{alg:clustering} after running on $Z$.
Next, we will prove that for any $\hat{Z} \in \mathcal{Z}$, \Cref{alg:clustering} will output the same $\mathcal{B}$ after running on $Z'$, which implies $\mathcal{A}_1(\mathcal{B}, \sigma)\wedge \mathcal{A}_3(\mathcal{B}, \sigma) \implies \mathcal{A}_2(\mathcal{B}, \sigma)$. 

%
%
Consdier a fixed $\hat{Z} \in \mathcal{Z}$. Define
\[D \triangleq \{v\in V \mid \text{ $A_v$ occurs on $Z$}\} \quad\text{ and }\quad \hat{D} \triangleq \{v\in V \mid \text{ $A_v$ occurs on $\hat{D}$}\}. \]
Recall that $Z$ avoids all the bad events in $\Phi(\mathcal{B})$.
Thus, for any $v\in D$, it holds that $v\in \bigcup_{\Lambda\in \mathcal{B}} \Lambda$, which means $\vbl(v)\subseteq S(\mathcal{B})$.
The same argument also holds for $\hat{Z}$, which gives $\vbl(v)\subseteq S(\mathcal{B})$.
Assuming $\mathcal{A}_1(\mathcal{B}, \sigma) \wedge \mathcal{A}_2(\mathcal{B}, \sigma)$, it holds that $Z_{S(\mathcal{B})}=\hat{Z}_{S(\mathcal{B})}=\sigma$.
Thus, we have $D=\hat{D}$.

A random set $\RV{R}\subseteq V$ is computed in the  \emph{Initialization} phase from the random assignment $\vec{Y}$ generated according to the product distribution $\nu$.
Denote by $R,\hat{R}\subseteq V$ the respective $\RV{R}$ sets computed in \emph{Initialization} phase from $Z,\hat{Z}$.
It can be verified that $R=\hat{R}$ since $D=\hat{D}$.

Let the nodes in $R=\hat{R}$ be sorted in the ascending order of IDs as $\{v_1,v_2,...,v_{|R|}\}$.
%
Let $1\leq i\leq |R|$. Suppose that  \Cref{alg:clustering} is executed at node $v_i\in R$ on the assignment $\vec{Y}=Z$.
Denote by $N_i$ the total number of iterations of the \textbf{while} loop, and for $0\leq j\leq N_i$ and  $1\leq k\leq |R|$, 
let $p^{(i,j)}_{v_k}$ and $r^{(i,j)}_{v_k}$ respectively denote the $p_{v_k}$ and $r_{v_k}$ computed right after the $j$-th iteration. 
Let $\hat{N}_i$, $\hat{p}^{(i,j)}_{v_k}$ and $\hat{r}^{(i,j)}_{v_k}$ be similarly defined for  $1\leq i\leq |\hat{R}|=|R|$ and $0\leq j\leq \hat{N}_i$ when \Cref{alg:clustering} is executed at node $v_i\in \hat{R}=R$ on the assignment $\vec{Y}=\hat{Z}$.
By convention, let $N_0=\hat{N}_0=0$  and $p^{(0,0)}_{v_i}=r^{(0,0)}_{v_i}=\hat{p}^{(0,0)}_{v_i}=\hat{r}^{(0,0)}_{v_i}=\perp$. 

Next, we prove by induction that, 
for any $0\leq i\leq |R|=|\hat{R}|$, 
it holds that $N_i=\hat{N}_i$,
and furthermore, for any $0\leq j\leq N_i=\hat{N}_i$, any $1\leq k\leq |R|=|\hat{R}|$
it always holds that $p^{(i,j)}_{v_k} = \hat{p}^{(i,j)}_{v_k}$ and $r^{(i,j)}_{v_k} = \hat{r}^{(i,j)}_{v_k}$.

The induction basis holds trivially as $N_0=\hat{N}_0=0$  and $p^{(0,0)}_{v_i}=r^{(0,0)}_{v_i}=\hat{p}^{(0,0)}_{v_i}=\hat{r}^{(0,0)}_{v_i}=\perp$. 

Suppose $p^{(i,j)}_{v_k} = \hat{p}^{(i,j)}_{v_k}$ and $r^{(i,j)}_{v_k} = \hat{r}^{(i,j)}_{v_k}$ for some $0\leq i\leq |R|$, $0\leq j\leq N_i$, $1\leq k\leq |R|$,
and further suppose $N_i=\hat{N}_i$ if $j=N_i$. 
Then, we prove that the same holds for the next iteration.
If $j=N_i$, by the same initialization in \Cref{alg:clustering}, 
we have $p^{(i+1,0)}_{v_k} = \hat{p}^{(i+1,0)}_{v_k}$ and $r^{(i+1,0)}_{v_k} = \hat{r}^{(i+1,0)}_{v_k}$ for $1\leq k\leq |R|$. 
If $j<N_i$, 
%
consider  the following three cases for the $j$-th iteration in \Cref{alg:clustering} at node $v_i$:
        \begin{itemize}
            \item  
            \textbf{Case 1: the \textbf{If} condition in \Cref{alg:clustering-line:if-1} is satisfied.} 
            In this case, there exists $v_{i'}\in R=\hat{R}$ with $i'\neq i$ such that $\dist(B_{r^{(i,j)}_{v_i}}(p^{(i,j)}_{v_i}), B_{\hat{r}^{(i,j)}_{v_{i'}}}(\hat{p}^{(i,j)}_{v_{i'}}))\leq (2\ell+2)$.
            By I.H.: $p^{(i,j)}_{v_k} = \hat{p}^{(i,j)}_{v_k}$ and $r^{(i,j)}_{v_k} = \hat{r}^{(i,j)}_{v_k}$ for $1\leq k\leq |R|$.
            Then the same \textbf{If} condition must be satisfied by the same $v_{i'}$ when $\vec{Y}=\hat{Z}$.
            Thus, $p^{(i,j+1)}_{v_k} = \hat{p}^{(i,j+1)}_{v_k}$ and $r^{(i,j+1)}_{v_k} = \hat{r}^{(i,j+1)}_{v_k}$ for $1\leq k\leq |R|$.

            \item 
            \textbf{Case 2: the \textbf{If} condition in \Cref{alg:clustering-line:if-2} is satisfied.}
            In this case, the new bad event $A_\lambda$ defined on variables outside the ball $B_{r^{(i,j)}_{v_i}}(p^{(i,j)}_{v_i})$ occurs on $Z$ and the radius grows as $r^{(i,j+1)}_{v_i}=r^{(i,j)}_{v_i}+\ell_0(\epsilon_0,\gamma_0,\delta_0)$.
            Since on $\vec{Y}=Z$, $\mathcal{B}$ is eventually output, $A_\lambda$ must be defined over $S(\mathcal{B})$.
            Hence the same \textbf{If} condition must be satisfied when $\vec{Y}=\hat{Z}$. 
            We have $p^{(i,j+1)}_{v_k} = \hat{p}^{(i,j+1)}_{v_k}$ and $r^{(i,j+1)}_{v_k} = \hat{r}^{(i,j+1)}_{v_k}$ for $1\leq k\leq |R|$.

            \item 
            \textbf{Case 3: otherwise.}
            In this case, it can be verified that the above two \textbf{If} conditions are not satisfied on $\vec{Y}=\hat{Z}$ either.
            It is obvious to see this for the \textbf{If} condition in \Cref{alg:clustering-line:if-1},
            since $p^{(i,j)}_{v_k} = \hat{p}^{(i,j)}_{v_k}$ and $r^{(i,j)}_{v_k} = \hat{r}^{(i,j)}_{v_k}$ for $1\leq k\leq |R|$ by I.H..
            For the \textbf{If} condition in \Cref{alg:clustering-line:if-2}, we consider the two subcases:
            (a)~if $B_{r^{(i,j)}_{v_i}}(p^{(i,j)}_{v_i})\in \mathcal{B}$, by the occurrence of $\mathcal{A}_3(\mathcal{B}, \sigma)$ , 
            the condition will not be triggered on $\hat{Z}$;
            (b)~if $B_{r^{(i,j)}_{v_i}}(p^{(i,j)}_{v_i})\notin \mathcal{B}$, in \Cref{alg:clustering}, the bad event $A_\lambda$ must be defined on $S(\mathcal{B})$. Thus, the condition will not be triggered on $\hat{Z}$.
            Thus, we have $p^{(i,j+1)}_{v_k} = \hat{p}^{(i,j+1)}_{v_k}$ and $r^{(i,j+1)}_{v_k} = \hat{r}^{(i,j+1)}_{v_k}$ for $1\leq k\leq |R|$,
            and further have $N_i=\hat{N}_i$.
        \end{itemize}

Thus, $p^{(|R|,N_{|R|})}_{v_k} = \hat{p}^{(|R|,N_{|R|})}_{v_k}$ and $r^{(|R|,N_{|R|})}_{v_k} = \hat{r}^{(|R|,N_{|R|})}_{v_k}$ for $1\leq k\leq |R|$.
Since $\mathcal{B}$ is output by \Cref{alg:clustering} after running on $Z$,
this shows that $\mathcal{B}$ is also output by \Cref{alg:clustering} after running on~$\hat{Z}$.

This proves the equivalence between $\mathcal{A}_1(\mathcal{B}, \sigma) \wedge \mathcal{A}_2(\mathcal{B}, \sigma)$ and $\mathcal{A}_1(\mathcal{B}, \sigma) \wedge \mathcal{A}_3(\mathcal{B}, \sigma)$.
Now we are ready to prove the clustered conditional Gibbs property of $(\vec{Y},\RV{B})$.
For any $\mathcal{B}\subseteq 2^V$ and $\sigma\in \Sigma_{S(\mathcal{B})}$ with $\Pr[\RV{B}=\mathcal{B}\wedge Y_{S(\mathcal{B})}=\sigma]>0$, and any $\tau\in \Sigma_{T(\mathcal{B})}$, 
\[\begin{aligned}
    \Pr[Y_{T(\mathcal{B})}=\tau \mid \RV{B}=\mathcal{B}\wedge Y_{S(\mathcal{B})}=\sigma ] &= \Pr[Y_{T(\mathcal{B})}=\tau \mid \mathcal{A}_1(\mathcal{B}, \sigma)\wedge \mathcal{A}_2(\mathcal{B}, \sigma)] \\
    &= \Pr[Y_{T(\mathcal{B})}=\tau \mid \mathcal{A}_1(\mathcal{B}, \sigma)\wedge \mathcal{A}_3(\mathcal{B}, \sigma)].
\end{aligned}\]
Due to \Cref{cond:initialization-stage}, $\vec{Y}$ follows the product distribution $\nu$.
Thus, the above conditional probability is precisely 
$\mu^{\sigma}_{\widehat{I}, T(\mathcal{B})}(\tau)=
\Pr[Y_{T(\mathcal{B})}=\tau \mid \mathcal{A}_1(\mathcal{B}, \sigma)\wedge \mathcal{A}_3(\mathcal{B}, \sigma)]$, where
\[\widehat{I} =\left(\{X_i\}_{i\in U}, \{A_v\}_{v\in V} \cup \Phi'(\mathcal{B})\right)=\left(\{X_i\}_{i\in U}, \{A_v\}_{v\in V} \cup\{\aug{I}{\epsilon_0, \gamma_0,\delta_0}{\Lambda} \mid  \Lambda\in \mathcal{B} \}\right).
\]

Thus, after \Cref{alg:clustering} has been sequentially executed on all nodes $v\in\RV{R}$, it holds that $(\vec{Y},\RV{B})$ satisfies the clustered conditional Gibbs property  on instance $I$ with parameter $(\epsilon_0, \gamma_0,\delta_0)$.

\paragraph{Balls are reasonably small.}
%
During the execution of \Cref{alg:clustering}  on some node $v\in \RV{R}$, the value of 
    $\mathcal{D}\triangleq\sum_{{v\in \RV{R}: p_v\neq \perp,r_v\neq \perp}} r_v$
may increase in the following cases:

\begin{itemize}

\item \textbf{Case 1: at initialization.} 
For a node $v\in V$, it increases $\mathcal{D}$ by $1+d\cdot \log n\cdot \log \log \log n$ while initializes its radius $r_v$ from $\perp$ to $1+d\cdot \log n\cdot \log \log \log n$. 
In the \emph{Clustering} phase, when \Cref{alg:clustering} is sequentially applied to all $v\in\RV{R}$, such initialization can happen at most $|\RV{R}|$ times. 
Thus, the total contribution to $\mathcal{D}$ of this case is bounded by $|\RV{R}| \cdot (1+d\cdot \log n\cdot \log \log \log n)=\tO(|\RV{R}|\cdot \log n)$. 

\item \textbf{Case 2: when the \textbf{If} condition in \Cref{alg:clustering-line:if-1} is triggered. }
In \Cref{alg:clustering}, once the \textbf{if} condition in \Cref{alg:clustering-line:if-1} is triggered, $\mathcal{D}$ will increase by $2\ell+1$,
and meanwhile, at least one node $u\in \RV{R}$ with $p_u\neq \perp$ will becomes $p_u=\perp$. 
Thus, in the \emph{Clustering} phase, this condition can be triggered at most $|\RV{R}|$ times.
The total contribution to  $\mathcal{D}$ of this case is bounded by $|\RV{R}| \cdot (2\ell+1) = \tO(|\RV{R}| \cdot \log^2 n \log^2 \frac{1}{\gamma})$. 

\item \textbf{Case 3: when the \textbf{If} condition in \Cref{alg:clustering-line:if-2} is triggered. }
For $v\in V$, and $1 \leq r \leq n$, define the event 
\[ \text{ $\mathcal{F}_{v,r}$ :\quad  $\aug{I}{\epsilon_0,\gamma_0,\delta_0}{B_{r}(v)}$ occurs on $\vec{Y}$. }\] 
We define $\mathcal{F} \triangleq \{\mathcal{F}_{v,r}\mid v\in V, 1\leq r \leq n \}$. 
For any $v\in V$ and $1\leq r\leq  n$, according to \Cref{lem:augmenting}, the probability of $\mathcal{F}_{v,r}$ is at most $\frac{1}{n^3}$.
And it can be observed that the \textbf{if} condition in \Cref{alg:clustering-line:if-2} is triggered $k$ times for some $k\ge 1$, only if at least $k$ events in $\mathcal{F}$ happen.
Moreover, these $k$ events must be mutually independent.
This is because once the \textbf{if} condition is triggered and the constructed augmenting event occurs, 
the involved random variable will be included in a ball and will not be used by any other constructed augmenting event that triggers \textbf{if} condition next time.

Thus, it is sufficient to bound the probability that there exists a  subset of $k$ mutually independent events in $\mathcal{F}$ such that all of them happen together, which is  
\[\sum_{\substack{F \subseteq  \mathcal{F}, |F|=k, \\ \text{the events in $F$} \\ \text{are independent}}} \prod_{f\in F} \Pr[\text{$f$ happens}] \leq  \sum_{\substack{F \subseteq  \mathcal{F}, |F|=k, \\ \text{the events in $F$} \\ \text{are independent}}} \prod_{f\in F}\frac{1}{n^3} \leq n^{2k} \cdot \frac{1}{n^{3k}} =\frac{1}{n^{k}}. \]
For any $0<\eta<1$, the probability that the \textbf{if} condition in \Cref{alg:clustering-line:if-2} is triggered at least $k=\log_n \frac{1}{\eta}$ times is at most $n^{-k}\le \eta$. 
Note that each time the \textbf{if} condition is triggered, the value of $\mathcal{D}$ will increase by $\ell$. 
Thus, with probability $1-\eta$, the contribution  $\mathcal{D}$ of this case is bounded by $\tO(|\RV{R}| \cdot \log^2 n \cdot \log^2 \frac{1}{\gamma} \cdot \log \frac{1}{\eta})$. 
\end{itemize}

Overall, for any $0<\eta<1$, we have $\mathcal{D}=\tO(|\RV{R}| \cdot \log^2 n \cdot \log^2 \frac{1}{\gamma} \cdot \log \frac{1}{\eta})$ with probability $1-\eta$.


\section{Analysis of Resampling}\label{sec:analysis-resampling}

In this section, we analyze the \emph{Resampling} phase of the algorithm.
The correctness and efficiency of the \RecursiveSampling{} procedure (\Cref{lem:item:recursive-sample-correctness,lem:item:recursive-sample-complexity} in \Cref{lem:recursive-sample-correctness-complexity}) are respectively proved in \Cref{sec:proof-recursive-sample-correctness,sec:proof-recursive-sample-complexity}.
The accuracy of the estimation  (\Cref{lem:estimate-augmenting}) in the augmented instance is proved in \Cref{sec:proof-estimate-augmenting}.
The correctness of substituting (\Cref{lem:substituting}) is proved in \Cref{sec:proof-substituting}.
Finally,  the analysis of the \emph{Resampling} phase (\Cref{lem:recursive-sample-correctness-complexity}) is wrapped up in \Cref{sec:proof-resample-correctness-complexity}.

\subsection{Correctness of \RecursiveSampling{} (Proof of \Cref{lem:item:recursive-sample-correctness} in \Cref{lem:recursive-sample-correctness-complexity})} \label{sec:proof-recursive-sample-correctness}



First, we prove \Cref{lem:item:recursive-sample-correctness} of \Cref{lem:recursive-sample-correctness-complexity}, 
which guarantees the correctness of \Cref{alg:recursive-sampling-without-cd}.

Assume that \Cref{cond:warm-up-recursive-without-cd} is satisfied by the input of
\RecursiveSampling$(\vec{Y}; I,\Lambda,\epsilon,\gamma,\delta,\alpha)$.
Our goal is to show that $\vec{Y}\sim\mu_I$ when the procedure returns. 
This is proved by a structural induction.
For the induction basis $\Lambda=V$,
\Cref{line:resample} is executed with probability $1$ and the resampled $\vec{Y}$ follows the distribution $\mu_I$.

For the general case,
assume that all recursive calls made within \RecursiveSampling$(\vec{Y}; I,\Lambda,\epsilon,\gamma,\delta,\alpha)$ return with the correct sampling results
as long as \Cref{cond:warm-up-recursive-without-cd} is satisfied by the input arguments to these recursive calls.
%
We then prove $\vec{Y}\sim\mu_I$ when \RecursiveSampling$(\vec{Y};I,\Lambda,\epsilon,\gamma,\delta,\alpha)$ returns.

%

In \Cref{alg:recursive-sampling-without-cd}, 
a $\rho\in[0,1)$ is drawn uniformly at random beforehand.
Then an estimation $[L,R]$ of the probability $P=\Pr_{\vec{X}\sim \mu_I}[\,\vec{X}\text{ avoids }A_{\lambda}\,]$ 
is dynamically improved based on local information, 
to determine whether $\rho<P$ (in which case \Cref{line:first-if} is satisfied, and the algorithm  enters the zone for generating $\vec{Y}\sim\mu_{\hat{I}}$) 
or $\rho\ge P$ (in which case \Cref{line:second-if} is satisfied, and the algorithm enters the zone for generating $\vec{Y}\sim\mu_{\hat{I}'}$).
This inspires us to define the following events:
%
\begin{equation}\nonumber
\begin{aligned} 
    &\mathcal{F}_1: \quad \text{\Cref{line:first-if} is satisfied, and the algorithm  enters the zone for generating $\vec{Y}\sim\mu_{\hat{I}}$}\\
    &\overline{\mathcal{F}_1}: \quad \text{\Cref{line:second-if} is satisfied, and the algorithm enters the zone for generating $\vec{Y}\sim\mu_{\hat{I}'}$.}
\end{aligned}
\end{equation}
Once the algorithm enters one of these two zones, it will not leave the zone until the algorithm returns. Thus, the above two events are mutually exclusive.
Furthermore, at least one of them must occur eventually. 
Indeed, for $P\triangleq\Pr_{\vec{X}\sim \mu_I}\left[\,\vec{X}\text{ avoids }A_{\lambda}\,\right]$, we have the following claim.
%

\begin{claim}\label{cla:sum-estimate-probability}
    $\Pr[\mathcal{F}_1]=P$ and $\Pr[\overline{\mathcal{F}_1}]=1-P$.
\end{claim}

\begin{proof}
    Let $\ell_0=\ell_0(\epsilon_0,\gamma_0,\delta_0)$.
    By \Cref{lem:augmenting}, $A_{\lambda}$ is defined on the random variables in $\vbl(B_{\ell_0}(\Lambda))$.
    For $i\ge 1$, let $\ell_i\triangleq \ell_0\left(\zeta_0^i, \alpha_1, \alpha_2\cdot \zeta_0^i\right)$. 
    For $i\ge 1$, let
    \[\begin{aligned}
        \Class{J}_i \triangleq \{J=(\{X_i\}_{i\in U_J}, \{A_v\}_{v\in V_J}) ~\mid~& 
    \text{$J$ is $\alpha$-satisfiable} \\
    &\text{and \,$J\left(V_J\setminus B^{D_J}_{\ell_0+\ell_i}(\Lambda)\right)$ is $\gamma$-satisfiable} \\
    &\text{and }\left.J\left(B^{D_J}_{\ell_0+\ell_i+1}(\Lambda)\right)=I\left(B^{D_I}_{\ell_0+\ell_i+1}(\Lambda)\right)  \right\}.
    \end{aligned}\]
    %
    For $i\ge 1$, let 
    \[
        L_i\triangleq\inf_{J\in \Class{J}_i} \Pr_{\vec{X}\sim \mu_J}[\text{ $\vec{X}$ avoids $A_{\bar{\lambda}}$ } ] \quad  
        \text{ and } \quad 
        R_i\triangleq\sup_{J\in \Class{J}_i} \Pr_{\vec{X}\sim \mu_J}[\text{ $\vec{X}$ avoids $A_{\bar{\lambda}}$ } ].
    \]
    Since $I\in \Class{J}_i$ for all $i\ge 1$, it always holds that $P\in [L_i,R_i]$. In particular, $P=L_i=R_i$ when $L_1=R_i$.
    And by \Cref{lem:estimate-augmenting},  we have ${R_i}-{L_i} \leq 4\zeta_0^i$ for all $i\ge 1$.

    Let $i_{\max}$ denote the smallest integer $i\ge 1$ with $B_{\ell_0+\ell_i}(\Lambda)=V$.
    Let $j_{\max}$ denote the smallest integer $i\ge 1$ with $L_i=R_i$.
    Observe that $j_{\max}\leq i_{\max}$.
    Then, the \textbf{while} loop stops within at most $j_{\max}$ iterations.
    Moreover, for $1\leq i\leq j_{\max}$, the values of $L_i$ and $R_i$ are computed in \Cref{line:estimate} in the $i$-th iteration.
    
    The probability of $\mathcal{F}_1$ is then calculated.
    By convention, assume $L_0=0$ and $R_0=1$.
    It holds that
    \[\begin{aligned}
        \Pr[\mathcal{F}_1] &= \sum_{1\leq i\leq i_{\max}} \Pr\left[\left(\rho< \max_{0\leq j\leq i} L_j \right) \wedge \left(\max_{0\leq j<i} L_j \leq \rho < \min_{0\leq j<i} R_j\right) \right] \\
        &= \sum_{1\leq i\leq i_{\max}} \frac{\max_{0\leq j\leq i} L_j- \max_{0\leq j<i} L_j}{\min_{0\leq j<i} R_j - \max_{0\leq j<i} L_j} \cdot \left( \min_{0\leq j<i} R_j - \max_{0\leq j<i} L_j \right)\\
        &= \sum_{1\leq i\leq i_{\max}} \left(\max_{0\leq j\leq i} L_j- \max_{0 \leq j<i} L_j \right) =\max_{0\leq i\leq i_{\max}} L_i = P.
    \end{aligned} \]
    As for the probability of  $\overline{\mathcal{F}_1}$, it holds that 
        \begin{align*}
        \Pr[\overline{\mathcal{F}_1}] &= \sum_{1\leq i\leq i_{\max}} \Pr\left[\left(\rho\geq  \min_{0\leq j\leq i} R_j \right) \wedge \left(\max_{0\leq j<i} L_j \leq \rho < \min_{0\leq j<i} R_j\right) \right] \\
        &= \sum_{1\leq i\leq i_{\max}} \frac{\min_{0\leq j< i} R_j - \min_{0\leq j\leq i} R_j }{\min_{0\leq j<i} R_j - \max_{0\leq j<i} L_j} \cdot \left( \min_{0\leq j<i} R_j - \max_{0\leq j<i} L_j \right)\\
        &= \sum_{1\leq i\leq i_{\max}} \left( \min_{0\leq j< i} R_j - \min_{0\leq j\leq i} R_j \right) = 1-\min_{0\leq i\leq i_{\max}} R_i =1-P.\qedhere
    \end{align*} 
\end{proof}

The following claim guarantees the soundness of \Cref{line:filter-soundness} in \Cref{alg:recursive-sampling-without-cd}.

\begin{claim} \label{cla:filter-soundness}
        With probability 1, $\max f>0$ and $\nu\left(\Omega_{\hat{I}}^{ Y_S \wedge Y_T}\right)>0$,
        where $f(\tau)\triangleq \frac{\nu\left(\Omega_{\hat{I}}^{\tau}\right)}{\nu\left(\Omega_{\hat{I}}^{Y_S\wedge \tau}\right)}$ for all $\tau\in \Sigma_T$ with $\nu\left(\Omega_{\hat{I}}^{Y_S\wedge \tau}\right)>0$, similar to \eqref{eq:bayes-filter-f-definition}.
\end{claim}

\begin{proof}
    Let $S=\vbl(\Lambda)$ and $T=U \setminus \vbl(B_{\ell_0(\epsilon,\gamma,\delta)}(\Lambda))$.
    We show that for any $\sigma_1 \in \Sigma_S$ and $\tau_1\in \Sigma_T$ with $\Pr[Y_S=\sigma_1 \wedge Y_T=\tau_1]>0$, it always holds that $\nu\left(\Omega_{\hat{I}}^{\sigma_1\wedge \tau_1}\right)>0$, and furthermore, it holds that $\max f>0$ conditioned on $Y_S=\sigma_1$, 
    which will prove the claim.
    Fix any $\sigma_1 \in \Sigma_S$ and $\tau_1\in \Sigma_T$ with $\Pr[Y_S=\sigma_1 \wedge Y_T=\tau_1]>0$.
    Since the input to \RecursiveSampling{} satisfies \Cref{cond:warm-up-recursive-without-cd}, it holds that $(\vec{Y}, \Lambda)$ satisfies the augmented conditional Gibbs property on instance $I$ with parameter $(\epsilon,\gamma,\delta)$,
    %
    which means $\mu^{\sigma_1}_{\hat{I}, T}(\tau_1)>0$, and hence $\nu\left(\Omega_{\hat{I}}^{\sigma_1\wedge \tau_1}\right)>0$,
    for the LLL instance $\hat{I}$ defined as in \eqref{eq:warm-up-def-hat-I}.
    
    For $\nu\left(\Omega_{\hat{I}}\right) \geq \alpha-\delta>0$, there exist $\sigma_2 \in \Sigma_S$ and $\tau_2 \in \Sigma_T$ such that  $\nu\left(\Omega^{\tau_2}_{\hat{I}}\right)>0$  and $\nu\left(\Omega^{\sigma_2 \wedge \tau_2}_{\hat{I}}\right)>0$. 
    According to \Cref{lem:augmenting}, we have that  $S$ and $T$ are $\epsilon$-correlated in instance $\hat{I}$, i.e.:
    \[\nu\left(\Omega^{\sigma_1 \wedge\tau_1}_{\hat{I}}\right)\cdot \nu\left(\Omega_{\hat{I}}^{\sigma_2\wedge \tau_2}\right) \leq (1+\epsilon) \cdot \nu\left(\Omega^{\sigma_1\wedge \tau_2}_{\hat{I}}\right)\cdot \nu\left(\Omega_{\hat{I}}^{\sigma_2\wedge \tau_1}\right). \]
    Thus,  it holds that $\nu\left(\Omega^{\sigma_1\wedge \tau_2}_{\hat{I}}\right)>0$ and $\max f\geq f(\tau_2) = \nu\left(\Omega^{\tau_2}_{\hat{I}}\right)/{\nu\left(\Omega^{\sigma \wedge \tau_2}_{\hat{I}}\right)}>0$.
\end{proof}



\begin{claim}\label{cla:while-termination-0}
    Conditioned on $\mathcal{F}_1$, when \RecursiveSampling$(\vec{Y}; I,\Lambda,\epsilon,\gamma,\delta,\alpha)$ returns, $\vec{Y}$ follows  $\mu_{\hat{I}}$. 
\end{claim}

    \begin{proof}
        By \Cref{cla:filter-soundness}, the Bayes filter constructed in \Cref{line:filter-soundness} is well-defined.
        Let $\mathcal{F}_2$ denote the event that the Bayes filter defined in \Cref{line:filter-soundness} succeeds. 
        Depending on whether $\mathcal{F}_2$ happens, we proceed in two cases. 

        First, assume that $\mathcal{F}_2$ does not happen.
        In this case, we only need to verify that the input of the recursive call  \RecursiveSampling$\left(\vec{Y}; \hat{I}, B_{r}(\Lambda)\cup\{\lambda\}, \frac{1}{2}, \gamma, \frac{\zeta_0 \alpha}{2}, \frac{\alpha}{2}\right)$ in \Cref{line:first-recursion} satisfies \Cref{cond:warm-up-recursive-without-cd}.
        %
        %
        %
        For any fixed $x\in \mathbb{N}^+$ with $\Pr[\mathcal{F}_1 \wedge \neg \mathcal{F}_2\wedge r=x ]>0$, 
        conditioned on $\mathcal{F}_1 \wedge \neg \mathcal{F}_2 \wedge r=x$,
      the properties asserted in \Cref{cond:warm-up-recursive-without-cd} are verified one by one on the input to this recursive call 
      as follow. 
        \begin{enumerate}
            \item By our assumption, the original input $(\vec{Y}; I,\Lambda,\epsilon,\gamma,\delta,\alpha)$ satisfies \Cref{cond:warm-up-recursive-without-cd}, which means $0<\epsilon\leq \frac{1}{2}$, $0<\alpha \leq \gamma<1$ and $0<\delta<\zeta_0 \cdot \alpha$.
            The same can be easily verified for $\left(\frac{1}{2}, \gamma, \frac{\zeta_0 \alpha}{2}, \frac{\alpha}{2}\right)$. 
            
            \item Since $(\vec{Y}; I,\Lambda,\epsilon,\gamma,\delta,\alpha)$ satisfies \Cref{cond:warm-up-recursive-without-cd}, the LLL instance $I(V\setminus \Lambda)$ is $\gamma$-satisfiable. 
            For $x>\ell_0(\epsilon,\gamma,\delta)$, we have $\hat{I}((V\cup \{\lambda\}) \setminus (B_{x}(\Lambda)\cup \{\lambda\}) ) = I(V\setminus B_x(\Lambda))$, which must also be $\gamma$-satisfiable since $V\setminus B_x(\Lambda)\subseteq V\setminus \Lambda$. 
            Furthermore, it holds that $\nu(\Omega_{\hat{I}})\geq \alpha - \delta \geq (1-\zeta_0) \alpha$.
            By choosing $\zeta_0>0$ to be a sufficient small constant, the LLL instance $\hat{I}$ is at least $\frac{\alpha}{2}$-satisfiable.
            
            \item 
            Let $\hat{S}=\vbl(B_{x}(\Lambda))$ and $\hat{T}=U\setminus \vbl(B_{x+\ell_0\left( \frac{1}{2}, \gamma, \frac{\zeta_0 \alpha}{2} \right)}(\Lambda))$. 
            Let $A_{\kappa}$ and $\hat{I}_{\kappa}$, where $\kappa\not\in V$, respectively denote the bad event and the augmented LLL instance constructed in \Cref{lem:augmenting}, where 
            \[
                A_{\kappa}\triangleq\aug{\hat{I}}{ \frac{1}{2}, \gamma, \frac{\zeta_0 \alpha}{2}}{B_x(\Lambda)}, 
                \quad\text{ and }\quad 
                \hat{I}_{\kappa}\triangleq (\{X_i\}_{i\in U}, \{A_v\}_{v\in V}\cup\{A_{\lambda}\}\cup \{A_{\kappa}\}).
            \]
            Let $\mathcal{A}$ denote the event that $\vec{Y}$ avoids the bad event $A_{\kappa}$. 
            For any $\sigma \in \Sigma_{\hat{S}}$ with $\Pr[\mathcal{F}_1 \wedge \neg \mathcal{F}_2\wedge r=x \wedge Y_{\hat{S}}=\sigma ]>0$,
            since $Y_{\hat{S}}=\sigma$ already ensures $r_1\geq x$, we have
            \[\mathcal{F}_1 \wedge \neg \mathcal{F}_2\wedge r=x \wedge Y_{\hat{S}}=\sigma   \iff \mathcal{F}_1 \wedge \neg \mathcal{F}_2\wedge \mathcal{A} \wedge Y_{\hat{S}}=\sigma .\]
            Moreover, given $\mathcal{A}$ and $Y_{\hat{S}} = \sigma$, the event $\neg \mathcal{F}_2$ is conditionally independent of $Y_{\hat{T}}$.
            Recall that $(\vec{Y},\Lambda)$ satisfies the augmented conditional Gibbs property on instance $I$ with parameter $(\epsilon,\gamma,\delta)$.
            %
            %
            %
            Thus, conditioned on $\mathcal{F}_1 \wedge \neg \mathcal{F}_2\wedge r=x \wedge Y_{\hat{S}}=\sigma$, it holds that $Y_{\hat{T}}$ follows the distribution $\mu^{\sigma}_{\hat{I}_1, \hat{T}}$.
            %
            %
            Thus, $(\vec{Y}, B_x(\Lambda)\cup\{\lambda\})$ is augmented conditional Gibbs on instance $\hat{I}$ with parameter $\left(\frac{1}{2}, \gamma, \frac{\zeta_0 \alpha}{2} \right)$.
        \end{enumerate}
        Altogether, conditioned on $\mathcal{F}_1\land \neg \mathcal{F}_2\land r=x$, the input to the recursive call in \Cref{line:first-recursion} satisfies \Cref{cond:warm-up-recursive-without-cd}.
        By the induction hypothesis, right after the recursive call in \Cref{line:first-recursion} returns, $\vec{Y}$ follows the distribution $\mu_{\hat{I}}$.
        Since this holds for all possible $x\in \mathbb{N}^+$, we have that the output $\vec{Y}\sim \mu_{\hat{I}}$ conditioned on $\mathcal{F}_1$ and $\neg \mathcal{F}_2$.

        The remaining case is that $\mathcal{F}_2$ happens.
        Let $S=\vbl(\Lambda)$ and $T=U \setminus \vbl(B_{\ell_0(\epsilon,\gamma,\delta)}(\Lambda))$.
        Fix any $\sigma \in \Sigma_S$ with $\Pr[\mathcal{F}_1 \wedge \mathcal{F}_2 \wedge Y_S=\sigma]>0$.
        We prove that, conditioned on $\mathcal{F}_1\land \mathcal{F}_2\land Y_S=\sigma$, the $\vec{Y}$ follows the distribution $\mu_{\hat{I}}$ after  \Cref{line:resample} being executed. 
        In the following analysis, let $\vec{Y}=(Y_i)_{i\in U}$ denote the original input random assignment and let $\vec{Y}'=(Y'_i)_{i\in U}$ denote the $\vec{Y}$ after \Cref{line:resample}  being executed.

        Recall that $S$ and $T$ are $\epsilon$-correlated in instance $\hat{I}$.
        For any $\pi\in \Sigma$ with $\nu(\Omega_{\hat{I}}^{\pi_T\wedge \sigma})=0$, it holds that $\mu_{\hat{I}, T}(\pi_T)=0$.
        Thus, for any $\pi\in \Sigma$ with $\mu_{\hat{I}, T}(\pi_T)>0$, it holds that $\nu(\Omega_{\hat{I}}^{\pi_T\wedge \sigma})>0$.
        For any $\pi\in \Sigma$ with $\mu_{\hat{I}}(\pi)>0$, conditioned on $\mathcal{F}_1$ and $Y_S=\sigma$, it holds that
        \[\begin{aligned}
            \Pr[\vec{Y}'=\pi \wedge \mathcal{F}_2] &= \Pr[Y'_T=\pi_T]\cdot \Pr[Y'_{U\setminus T} = \pi_{U\setminus T} \mid Y'_T=\pi_T] \cdot \Pr[\mathcal{F}_2\mid \vec{Y}'=\pi] \\
            &= \mu^{\sigma}_{\hat{I}, T}(\pi_T)\cdot \mu^{\pi_{T}}_{\hat{I}, U\setminus T}(\pi_{U\setminus T}) \cdot \frac{\nu\left(\Omega_{\hat{I}}^{\pi_T}\right)}{\nu\left(\Omega_{\hat{I}}^{\sigma\wedge \pi_T }\right)} \cdot \frac{1}{\max f} = \mu_{\hat{I}}(\pi) \cdot \frac{\nu\left(\Omega_{\hat{I}}\right)}{\nu\left(\Omega_{\hat{I}}^{\sigma}\right)} \cdot \frac{1}{\max f}.
        \end{aligned}\]
        Then, conditioned on $\mathcal{F}_1$ and $Y_S=\sigma$, it holds that
        \[\Pr[\mathcal{F}_2] = \sum_{\pi\in \Sigma} \Pr[\vec{Y}'=\pi \wedge \mathcal{F}_2] =  \sum_{\pi\in \Sigma} \mu_{\hat{I}}(\pi) \cdot \frac{\nu\left(\Omega_{\hat{I}}\right)}{\nu\left(\Omega_{\hat{I}}^{\sigma}\right)} \cdot \frac{1}{\max f} = \frac{\nu\left(\Omega_{\hat{I}}\right)}{\nu\left(\Omega_{\hat{I}}^{\sigma}\right)} \cdot \frac{1}{\max f}. \]
        Thus, conditioned on $\mathcal{F}_1\land \mathcal{F}_2\land Y_S=\sigma$, the output $\vec{Y}'$ follows $\mu_{\hat{I}}$. 
        By the law of total probability, conditioned on $\mathcal
        F_1$ and $\mathcal{F}_2$, after \RecursiveSampling$(\vec{Y}; I,\Lambda,\epsilon,\gamma,\delta,\alpha)$ returns, $\vec{Y}$ follows  $\mu_{\hat{I}}$.
\end{proof}

    \begin{claim}\label{cla:while-termination-1}
        Conditioned on $\overline{\mathcal{F}_1}$, when \RecursiveSampling$(\vec{Y}; I,\Lambda,\epsilon,\gamma,\delta,\alpha)$ returns, $\vec{Y}$ follows  $\mu_{\hat{I}'}$.
    \end{claim}

    \begin{proof}

    In this case, we only need verify that the input $\left(\vec{Y}; \hat{I}', B_{s}(\Lambda)\cup \left\{\overline{\lambda}\right\},  \frac{1}{2}, \gamma,  \frac{\zeta_0 \alpha  (1-\hat{R})}{2}, \alpha  (1-R) \right)$ of the recursive call to \RecursiveSampling{}
    in \Cref{line:second-recursion} satisfies \Cref{cond:warm-up-recursive-without-cd}.
        %
        %
        %
        For any fixed $x\in \mathbb{N}^+$ with $\Pr(\overline{\mathcal{F}_1} \wedge s=x )>0$,  conditioned on $\overline{\mathcal{F}_1}$ and $ s=x$,
        the properties asserted in \Cref{cond:warm-up-recursive-without-cd} are verified one by one on the input to this recursive call as follows. 
        \begin{enumerate}
            \item By our assumption, the original input $(\vec{Y}; I,\Lambda,\epsilon,\gamma,\delta,\alpha)$ satisfies \Cref{cond:warm-up-recursive-without-cd}, which means $0<\epsilon\leq \frac{1}{2}$, $0<\alpha \leq \gamma<1$ and $0<\delta<\zeta_0 \cdot \alpha$.
            The same can be verified on $\left(\frac{1}{2}, \gamma,  \frac{\zeta_0 \alpha  (1-\hat{R})}{2}, \alpha  (1-R) \right)$. 

            \item Since $(\vec{Y}; I,\Lambda,\epsilon,\gamma,\delta,\alpha)$ satisfies \Cref{cond:warm-up-recursive-without-cd}, the LLL instance $I(V\setminus \Lambda)$ is $\gamma$-satisfiable. 
            For $x>\ell_0(\epsilon,\gamma,\delta)$, we have $\hat{I}'((V\cup \{\bar{\lambda}\}) \setminus (B_{x}(\Lambda)\cup \{\bar{\lambda}\}) ) = I(V\setminus B_x(\Lambda))$, which must also be $\gamma$-satisfiable since $V\setminus B_x(\Lambda)\subseteq V\setminus \Lambda$. 
            Furthermore, it holds that $\nu(\Omega_{\hat{I}'})\geq \nu(\Omega_I)\cdot (1-P) \geq \alpha\cdot (1-R)$.
            %
            
            \item 
            Let $\hat{S}=\vbl(B_{x}(\Lambda))$ and $\hat{T}=U\setminus \vbl(B_{x+\ell_0\left( \frac{1}{2}, \gamma, \frac{\zeta_0 \alpha}{2} \right)}(\Lambda))$. 
            Let $A_{\kappa}$ and $\hat{I}_{\kappa}$, where $\kappa\not\in V$, respectively denote the bad event and the augmented LLL instance constructed in \Cref{lem:augmenting}, where 
            \[
                A_{\kappa}\triangleq\aug{\hat{I}'}{ \frac{1}{2}, \gamma,  \frac{\zeta_0 \alpha  (1-R)}{2}}{B_x(\Lambda)}, 
                \quad\text{ and }\quad 
                \hat{I}_1'\triangleq (\{X_i\}_{i\in U}, \{A_v\}_{v\in V}\cup\{A_{\bar{\lambda}}\}\cup \{A_{\kappa}\}\}.
            \]
            Let $\mathcal{A}$ denote the event that $\vec{Y}$ avoids the bad event $A_{\kappa}$. 
            For any $\sigma \in \Sigma_{\hat{S}}$ with $\Pr[\overline{\mathcal{F}_1} \wedge s=x \wedge Y_{\hat{S}}=\sigma ]>0$, since $Y_{\hat{S}}=\sigma$ already ensures $s\geq x$, we have
            \[\overline{\mathcal{F}_1} \wedge s=x \wedge Y_{\hat{S}}=\sigma   \iff \neg \overline{\mathcal{F}_1} \wedge \mathcal{A} \wedge Y_{\hat{S}}=\sigma .\]
            %
            %
            Recall that by our assumption, $(\vec{Y},\Lambda)$ satisfies the augmented conditional Gibbs property on instance $I$ with parameter $(\epsilon,\gamma,\delta)$.
            %
            %
            Thus, conditioned on $\overline{\mathcal{F}_1} \wedge s=x \wedge Y_{\hat{S}}=\sigma$, it holds that $Y_{\hat{T}}$ follows the distribution $\mu^{\sigma}_{\hat{I}_1', \hat{T}}$, i.e. $(\vec{Y}, B_x(\Lambda)\cup\{\overline{\lambda}\})$ satisfies the augmented conditional Gibbs property on instance $\hat{I}'$ with parameter $\left(\frac{1}{2}, \gamma,  \frac{\zeta_0 \alpha  (1-R)}{2} \right)$. 
        \end{enumerate}
        Altogether, conditioned on $\overline{\mathcal{F}_1}$ and $r=x$, the input to the recursive call in \Cref{line:second-recursion} satisfies \Cref{cond:warm-up-recursive-without-cd}.
        By the induction hypothesis, right after the recursive call in \Cref{line:second-recursion} returns, $\vec{Y}$ follows the distribution $\mu_{\hat{I}}$.
                Since this holds for all possible $x\in \mathbb{N}^+$, by the law of total probability, we have that the output $\vec{Y}\sim \mu_{\hat{I}'}$ conditioned on $\overline{\mathcal{F}_1}$.
    \end{proof}

Now we are ready to finalize the proof of the correctness of \RecursiveSampling{}.
    %
Denote by $\vec{Y}^*$ the random assignment $\vec{Y}$ when \RecursiveSampling$(\vec{Y}; I,\Lambda,\epsilon,\gamma,\delta,\alpha)$ returns.
By \Cref{cla:sum-estimate-probability} and \Cref{cla:while-termination-0},
    for any $\pi\in \Sigma_V$ avoiding  $A_{\lambda}$,  it holds that 
    \[\begin{aligned}
        \Pr[\vec{Y}^*=\pi] &= \Pr[\vec{Y}^*=\pi \mid \mathcal{F}_1] \cdot \Pr[\mathcal{F}_1] +\Pr[\vec{Y}^*=\pi \mid \overline{\mathcal{F}_1}] \cdot \Pr[\overline{\mathcal{F}_1}] \\
        &= \mu_{\hat{I}}(\pi) \cdot P = \mu_I(\pi). 
    \end{aligned} \]
On the other hand, by \Cref{cla:sum-estimate-probability} and \Cref{cla:while-termination-1},  for any $\pi\in \Sigma_V$ not avoiding $A_{\lambda}$, it holds that 
    \[\begin{aligned}
        \Pr[\vec{Y}^*=\pi] &= \Pr[\vec{Y}^*=\pi \mid \mathcal{F}_1] \cdot \Pr[\mathcal{F}_1] +\Pr[\vec{Y}^*=\pi \mid \overline{\mathcal{F}_1}] \cdot \Pr[\overline{\mathcal{F}_1}] \\
        &= \mu_{\hat{I}'}(\pi) \cdot (1-P)= \mu_I(\pi). 
    \end{aligned} \]

This proves that $\vec{Y}\sim\mu_I$ when \RecursiveSampling$(\vec{Y}; I,\Lambda,\epsilon,\gamma,\delta,\alpha)$ returns, assuming \Cref{cond:warm-up-recursive-without-cd} satisfied initially by the input,
which finishes the inductive proof of the correctness of \RecursiveSampling{}.
    

\subsection{Efficiency of \RecursiveSampling{} (Proof of \Cref{lem:item:recursive-sample-complexity} in \Cref{lem:recursive-sample-correctness-complexity})}
\label{sec:proof-recursive-sample-complexity}

We now prove \Cref{lem:item:recursive-sample-complexity} of \Cref{lem:recursive-sample-correctness-complexity}, which bounds the efficiency of \Cref{alg:recursive-sampling-without-cd}.

Assume that \Cref{cond:warm-up-recursive-without-cd} is satisfied by the input of
\RecursiveSampling$(\vec{Y}; I,\Lambda,\epsilon,\gamma,\delta,\alpha)$.
To  upper bound the complexity of the recursive algorithm,
we construct a potential $\mathcal{P}\ge 0$,
%
%
which is computed during the execution of \RecursiveSampling{}  according to the following rules which are added into \Cref{alg:recursive-sampling-without-cd}:
\begin{itemize}
\item $\mathcal{P}$ is initialized to $0$, and the current value of $\mathcal{P}$ is returned whenever the algorithm returns;
\item when the recursive call in \Cref{line:first-recursion} returns with some value $\mathcal{P}_1$,  
the value of $\mathcal{P}$ is increased by~$\mathcal{P}_1$;
%
\item when the recursive call in \Cref{line:second-recursion} returns with some value $\mathcal{P}_2$,  
the current value of $\mathcal{P}$ is increased by $\mathcal{P}_2+ \lceil\log \frac{1}{1-R}\rceil+1$, where $R$ is the current estimation upper bound calculated in \Cref{line:estimate};
\item whenever  \Cref{line:add-estimate-radius} is executed, the value of $\mathcal{P}$ is increased by 1.
%
 \end{itemize}
This procedure for computing the potential $\mathcal{P}$ is explicitly described in \Cref{alg:recursive-sampling-ideal-P}.
Note that other than the part for computing the value of $\mathcal{P}$, \Cref{alg:recursive-sampling-ideal-P} is exactly the same as \Cref{alg:recursive-sampling-without-cd}.
Therefore, we refer to them both by the same name  \RecursiveSampling$(\vec{Y}; I,\Lambda,\epsilon,\gamma,\delta,\alpha)$.

\begin{algorithm}
\caption{\RecursiveSampling($\vec{Y}$; $I, \Lambda, \epsilon, \gamma, \delta, \alpha$) that computes $\mathcal{P}$ during execution}
\label{alg:recursive-sampling-ideal-P}
\SetKwInOut{Input}{Input}
\SetKwInOut{Data}{Data}
\SetKwInOut{Output}{Output}

\Input{LLL instance $I=\left(\{X_i\}_{i\in U}, \{A_v\}_{v\in V}\right)$,  subset $\Lambda\subseteq V$, parameter  $(\epsilon, \gamma, \delta, \alpha)$;}
\Data{assignment $\vec{Y}=(Y_i)_{i\in U}$ stored globally that can be updated by the algorithm;}
\Output{integer $\mathcal{P}\ge 0$;}

\tcp{\small An integer $\mathcal{P}\ge 0$ is computed and returned (\textcolor{blue}{which are highlighted}); otherwise, the algorithm is the same as  \Cref{alg:recursive-sampling-without-cd}.}

\SetKwIF{withprob}{}{}{with probability}{do}{}{}{}

    
initialize \textcolor{blue}{$\mathcal{P}\gets 0$}, $i\gets 1$, and define $\ell_0\triangleq \ell_0(\epsilon,\gamma, \delta)$\; 

\tcp{\small The value of $\mathcal{P}$ is initialized to 0.}

draw $\rho\in[0,1)$ uniformly at random\; 

\While{true
}{


    ${\ell}_i \gets \ell_0\left(\zeta_0^i, \gamma, \alpha\cdot \zeta_0^i\right)$\;

    compute the smallest interval $[{L},{R}]$ containing $P\triangleq \Pr_{\vec{X}\sim \mu_I}[\,\vec{X}\text{ avoids }A_{\lambda}\,]$
    based on $\Lambda$, $A_{\lambda}$, $I\left(B_{\ell_0+{\ell}_i+1}(\Lambda)\right)$, assuming that $I$ is $\alpha$-satisfiable and $I\left(V\setminus B_{\ell_0+{\ell}_i+1}(\Lambda)\right)$ is $\gamma$-satisfiable\; 



    \If{$\rho< {L}$ 
    }{

            
                
        
        define $T\triangleq U\setminus \vbl(B_{\ell_0}(\Lambda))$\;

        {Define $f(\tau)\triangleq \frac{\nu\left(\Omega_{\hat{I}}^{\tau}\right)}{\nu\left(\Omega_{\hat{I}}^{Y_S\wedge \tau}\right)}$ for all $\tau\in \Sigma_T$ with $\nu\left(\Omega_{\hat{I}}^{Y_S\wedge \tau}\right)>0$, similar to \eqref{eq:bayes-filter-f-definition}}\;
        \withprob{$\frac{f\left(Y_{T}\right)}{\max f}$ 
        }{ 
            update $\vec{Y}$ by redrawing $Y_{U\setminus T}\sim \mu_{\hat{I}, U\setminus T}^{Y_{T}}$\; 
            
        } 
        \Else{
            initialize  $r\gets \ell_0+1$\;

            
            \While{$\vec{Y}$ does not avoid the bad event $\aug{\hat{I}}{{1}/{2},\, \gamma,\, {\zeta_0\alpha}/{2}}{B_{r}(\Lambda)}$
            }{
                grow the ball:  $r\gets r+\ell_0\left(\frac{1}{2}, \gamma,\frac{\zeta_0\alpha}{2}\right)$ and  \textcolor{blue}{$\mathcal{P}\gets \mathcal{P}+1$}\; \label{line:compute-p-r1}
                
 \tcp{\small The value of $\mathcal{P}$ is increased by 1.}
            }            
            
            
            \textcolor{blue}{$\mathcal{P}\gets\mathcal{P}\, +\, $}\RecursiveSampling$\left(\vec{Y}; \hat{I}, B_{r}(\Lambda)\cup\{\lambda\}, \frac{1}{2}, \gamma, \frac{\zeta_0 \alpha}{2}, \frac{\alpha}{2}\right)$\; \label{line:compute-p-p1}


            \tcp{\small The value of $\mathcal{P}$ is increased by the amount returned by the recursive call.}
        }
        
            \Return{\textcolor{blue}{$\mathcal{P}$}}\; 
        }
        \ElseIf{$\rho\ge {R}$ 
        }{
        
            initialize $s\gets \ell_0 +1$\;

    
            \While{$\vec{Y}$ does not avoid the bad event $\aug{\hat{I}'}{1/2,\,\gamma,\, {\zeta_0\alpha  (1-R)}/{2}}{B_{s}(\Lambda)}$ 
            }{
                grow the ball: $s\gets s+\ell_0\left(\frac{1}{2},\gamma, \frac{\zeta_0\alpha (1-R)}{2}\right)$ and  \textcolor{blue}{$\mathcal{P}\gets \mathcal{P}+1$}\; \label{line:compute-p-r2}

 \tcp{\small The value of $\mathcal{P}$ is increased by 1.}
            }
            \textcolor{blue}{$\mathcal{P}\gets \mathcal{P}+ \lceil \log \frac{1}{1-R}\rceil + $}\RecursiveSampling$\left(\vec{Y}; \hat{I}', B_{s}(\Lambda)\cup \left\{\overline{\lambda}\right\},  \frac{1}{2}, \gamma,  \frac{\zeta_0 \alpha  (1-R)}{2}, \alpha  (1-R) \right)$\; \label{line:compute-p-p2}

            
            \tcp{\small The value of $\mathcal{P}$ is increased by the amount returned by the recursive call.}
            
           \Return{\textcolor{blue}{$\mathcal{P}$}}\; 

        }
        \Else{
        
        enter the next iteration (and refine the estimation of $P$): $i\gets i+1$ and  \textcolor{blue}{$\mathcal{P}\gets \mathcal{P}+1$}\;   \label{line:compute-p-estimate}

                   \tcp{\small The value of $\mathcal{P}$ is increased by 1.}
        }
    }
\end{algorithm}

%
Next, we prove \Cref{cla:bound-radius-using-p} and \Cref{cla:bound-p},
where \Cref{cla:bound-radius-using-p} says that the potential $\mathcal{P}$ computed as above can be used to bound the radius $r$ of \RecursiveSampling$(\vec{Y}; I,\Lambda,\epsilon,\gamma,\delta,\alpha)$ stated in  \Cref{lem:item:recursive-sample-complexity} of \Cref{lem:recursive-sample-correctness-complexity},
and  \Cref{cla:bound-p} gives an upper bound on the potential~$\mathcal{P}$.

\begin{claim}\label{cla:bound-radius-using-p}
    Suppose that \Cref{cond:warm-up-recursive-without-cd} is satisfied by the input of
    \RecursiveSampling$(\vec{Y}; I,\Lambda,\epsilon,\gamma,\delta,\alpha)$,     
    who accesses $I(B_{r}(\Lambda))$, updates $Y_{\vbl(B_{r}(\Lambda))}$,
    and returns a $\mathcal{P}\ge 0$.
    It holds that 
    $$r\le \ell_0(\epsilon,\gamma,\delta) + O\left(\mathcal{P}\cdot\left(\mathcal{P} + \log \frac{1}{\alpha} \right) \cdot \log \frac{1}{\gamma} \cdot \log\left(\mathcal{P}\log \frac{1}{\gamma}\log \frac{1}{\alpha} \right)\right).$$
\end{claim}

\begin{proof}

    Let $c>1$ be a sufficiently large constant. 
    For $\mathcal{P}\in \mathbb{N}$, $\gamma\in (0,1)$, $\alpha\in (0,1)$,  define
    \[\ellf(\mathcal{P},\gamma,\alpha)= c\cdot \mathcal{P}\cdot\left(\mathcal{P} + \log \frac{1}{\alpha} \right) \cdot \log \frac{1}{\gamma} \cdot \log\left(\mathcal{P}\log \frac{1}{\gamma}\log \frac{1}{\alpha} \right).\]
    
    \Cref{cla:bound-radius-using-p} is proved by showing  that $r\le \ell_0(\epsilon,\gamma,\delta)+\ellf(\mathcal{P},\gamma,\alpha)$.
    This is proved by a structural induction.
    Suppose that the input is the LLL instance $I=\left(\{X_i\}_{i\in U}, \{A_v\}_{v\in V}\right)$, the random assignment $\vec{Y}=(Y_i)_{i\in U}$, and the subset $\Lambda\subseteq V$ of bad events, along with the parameter  $(\epsilon, \gamma, \delta, \alpha)$ satisfying \Cref{cond:warm-up-recursive-without-cd}.
    
    First, for the induction basis, when  $\Lambda=V$, the induction hypothesis holds as $r=0$ and $\mathcal{P}=0$.

    Now, consider the general case. 
    %
    %
    %
    %
    Define
    \begin{align*}
        \ell
        &\triangleq \ell_0(\epsilon,\gamma,\delta)= O\left( \log \frac{1}{\epsilon} \log \frac{1}{\gamma} \log \frac{1}{\delta} \log\left(\log \frac{1}{\epsilon} \log \frac{1}{\gamma} \log \frac{1}{\alpha} \right)\right), \\
        \ell_1
        &\triangleq\ell_0\left(\frac{1}{2}, \gamma, \frac{\zeta_0 \alpha}{2}\right)= O\left( \log \frac{1}{\gamma} \log \frac{1}{\alpha} \log\left( \log \frac{1}{\gamma} \log \frac{1}{\alpha} \right)\right) ,\\
        \ell_2
        &\triangleq\ell_0\left(\frac{1}{2},\gamma, \frac{\zeta_0\alpha (1-R)}{2}\right)= O\left( \log \frac{1}{\gamma} \log \frac{1}{\alpha (1-R)} \log\left(\log \frac{1}{\gamma} \log \frac{1}{\alpha (1-R)} \right)\right),\\
        \ell_3
        &\triangleq \ell_0\left(\zeta_0^K, \gamma, \alpha\cdot \zeta_0^K\right) = O\left( K\log \frac{1}{\gamma} \left(\log \frac{1}{\alpha}+K\right) \log\left(K\log \frac{1}{\gamma} \left(\log \frac{1}{\alpha}+K \right)\right)\right),\\
        &\quad\text{ where $K$ stands for the number of times that \Cref{line:compute-p-estimate} is executed.}
    \end{align*}

The induction then proceeds in the following three cases:
    \begin{itemize}
        \item \textbf{Case 1: \Cref{line:compute-p-p1} is executed.} 
        Let $K_1$ denote the number of times that \Cref{line:compute-p-r1} is executed.
        Let $\mathcal{P}_1$ denote the value returned by the recursive call in \Cref{line:compute-p-p1}.
        In this case, we have $\mathcal{P}=K+ \mathcal{P}_1+K_1$.
        By the induction hypothesis, it holds that 
        \[r\leq \ell+(K_1+1)\cdot \ell_1 + \ellf\left(\mathcal{P}_1,\gamma,\frac{\zeta_0 \alpha}{2}\right)+\ell_3.\]

        \item \textbf{Case 2: \Cref{line:second-recursion} is executed. }
        Let $K_2$ denote the number of times that \Cref{line:compute-p-r2} is executed.
        Let $\mathcal{P}_2$ denote the value returned by the recursive call in  \Cref{line:compute-p-p2}.
        In this case we have $\mathcal{P}=K + \mathcal{P}_2 + K_2 + \lceil \log \frac{1}{1-R}\rceil$.
        By the induction hypothesis, it holds  that 
        \[r \leq \ell+ (K_2+1)\cdot \ell_2 +\ellf\left(\mathcal{P}_2,\gamma, \frac{\zeta_0\alpha (1-R)}{2}\right) +\ell_3.\]

        \item \textbf{Case 3: otherwise. }
        In this case, we have $\mathcal{P}=K$.
        It holds that 
        \[r \leq \ell+\ell_3+1.\]
\end{itemize}

By choosing  $c$ to be a large enough constant, one can verify that $r\le \ell_0+\ellf(\mathcal{P},\gamma,\alpha)$ in all cases.
\end{proof}

\begin{claim}\label{cla:bound-p}
Suppose that \Cref{cond:warm-up-recursive-without-cd} is satisfied by the input of
    \RecursiveSampling$(\vec{Y}; I,\Lambda,\epsilon,\gamma,\delta,\alpha)$,     
    who returns a $\mathcal{P}\ge 0$.
    For any $\eta\in (0,1)$, with probability at least $1-\eta$, it holds that $\mathcal{P}=O\left(\log^2\frac{1}{\eta}\right)$. 
\end{claim}

\begin{proof}

    We prove that,  for some constant $d_1$, it holds that $\Pr[\mathcal{P}> d_1 \cdot \log^2 \frac{1}{\eta}]<\eta$ for any $\eta\in (0,1)$.

    This is proved by induction.
    Suppose that the input is the LLL instance $I=\left(\{X_i\}_{i\in U}, \{A_v\}_{v\in V}\right)$, the random assignment $\vec{Y}=(Y_i)_{i\in U}$, and the subset $\Lambda\subseteq V$ of bad events, along with the parameter  $(\epsilon, \gamma, \delta, \alpha)$ satisfying \Cref{cond:warm-up-recursive-without-cd}.
    
    For the induction basis, when $\Lambda=V$,
    the induction hypothesis holds trivially since $\mathcal{P}=0$.
    
    For the general case.
    by the induction hypothesis, all recursive calls return the correctly bounded values of $\mathcal{P}$'s as long as \Cref{cond:warm-up-recursive-without-cd} is satisfied by the input.
    We prove  that \RecursiveSampling$(\vec{Y};I,\Lambda,\epsilon,\gamma,\delta,\alpha)$ returns a correctly bounded $\mathcal{P}$.

    Let $K$ denote the number of times that \Cref{line:compute-p-estimate} is executed.
    Let $K_1$ denote the number of times that \Cref{line:compute-p-r1} is executed.
    Let $K_2$ denote the number of times that \Cref{line:compute-p-r2} is executed.
    If \Cref{line:compute-p-p1} is executed, let $\mathcal{P}_1$ denote the value returned by the recursive call in \Cref{line:compute-p-p1}.
    If \Cref{line:compute-p-p2} is executed, let $\mathcal{P}_2$ denote the value returned by the recursive call in \Cref{line:compute-p-p2}.
    
    For $\eta\in (0,1)$, let $k\triangleq 10 \log \frac{1}{\eta}+d_2$ for sufficient large constant $d_2$.
    One can verify that if  $\mathcal{P}>d_1\cdot \log^2 \frac{1}{\eta}$, then necessarily at least one of the following events must have happened.
    \begin{itemize}
        \item \textbf{Event $\mathcal{E}_1$:  $K>k$.}
        The probability that \Cref{line:compute-p-estimate} is executed more than $k$ times is $R-L$, where $R$ and $L$ respectively take values of these variables in the $k$-th iteration of the \textbf{while} loop.
        %
        %
        %
        By \Cref{lem:estimate-augmenting}, we have $R-L\leq 4\zeta_0^k$.
        By choosing $\zeta_0$ to be a small enough constant, we have 
        \[\Pr(\mathcal{E}_1) = \Pr(K > k) \leq 4\zeta_0^k \leq 0.01\eta. \]
            
        \item \textbf{Event $\mathcal{E}_2$: \Cref{line:compute-p-p1} is executed and $K_1 \geq k$.}
        We have 
        $\Pr(\mathcal{E}_2)\leq \Pr\left(\forall A\in\Phi: \text{$A$ occurs on $\vec{Y}$}\right)$,
        where
        \[\Phi \triangleq \left\{\aug{\hat{I}}{\frac{1}{2}, \frac{1}{2}, \gamma,\frac{\zeta_0\alpha}{2}}{B_{r}(\Lambda)} ~\middle| ~r=\ell_0(\epsilon,\gamma,\delta)+1+i \cdot \ell_0\left(\frac{1}{2}, \gamma,\frac{\zeta_0\alpha}{2}\right) \text{, for } 0\leq i<k\right\}.
        \]
        Let $X_{\beta}$ and $A_{\kappa}$ denote the respective random variable and bad event constructed in \Cref{lem:substituting} under parameter $\Lambda$, $Y_{\vbl(\Lambda)}$, $(\epsilon,\gamma,\delta)$, $I(B_{\ell(\epsilon,\gamma,\delta)+1}(\Lambda))$.
        Define 
        \begin{equation}\label{eq:event-2-U-V-I}
            \begin{aligned}
            & U' \triangleq U\setminus \vbl(B_{\ell_0(\epsilon,\gamma,\delta)} (\Lambda)), \\
            & V' \triangleq V\setminus B_{\ell_0(\epsilon,\gamma,\delta)+1} (\Lambda), \\
            & I'\triangleq \left(\{X_i\}_{i\in U'}\cup \{X_{\beta}\}, \{A_v\}_{v\in V'} \cup \{A_{\kappa}\}\right).
        \end{aligned}
        \end{equation}
        Let $Y_{\beta}$ be drawn independent according to the marginal distribution $\mu^{Y_{\vbl(\kappa) \setminus\{\beta\}}}_{I',\beta}$.
        This is well-defined because $Y_{\vbl(\kappa) \setminus\{\beta\}}\subseteq U'$.
        Then, define assignment $\vec{Y}'$ as
        \begin{align} \label{eq:event-2-V}
            \vec{Y}'\triangleq Y_{ U'}\land Y_{\beta}.
        \end{align}
        %
        %
        Since $(\vec{Y}, \Lambda)$ satisfies the augmented conditional Gibbs property on $I$ with parameter $(\epsilon,\gamma,\delta)$, we have $Y_{U'}\sim \mu^{Y_{\vbl(\Lambda)}}_{\hat{I}, U'}$,
        where $\hat{I}$ is defined as in \eqref{eq:warm-up-def-hat-I}.
        %
        %
        By \Cref{lem:substituting}, we have $Y_{U'}\sim \mu_{I',U'}$, which means $\vec{Y}'\sim \mu_{I'}$,
        %
        %
        %
        i.e.~for any $A\in \Phi$, we have that $A$ occurs on $\vec{Y}$ if and only if $A$ occurs on $\vec{Y}'$.
        Thus, 
        \[\begin{aligned}
                \Pr(\mathcal{E}_2) &\leq \Pr\left(\forall A\in\Phi: \text{$A$ occurs on $\vec{Y}$}\right) 
                =\Pr\left(\forall A\in\Phi: A\text{ occurs on }\vec{Y}'\right) \\
                &\leq \frac{(\zeta_0 \alpha)^{k}}{(1-\epsilon)(\alpha-\delta)} \leq 4\zeta_0^k\leq 0.01\eta,
        \end{aligned}\]
         by choosing $\zeta_0$ to be a sufficiently small constant.
            
        \item \textbf{Event $\mathcal{E}_3$: \Cref{line:compute-p-p1} is executed and $\mathcal{P}_1\ge d_1\log^2 \frac{1}{1.1\eta}$.}
        We first bound the probability that \Cref{line:compute-p-p1} is executed.
        Recall that $S$ and $T$ are $\epsilon$-correlated in instance $\hat{I}$,
        where $\hat{I}$ is defined as in \eqref{eq:warm-up-def-hat-I}.
        According to \Cref{def:correlation}, it holds with probability $1$ that 
        \[\frac{f(Y_T)}{f_{\max}}\geq \frac{1}{1+\epsilon} \geq \frac{1}{2}. \]
        Thus, the probability that \Cref{line:compute-p-p1} is executed is at most $\frac{1}{2}$.
        By the induction hypothesis, we have $\Pr(\mathcal{P}_1>d_1\cdot \log^2\frac{1}{1.1\eta}\mid\text{\Cref{line:compute-p-p1} is executed})<1.1\eta$.
        Overall, we have
        \[\Pr(\mathcal{E}_3) \leq \frac{1}{2} \cdot 1.1\eta \leq 0.55\eta. \]
            
        \item \textbf{Event $\mathcal{E}_4$: \Cref{line:compute-p-p2} is executed and $K_2\geq k$.}
        The probability of $\mathcal{E}_4$ can be bounded similarly as $\mathcal{E}_2$.
        Let $I'$, $Y'$ be defined as in \eqref{eq:event-2-U-V-I} and \eqref{eq:event-2-V} respectively.
        Let $\hat{I}'$ defined as in \eqref{eq:warm-up-def-hat-I}.
        Define
        \[\Phi' \triangleq \left\{\aug{\hat{I}'}{\frac{1}{2},\gamma, \frac{\zeta_0\alpha (1-R)}{2}}{B_{r}(\Lambda)} ~\middle| ~r=\ell_0(\epsilon,\gamma,\delta)+1+i \cdot \ell_0\left(\frac{1}{2},\gamma, \frac{\zeta_0\alpha (1-R)}{2}\right) \text{, for } 0\leq i<k\right\}\]
        %
        %
        For any $A\in \Phi'$,  we have that $A$ occurs on $\vec{Y}$ if and only if $A$ occurs on $\vec{Y}'$, and
        \[\begin{aligned}
                \Pr(\mathcal{E}_4) 
                &\leq \Pr\left(\forall A\in\Phi': \text{$A$ occurs on $\vec{Y}$}\right) 
                =\Pr\left(\forall A\in\Phi': \text{$A$ occurs on $\vec{Y}'$}\right) \\
                &\leq \frac{(\zeta_0 \alpha)^{k}}{(1-\epsilon)(\alpha-\delta)} \leq 4\zeta_0^k\leq 0.01\eta. 
        \end{aligned}\]
            
        \item \textbf{Event $\mathcal{E}_5$: \Cref{line:compute-p-p2} is executed and $\log \frac{1}{1-R}\geq k$. }
        The probability that \Cref{line:compute-p-p2} is executed given $R$, is precisely $1-R$.
        Thus, we have
        \[\Pr(\mathcal{E}_5) = \Pr\left(\text{ \Cref{line:compute-p-p2} is executed }  ~\middle|~ \log \frac{1}{1-R}\geq k\right) \cdot \Pr\left(\log \frac{1}{1-R}\geq k\right) \leq 2^{-k}\leq 0.01\epsilon. \]
            
        \item \textbf{Event $\mathcal{E}_6$: \Cref{line:compute-p-p2} is executed and $\mathcal{P}_2>d_1\log^2 \frac{1}{1.1\eta}$.}
        The probability that \Cref{line:compute-p-p2} is executed is bounded by $1-R$.
        Recall that $1-R$ is the lower bound of $P\triangleq \Pr_{\vec{X}\sim \mu_I}(\text{ $\vec{X}$ does not avoids $A_{\lambda}$})$.
        Assuming \Cref{cond:warm-up-recursive-without-cd}, this probability $P$ is at most $\frac{\delta}{\alpha}\leq \zeta_0$.
        By the induction hypothesis, we have $\Pr(\mathcal{P}_2>d_1\cdot \log^2\frac{1}{1.1\eta})<1.1\eta$.
        By choosing $\zeta_0$ to be a sufficiently small constant, we have
        \[\Pr(\mathcal{E}_6) \leq \zeta_0 \cdot 1.1\eta \leq 0.01\eta. \]
            
        \end{itemize}

        By the union bound applied to all the above cases, we have 
        \[\Pr\left( \mathcal{P}>d_1\cdot \log^2 \frac{1}{\eta}\right)<\eta,
        \]
        for any $\eta\in (0,1)$.
    This proves \Cref{cla:bound-p}.
    \end{proof}
Combine \Cref{cla:bound-radius-using-p} and \Cref{cla:bound-p}. For any $0<\eta<1$, with probability at least $1-\eta$, 
the following upper bound holds for the radius $r$ of \Cref{alg:recursive-sampling-without-cd}, assuming that its input satisfies \Cref{cond:warm-up-recursive-without-cd}:
    %
    \[
    r\le \ell_0(\epsilon,\gamma,\delta) + O\left(\left(
    \log\frac{1}{\gamma}\cdot \log^4 \frac{1}{\eta} + \log\frac{1}{\gamma}\cdot \log^2 \frac{1}{\eta}\cdot \log \frac{1}{\alpha}\right)\cdot \log \left(\log \frac{1}{\eta}\cdot\log \frac{1}{\gamma}\cdot\log \frac{1}{\alpha} \right)\right) .\] 
This proves  \Cref{lem:item:recursive-sample-complexity} in \Cref{lem:recursive-sample-correctness-complexity}.





        

        

\subsection{Accuracy of estimation (Proof of \Cref{lem:estimate-augmenting})}\label{sec:proof-estimate-augmenting}
We prove  \Cref{lem:estimate-augmenting}, which guarantees the accuracy of the estimation of the probability 
\[
P\triangleq \Pr_{\vec{X}\sim \mu_I}[\,\vec{X}\text{ avoids }A_{\lambda}\,]
\]
from the local information.

%
We first prove the following technical lemma, which follows directly from~\Cref{def:correlation}. 

\begin{lemma}\label{lem:correlation-extension}
Let $I=(\{X_i\}_{i\in U}$, $\{A_v\}_{v\in V})$ be a LLL instance. 
For any $\epsilon>0$, any disjoint  $S,T\subset U$ with $S \cup T\neq U$, any real functions:
\[p_1,p_2:\ \Sigma_{S}\to \mathbb{R}_{\ge 0}, \quad\quad\quad q_1,q_2:\ \Sigma_{T}\to\mathbb{R}_{\ge 0},\]
if $S$ and $T$ are $\epsilon$-correlated in $I$, it holds that 
\[
\begin{aligned}
    &\sum_{\sigma\in \Sigma_{S},\tau\in\Sigma_{T}}p_1(\sigma)\cdot q_1(\tau)\cdot \nu\left  (\Omega_I^{\sigma\wedge \tau}\right)  \cdot \sum_{\sigma\in \Sigma_{S},\tau\in\Sigma_{T}}p_2(\sigma)\cdot q_2(\tau) \cdot \nu\left  (\Omega_I^{\sigma\wedge \tau}\right)\\
    \leq  (1+\epsilon)\cdot  &\sum_{\sigma\in \Sigma_{S},\tau\in\Sigma_{T}}p_1(\sigma)\cdot q_2(\tau) \cdot \nu\left  (\Omega_I^{\sigma\wedge \tau}\right)\cdot \sum_{\sigma\in \Sigma_{S},\tau\in\Sigma_{T}}p_2(\sigma)\cdot q_1(\tau) \cdot \nu\left  (\Omega_I^{\sigma\wedge \tau}\right).
\end{aligned}\]
\end{lemma}

\begin{proof}[Proof of \Cref{lem:correlation-extension}]

    The lemma follows by verifying:
\begin{align*}
    &\sum_{\sigma\in \Sigma_{S},\tau\in\Sigma_{T}}p_1(\sigma)\cdot q_1(\tau)\cdot \nu\left  (\Omega_I^{\sigma\wedge \tau}\right)  \cdot \sum_{\sigma\in \Sigma_{S},\tau\in\Sigma_{T}}p_2(\sigma)\cdot q_2(\tau) \cdot \nu\left  (\Omega_I^{\sigma\wedge \tau}\right)\\
    =& \sum_{\sigma_1,\sigma_2\in \Sigma_S, \tau_1, \tau_2 \in \Sigma_T} p_1(\sigma_1) \cdot q_1(\tau_1) \cdot  p_2(\sigma_2) \cdot q_2(\tau_2)  \cdot \nu\left(\Omega_I^{\sigma_1\wedge \tau_1}\right) \cdot \nu\left(\Omega_I^{\sigma_2\wedge \tau_2}\right) \\
    \leq& \sum_{\sigma_1,\sigma_2\in \Sigma_S, \tau_1, \tau_2 \in \Sigma_T} p_1(\sigma_1) \cdot q_1(\tau_1) \cdot  p_2(\sigma_2) \cdot q_2(\tau_2)  \cdot  (1+\epsilon) \cdot \nu\left(\Omega_I^{\sigma_1\wedge \tau_2}\right) \cdot \nu\left(\Omega_I^{\sigma_2\wedge \tau_1}\right)\\
    \leq & (1+\epsilon)\cdot  \sum_{\sigma\in \Sigma_{S},\tau\in\Sigma_{T}}p_1(\sigma)\cdot q_2(\tau) \cdot \nu\left  (\Omega_I^{\sigma\wedge \tau}\right)\cdot \sum_{\sigma\in \Sigma_{S},\tau\in\Sigma_{T}}p_2(\sigma)\cdot q_1(\tau) \cdot \nu\left  (\Omega_I^{\sigma\wedge \tau}\right).\qedhere
\end{align*}    
\end{proof}

Now we can prove \Cref{lem:estimate-augmenting}. 

\begin{proof}[Proof of \Cref{lem:estimate-augmenting}]

    
        


    For any $\sigma\in \Sigma_{\vbl(\Lambda)}$, define
        \[\phi(\sigma)\triangleq \begin{cases}
            1 & \text{$\sigma$ avoids the bad event $A_{\lambda}$, } \\
            0 & \text{$\sigma$ does not avoid the bad event $A_{\lambda}$. }
        \end{cases}\]
    For any $v\in V$, $\vbl(v)\subseteq W \subseteq V$, $\sigma\in \Sigma_W$, define
        \[\phi_v(\sigma)\triangleq \begin{cases}
            1 & \text{$\sigma$ avoids the bad event $A_v$, } \\
            0 & \text{$\sigma$ does not avoid the bad event $A_v$. }
        \end{cases}\]
    Let  $S=\vbl(\Lambda)$ and  $T=U\setminus\vbl(B_{\ell}(\Lambda))$.
    For any  $\sigma\in \Sigma_{S}$, define
    \[p_1(\sigma) \triangleq \prod_{\vbl(v)\subseteq S} \phi_v(\sigma) , \quad p_2(\sigma) \triangleq p_1(\sigma) \cdot \phi(\sigma). \]
    For any $\tau\in \Sigma_T$, define
    \[q_1(\tau) \triangleq  \prod_{\vbl(v)\subseteq T } \phi_v(\tau), \quad q_2(\tau) \triangleq 1. \]
    We considered a new bad event $A_{w}$ constructed by \Cref{lem:augmenting} and a augmented LLL instance $I_w$:
    \[    A_{w} \triangleq \aug{I}{\epsilon^k, \alpha_2, \alpha_1\cdot \epsilon^k}{\Lambda}, \text{ where }w\not\in V,
    \quad\text{ and }\quad 
    I_w=(\{X_i\}_{i\in U},\{A_v\}_{v\in V\cup\{w\}} ).\]
    %
    %
    According to \Cref{lem:augmenting}, we have $S$ and $T$ are $\epsilon^k$-correlated in  LLL instance $I_w$. 
    By \Cref{lem:correlation-extension}, we have 
    \begin{equation}
        \label{eq:estimate-epsilon-correlated}
        \begin{aligned}
            &\sum_{\sigma\in \Sigma_{S},\tau\in\Sigma_{T}}p_1(\sigma)\cdot q_1(\tau)\cdot \nu\left (\Omega_{I_w}^{\sigma\wedge \tau}\right)  \cdot \sum_{\sigma\in \Sigma_{S},\tau\in\Sigma_{T}}p_2(\sigma)\cdot q_2(\tau) \cdot \nu\left  (\Omega_{I_w}^{\sigma\wedge \tau}\right)\\
            \leq  (1+\epsilon^k)\cdot  &\sum_{\sigma\in \Sigma_{S},\tau\in\Sigma_{T}}p_1(\sigma)\cdot q_2(\tau) \cdot \nu\left  (\Omega_{I_w}^{\sigma\wedge \tau}\right)\cdot \sum_{\sigma\in \Sigma_{S},\tau\in\Sigma_{T}}p_2(\sigma)\cdot q_1(\tau) \cdot \nu\left  (\Omega_{I_w}^{\sigma\wedge \tau}\right).
        \end{aligned}
    \end{equation}
    Let 
    \[P\triangleq\Pr_{\vec{X}\sim \mu_{I}}[\text{ $\vec{X}$ avoids $A_{\lambda}$ }]     , 
    \quad\text{ and }\quad P_w\triangleq\Pr_{\vec{X}\sim \mu_{I_w}}[\text{ $\vec{X}$ avoids $A_{\lambda}$}].\]
    It holds that 
    \[P_w=\frac{\sum_{\sigma\in \Sigma_{S},\tau\in\Sigma_{T}}p_2(\sigma)\cdot q_1(\tau) \cdot \nu\left  (\Omega_{I_w}^{\sigma\wedge \tau}\right)}{\sum_{\sigma\in \Sigma_{S},\tau\in\Sigma_{T}}p_1(\sigma)\cdot q_1(\tau)\cdot \nu\left  (\Omega_{I_w}^{\sigma\wedge \tau}\right)}. \]
    Then, define
    \[\hat{P}\triangleq \frac{\sum_{\sigma\in \Sigma_{S},\tau\in\Sigma_{T}}p_2(\sigma)\cdot q_2(\tau) \cdot \nu\left  (\Omega_{I_w}^{\sigma\wedge \tau}\right)}{\sum_{\sigma\in \Sigma_{S},\tau\in\Sigma_{T}}p_1(\sigma)\cdot q_2(\tau) \cdot \nu\left  (\Omega_{I_w}^{\sigma\wedge \tau}\right)}. \]
    The above $P_w$ and $\hat{P}$ are well-defined, since their denominators are non-zero, i.e.
    %
    \[\begin{aligned}
        \sum_{\sigma\in \Sigma_{S},\tau\in\Sigma_{T}}p_1(\sigma)\cdot q_2(\tau) \cdot \nu\left  (\Omega_{I_w}^{\sigma\wedge \tau}\right) &\geq \sum_{\sigma\in \Sigma_{S},\tau\in\Sigma_{T}}p_1(\sigma)\cdot q_1(\tau)\cdot \nu\left (\Omega_{I_w}^{\sigma\wedge \tau}\right) \\
        &= \nu(\Omega_{I_w}) \geq \alpha_2-\alpha_1\cdot \epsilon^k>0.
    \end{aligned}\]
    Obviously, the value of $\hat{P}$ can be determined by $\Lambda$, $A_{\lambda}$ and $I(B_{\ell+1}(\Lambda))$. 
    By \Cref{eq:estimate-epsilon-correlated} and simply swap the definition of $q_1$ with $q_2$, we directly have 
    \begin{align}\label{eq:estimate-1}
        \hat{P} \leq \left(1+\epsilon^k\right) \cdot P_w \leq \left(1+\epsilon^k\right)^2 \cdot \hat{P}.
    \end{align}
    %
    %
    %
    Now we try to estimate the difference between $P$ and $P_w$. 
    First, we have 
    \begin{equation}\label{eq:estimate-2}
    \begin{aligned}
        P & = \frac{\sum_{\sigma\in \Sigma_S}  p_2(\sigma) \cdot \nu(\Omega^{\sigma}_{I}) }{\nu\left(\Omega_{I}\right)}
        \leq \frac{\sum_{\sigma\in \Sigma_S}  p_2(\sigma)\cdot (\nu(\Omega^{\sigma}_{I_w}) + \nu(A_{w}) \cdot \nu_S(\sigma)) }{\nu(\Omega_{I})}\\
        &\leq \frac{\sum_{\sigma\in \Sigma_S}  p_2(\sigma)\cdot \nu(\Omega^{\sigma}_{I_w}) }{\nu(\Omega_{I_w})} +  \frac{\sum_{\sigma\in \Sigma_S} p_2(\sigma) \cdot  \nu(A_w)\cdot \nu_S(\sigma)}{\nu(\Omega_{I})}
        \leq P_w + \frac{\alpha_1\cdot \epsilon^k}{\alpha_1} =  P_w + \epsilon^k.
    \end{aligned}
    \end{equation}
    On the other hand, we have 
    \begin{equation}\label{eq:estimate-3}
        \begin{aligned}
        P &= \frac{\sum_{\sigma\in \Sigma_S}  p_2(\sigma)\cdot \nu(\Omega^{\sigma}_{I}) }{\nu(\Omega_{I})}
        \geq \frac{\sum_{\sigma\in \Sigma_S}  p_2(\sigma)\cdot \nu(\Omega^{\sigma}_{I_w}) }{\nu(\Omega_{I_w})} \cdot \frac{\nu(\Omega_{I_w})}{\nu(\Omega_{I})}\\
        & \geq \frac{\sum_{\sigma\in \Sigma_S}  p_2(\sigma) \cdot \nu(\Omega^{\sigma}_{I_w}) }{\nu(\Omega_{I_w})} \cdot \frac{\nu(\Omega_{I})-\nu(A_w)}{\nu(\Omega_{I})}
        \geq P_w \cdot \left(1-\frac{\alpha_1\cdot \epsilon^k}{\alpha_1}\right) = P_w\cdot (1-\epsilon^k).
    \end{aligned}
    \end{equation}
    By combining \eqref{eq:estimate-1} \eqref{eq:estimate-2} \eqref{eq:estimate-3}, we have 
    \[\hat{P} - 2 \epsilon^k \leq  P \leq  \hat{P}+2\epsilon^k.\qedhere
    \]
\end{proof}

\subsection{Analysis of \emph{Substituting} (Proof of \Cref{lem:substituting})}\label{sec:proof-substituting}

We prove \Cref{lem:substituting}, which guarantees the soundness of the substituting trick. 
We  first construct the new random variable $X_{\beta}$ and bad event $A_{\kappa}$, 
and then prove the identity stated in \Cref{lem:substituting}.

\paragraph{The construction of $X_{\beta}$. }
We construct a new random variable $X_{\beta}$ follows a distribution $\nu_{\beta}$ over a domain $\Sigma_{\beta}$. 
Here is the construction of $\Sigma_{\beta}$ and $\nu_{\beta}$ .
\begin{itemize}
    
    \item The domain $\Sigma_{\beta}$. 
    Let $R=\vbl(B_{\ell+1}(\Lambda)) \setminus \vbl(B_{\ell}(\Lambda))$. 
    We set $\Sigma_{\beta}=\Sigma_{R}$. 
        
    \item The distribution $\nu_{\beta}$.
    Let $S=\vbl(\Lambda)$ and $T=U\setminus \vbl(B_{\ell} (\Lambda))$.
    Let $\omega \in \Sigma_{T \setminus R}$ be an arbitrary configuration.
    For any $\pi\in \Sigma_{\beta}$, define
    \[P(\pi) \triangleq \frac{\nu\left( \Omega_{\hat{I}}^{\omega \wedge \pi \wedge \sigma}\right)}{\nu_T(\omega \wedge \pi) \cdot \nu_{S}(\sigma)}. \]
    %
    %
    Note that $P(\pi)$ does not depend on $\omega$. 
    Order all the configurations $\pi\in \Sigma_{\beta}$ increasingly by $P(\pi)$ as $\pi_1, \pi_2, ..., \pi_{|\Sigma_{\beta}|}$. 
    Then, for $1\leq i\leq |\Sigma_{\beta}|$, we set
    \[\nu_{\beta}(\pi_i) = 
        \begin{cases}
            \frac{P\left(\pi_{1}\right)}{P\left(\pi_{|\Sigma_{\beta}|}\right)} &  i=1, \\
            \frac{P\left(\pi_{i}\right) - P\left(\pi_{i-1}\right)}{P\left(\pi_{|\Sigma_{\beta}|} \right)} & 1<i\leq |\Sigma_{\beta}|.
        \end{cases}
    \]
\end{itemize}

\paragraph{The construction of $A_{\kappa}$. } 
We construct the new bad event $A_{\kappa}$ as follows. The event $A_{\kappa}$ is defined on the variables in $\vbl(\kappa) \triangleq R\cup \{\beta\}$ such that  for any $\tau\in \Sigma_{\vbl(\kappa)}$,
\[
A_{\kappa}\text{ occurs if and only if } P(\tau_{R})< P(\tau_{\beta}).
\]
%

\begin{proof}[Proof of \Cref{lem:substituting}]
We prove the claimed properties one by one. 

To prove \Cref{lem:substituting-item-1} in \Cref{lem:substituting}, it is easy to verify that the constructions of $X_{\beta}$ and $A_{\kappa}$ depend only on the specifications of $\Lambda$, $\sigma$, $(\epsilon, \gamma,\delta)$  and  $I(B_{\ell+1}(\Lambda))$.

Next, we prove \Cref{lem:substituting-item-2} in \Cref{lem:substituting}. Fix any $W\subseteq T$ and its complement $\overline{W}= T\setminus W$. 
Define the following set of assignments:
\[\Sigma_{\overline{W}}' \triangleq \{\overline{w}\in \Sigma_{\overline{W}} \mid \text{$\overline{\omega}\wedge \omega$ avoids all the bad events $A_v$ such that $\vbl(v)\subseteq T$ and $\vbl(v)\nsubseteq W$}\}.\]
Let $U_{\sigma}\triangleq (U\setminus \vbl(B_{\ell}(\Lambda)))\cup \{\beta\}$.
For any $\omega\in \Sigma_{W}$ and $\overline{\omega}\in \Sigma_{\overline{W}}$, it holds that 
\[\mu^{\omega \wedge \sigma}_{\hat{I},\overline{W}} (\overline{\omega}) =\begin{cases}
    \frac{\nu\left(\Omega_{\hat{I}}^{\sigma \wedge \omega \wedge \overline{\omega}}\right)}{\nu\left(\Omega_{\hat{I}}^{\sigma \wedge \omega }\right)} & \text{if $\nu\left(\Omega_{\hat{I}}^{\sigma \wedge \omega }\right)>0$ and $ \overline{\omega}\in \Sigma_{\overline{W}}' $ }, \\    
    0 & \text{otherwise. } \\
\end{cases}\]
and
\[\mu^{\omega}_{I_{\sigma},\overline{W}} (\overline{\omega}) =\begin{cases}
    \frac{\nu_{U_{\sigma}}\left(\Omega_{I_{\sigma}}^{\omega \wedge \overline{\omega}}\right)}{\nu_{U_{\sigma}}\left(\Omega_{I_{\sigma}}^{ \omega }\right)} & \text{if $\nu_{U_{\sigma}}\left(\Omega_{I_{\sigma}}^{ \omega }\right)>0$ and  $\overline{\omega}\in \Sigma_{\overline{W}}'$},\\

    0 & \text{otherwise}. 
    
\end{cases}\]
We consider the following  cases:
\begin{itemize}
    \item \textbf{Case 1: $\overline{\omega}\notin \Sigma_{\overline{W}}'$.} 
    In this case, it holds that $\mu^{\omega \wedge \sigma}_{\hat{I},\overline{W}} (\overline{\omega}) = \mu^{\omega}_{I_{\sigma},\overline{W}} (\overline{\omega})=0$. 

    \item \textbf{Case 2: $\overline{\omega}\in \Sigma_{\overline{W}}'$ and  $\nu\left(\Omega_{\hat{I}}^{\sigma \wedge \omega }\right)=0$.} 
    In this case, we have 
    \[\nu_{U_{\sigma}}\left(\Omega_{I_{\sigma}}^{ \omega }\right)
    =\sum_{\overline{\omega}\in \Sigma_{\overline{W}}'} \nu_{U_{\sigma}}\left(\Omega_{I_{\sigma}}^{\omega \wedge \overline{\omega}} \right)
    =\sum_{\overline{\omega}\in \Sigma_{\overline{W}}'} \nu\left(\Omega_{\hat{I}}^{\sigma \wedge \omega \wedge \overline{\omega}} \right)
    =\nu\left(\Omega_{\hat{I}}^{\sigma \wedge \omega }\right)=0.\]
    Therefore, it holds that $\mu^{\omega \wedge \sigma}_{\hat{I},\overline{W}} (\overline{\omega}) = \mu^{\omega}_{I_{\sigma},\overline{W}} (\overline{\omega})=0$.

    \item \textbf{Case 3: $\overline{\omega}\in \Sigma_{\overline{W}}'$ and  $\nu\left(\Omega_{I_{\sigma}}^{ \omega }\right)>0$.}
    In this case, it holds that
    \[\mu^{\omega \wedge \sigma}_{\hat{I},\overline{W}} (\overline{\omega}) 
    = \frac{\nu\left(\Omega_{\hat{I}}^{\sigma \wedge \omega \wedge \overline{\omega}}\right)}{\nu\left(\Omega_{\hat{I}}^{\sigma \wedge \omega }\right)} 
    = \frac{\nu\left(\Omega_{\hat{I}}^{\sigma \wedge \omega \wedge \overline{\omega}}\right)}{\sum_{\overline{\omega}\in \Sigma_{\overline{W}}'} \nu\left(\Omega_{\hat{I}}^{\sigma \wedge \omega \wedge \overline{\omega}} \right)}\]
    and
    \[\begin{aligned}
    \mu^{\omega}_{I_{\sigma},\overline{W}} (\overline{\omega}) 
    &=  \frac{\nu_{U_{\sigma}}\left(\Omega_{I_{\sigma}}^{\omega \wedge \overline{\omega}}\right)}{\nu_{U_{\sigma}}\left(\Omega_{I_{\sigma}}^{ \omega }\right)} 
    = \frac{\nu_{U_{\sigma}}\left(\Omega_{I_{\sigma}}^{\omega \wedge \overline{\omega}}\right)}{\sum_{\overline{\omega}\in \Sigma_{\overline{W}}'} \nu_{U_{\sigma}}\left(\Omega_{I_{\sigma}}^{ \omega \wedge \overline{\omega}}\right)} 
    = \frac{\nu_T(\omega \wedge \overline{\omega}) \cdot \nu_{\beta} ( (\omega \wedge \overline{\omega})_R)}{\sum_{\overline{\omega}\in \Sigma_{\overline{W}}'} \nu_T(\omega \wedge \overline{\omega}) \cdot \nu_{\beta} ( (\omega \wedge \overline{\omega})_R)} \\ 
    &= \frac{\nu_T(\omega \wedge \overline{\omega}) \cdot P( (\omega \wedge \overline{\omega})_R)}{\sum_{\overline{\omega}\in \Sigma_{\overline{W}}'} \nu_T(\omega \wedge \overline{\omega}) \cdot P ( (\omega \wedge \overline{\omega})_R)} 
    = \frac{\nu\left(\Omega_{\hat{I}}^{\sigma \wedge \omega \wedge \overline{\omega}}\right)}{\sum_{\overline{\omega}\in \Sigma_{\overline{W}}'} \nu\left(\Omega_{\hat{I}}^{\sigma \wedge \omega \wedge \overline{\omega}} \right)}
    \end{aligned}\]
    Thus, we have $\mu^{\omega \wedge \sigma}_{\hat{I},\overline{W}} (\overline{\omega}) = \mu^{\omega}_{I_{\sigma},\overline{W}} (\overline{\omega})$.
\end{itemize}

Combing three cases, we prove that  
\[\forall \omega\in \Sigma_{W}, \overline{\omega}\in \Sigma_{\overline{W}}: \quad \mu^{\omega \wedge \sigma}_{\hat{I},\overline{W}} (\overline{\omega}) = \mu^{\omega}_{I_{\sigma},\overline{W}} (\overline{\omega}).\]

At last, we prove \Cref{lem:substituting-item-3} in \Cref{lem:substituting}.
As before, for any $v\in V$, $\omega\in \Sigma_W$, where $\vbl(v)\subseteq W \subseteq V$, let $\phi_v(\omega)$ indicate whether $\omega$ avoids  $A_v$.
Let $S=\vbl(\Lambda)$ and $T=U\setminus \vbl(B_{\ell}(\Lambda))$.
%
%
For any $x\in \Sigma_S$, define    
\[
        p_1(x)\triangleq  \prod_{ \vbl(v)\subseteq S} \phi_v(x)
        \quad   \text{ and }\quad
        p_2(x) \triangleq  
        \begin{cases}
            \frac{1}{\nu_S(x)} &  x=\sigma , \\
            0 & x\neq \sigma.
        \end{cases}
\]
%
%
For any $y\in \Sigma_T$, define
\[      q_1(y)\triangleq  \prod_{\vbl(v)\subseteq T} \phi_v(y)
        \quad   \text{ and }\quad
        q_2(y) \triangleq 
        \begin{cases}
            \frac{1}{\nu_T(y)} &  y_R=\pi_{|\Sigma_{\beta}|} , \\
            0 & y_R\neq\pi_{|\Sigma_{\beta}|}.
        \end{cases}
\]
It holds that
    \[\begin{aligned}
        \nu_{U_{\sigma}}(\Omega_{I_{\sigma}}) &= \sum_{\tau\in \Sigma_T} \nu_{T}(\tau) \cdot q_1(\tau) \cdot \frac{P(\tau_{R})}{P(\pi_{|\Sigma_{\beta}|})} = \frac{\sum_{\sigma\in \Sigma_S, \tau\in \Sigma_T}  p_2(\sigma) \cdot q_1(\tau)\cdot \nu(\Omega_{\hat{I}}^{\sigma\wedge\tau})}{\sum_{\sigma\in \Sigma_S, \tau\in \Sigma_T}  p_2(\sigma) \cdot q_2(\tau)\cdot \nu(\Omega_{\hat{I}}^{\sigma\wedge\tau})}. 
    \end{aligned} \]
    and
    \[\begin{aligned}
        \nu(\Omega_{\hat{I}}) &\leq \frac{ \nu_{U}(\Omega_{\hat{I}})}{\sum_{\sigma\in \Sigma_S, \tau\in \Sigma_T}  p_1(\sigma) \cdot q_2(\tau)\cdot \nu(\Omega_{\hat{I}}^{\sigma\wedge\tau})} = \frac{\sum_{\sigma\in \Sigma_S, \tau\in \Sigma_T}  p_1(\sigma) \cdot q_1(\tau)\cdot \nu(\Omega_{\hat{I}}^{\sigma\wedge\tau})}{\sum_{\sigma\in \Sigma_S, \tau\in \Sigma_T}  p_1(\sigma) \cdot q_2(\tau)\cdot \nu(\Omega_{\hat{I}}^{\sigma\wedge\tau})}.
    \end{aligned}\]
%

%
By \Cref{lem:correlation-extension}, we have $\nu_{U}(\Omega_{\hat{I}})\leq (1+\epsilon) \cdot \nu_{U_{\sigma}}(\Omega_{I_{\sigma}})$, which implies $\nu_{U_{\sigma}}(\Omega_{I_{\sigma}})\geq (1-\epsilon) \cdot \nu_{U}(\Omega_{\hat{I}})$. Recall that the original LLL instance $I$ is $\alpha$-satisfiable and the additional bad event occurs with probability no more than $\delta$. Therefore, the LLL instance $I_{\sigma}$ is at least $(1-\epsilon)(\alpha-\delta)$-satisfiable. 
\end{proof}

\subsection{Wrapping up the analysis of \emph{Resampling} (Proof of \Cref{lem:resample-correctness-complexity})}\label{sec:proof-resample-correctness-complexity}
At last, we wrap up the analysis of the \emph{Resampling} phase and prove \Cref{lem:resample-correctness-complexity}.

\begin{proof}[Proof of \Cref{lem:item-resample-correctness} of \Cref{lem:resample-correctness-complexity}]
Denote by $v_1,v_2,...,v_{n}$ the sequence of nodes obtained by sorting $V$ in ascending order according to IDs.
For $1\leq k\leq n$,  denote by $\RV{B}_k$ the ball associated to node $v_k$, i.e.
\[ \RV{B}_k\triangleq \begin{cases}
     B_{r_{v_k}}(p_{v_k}) & \text{if }v_k\in \RV{R}\text{ and }p_{v_k},r_{v_k}\not\in\{\perp\}, \\
    \emptyset& \text{otherwise.}
\end{cases}\]
%
Fix arbitrary $B_1,B_2,...,B_n\subseteq V$ and assignments $\sigma_1\in \Sigma_{\vbl(B_1)}, \sigma_2\in \Sigma_{\vbl(B_2)}$, ..., $\sigma_n\in \Sigma_{\vbl(B_n)}$ satisfying 
\[\Pr\left[\left(\bigwedge_{k=1}^{n} \RV{B}_k=B_k\right) \wedge \left(\bigwedge_{k=1}^{n} Y_{\vbl(B_k)}=\sigma_k\right) \right]>0.\]
Assume that the input $\vec{Y}$ and $p_v,r_v$ for $v\in\RV{R}$ satisfy the properties asserted by~\Cref{lem:clustering-correctness}.
Our goal is to show that conditioned on $\RV{B}_k=B_k$ and $Y_{\vbl(B_k)}=\sigma_k$ for all $1\leq k\leq n$, right after \Cref{alg:resampling} being sequentially executed on all nodes in $\RV{R}$, the output assignment $\vec{Y}$  follows the distribution $\mu_I$.

%
For $0\leq k\leq n$, let $\Vec{Y}^{(k)}$ denote the random assignment $\Vec{Y}$ right after \Cref{alg:resampling} terminates on all the active nodes $v_j$ with $j\leq k$.
For $0\leq k \leq n$, define
\[N_k\triangleq \{j \mid  k<j\leq n \land B_j\neq \emptyset \}. \]
%
%
%
%
Recall the choice of parameter $(\epsilon_0,\gamma_0,\delta_0)$ in \eqref{eq:def-parameter-without-cd-expected-complexity} and let $\ell\triangleq \ell_0(\epsilon_0,\gamma_0,\delta_0)$.
In the rest of the proof, we show that $(\vec{Y}^{(k)}, \{B_i\}_{i\in N_k})$ satisfies the clustered conditional Gibbs property on $I$ with parameter $(\epsilon_0,\gamma_0,\delta_0)$ for any $0\leq k\leq n$.
In particular, this guarantees that  $(\vec{Y}^{(n)}, \emptyset)$ satisfies the clustered conditional Gibbs property, which implies that the final output $\vec{Y}^{(n)}$ follows the distribution $\mu_I$.
%

For $k\in N_{0}$, let $A_{\lambda_k}\triangleq \aug{I}{\epsilon_0,\gamma_0,\delta_0}{B_k}$.
For $0\leq k\leq  n$, define
\[\begin{aligned}
    & \overline{U}^{(k)} \triangleq U\setminus \bigcup_{j\in N_k} \vbl(B_{\ell_0(\epsilon_0,\gamma_0,\delta_0)} (B_j)), \quad \text{and} \quad \overline{I}^{(k)} \triangleq  \left(\{X_i\}_{i\in U}, \{A_v\}_{v\in V}\cup \{A_{\lambda_j}\}_{j \in N_k}\right). 
\end{aligned}\]
The goal is to prove that $Y^{(k)}_{\overline{U}^{(k)}}$ follows the distribution $\mu^{\bigwedge_{j\in N_k} \sigma_j}_{\overline{I}^{(k)}, \overline{U}^{(k)}}$ for every $0\leq k\leq n$.

We prove this by induction on $k$. 
Recall that we assume the input satisfies the properties asserted by \Cref{lem:clustering-correctness}. According to this, it holds that $(Y^{(0)}, \{B_i\}_{i\in N_0})$ satisfies \Cref{cond:clustering}.
Thus, $(\vec{Y}^{(0)}, \{B_i\}_{i\in N_0})$ satisfies the clustered conditional Gibbs property on $I$ with parameter $(\epsilon_0,\gamma_0,\delta_0)$. 
This becomes the basis of the induction.

Then, we consider the general case. 
By induction hypothesis, assume that $(\vec{Y}^{(k-1)},  \{B_i\}_{i\in N_{k-1}})$ satisfy the clustered conditional Gibbs property on $I$ with parameter $(\epsilon_0,\gamma_0,\delta_0)$ for all $1\leq k\leq n$.
This implies
\begin{align} \label{eq:resample-induction-1}
    Y^{(k-1)}_{\overline{U}^{(k-1)}} \sim \mu^{\bigwedge_{j\in N_{k-1}} \sigma_j}_{\overline{I}^{(k-1)}, \overline{U}^{(k-1)}} \text{ for all }1\le k\le n.
\end{align}
%
%
If $B_k=\emptyset$,  then $(\vec{Y}^{(k)},  \{B_i\}_{i\in N_k})=(\vec{Y}^{(k-1)},  \{B_i\}_{i\in N_{k-1}})$ and the induction follows directly.
If $B_k\neq \emptyset$, then $v_k$ must be an active node in $\RV{R}$ on which \Cref{alg:resampling} is executed. 
We first verify that the input to \Cref{line:callrecursivesampling} of \Cref{alg:resampling} executed on $v_k$, denoted by $(\vec{Y}'; I', B_k,\epsilon_0,\gamma_0,\delta_0,\gamma_0)$, satisfies \Cref{cond:warm-up-recursive-without-cd}.
%

%
%


For $j\in N_0$, let $X_{\beta_j}$ and $A_{\kappa_j}$ denote the respective random variable and bad event constructed in \Cref{lem:substituting} under parameter $B_j$, $\sigma_j$, $(\epsilon_0,\gamma_0,\delta_0)$, $I(B_{\ell+1}(B_j))$.

For $0 \leq s \leq t\leq n$, define
\[\begin{aligned}
    &U^{(s,t)} \triangleq U\setminus \bigcup_{i\in N_{s}\setminus N_t} \vbl(B_{\ell} (B_i)),  \\
    &V^{(s,t)} \triangleq V\setminus \bigcup_{i\in N_s\setminus N_t}
    B_{\ell+1} (B_i), \\
    &I^{(s,t)} \triangleq \left(\{X_i\}_{i\in U^{(s,t)}}\cup \{X_{\beta_i}\}_{i\in N_s\setminus N_t}, \{A_v\}_{v\in V^{(s,t)}}\cup \{A_{\kappa_i}\}_{i\in N_s\setminus N_t}\right).
\end{aligned}\]
Observe that $I^{(s,s)}=I$ for all $0\leq s\leq n$.
%
%
Observe that, for $0\leq s< t\leq n$, if $B_t=\emptyset$, then $I^{(s,t)}=I^{(s,t-1)}$. And if $B_t\neq \emptyset$, then $I^{(s,t)}$ is obtained from $I^{(s,t-1)}$ by substituting the ball $B_{\ell+1}(B_t)$ with $A_{\lambda_t}$ and substituting $\vbl(B_{\ell}(B_t))$ with $X_{\beta_t}$.
Recall that we assume the input satisfies the properties asserted by \Cref{lem:clustering-correctness}. According to this, it holds that $(Y^{(0)}, \{B_i\}_{i\in N_0})$ satisfies \Cref{cond:clustering}. Thus, for any distinct $i,j\in N_0$, the distance between $B_i$ and $B_j$ is at least $2(\ell+2)$.
One can verify that $I^{(k,n)}=I'$.

For $k\leq j\leq n$, define
\[
\vec{Y}^{(k,j)}\triangleq Y^{(k-1)}_{U^{(k,j)}} \wedge (Y'_{\beta_j})_{j\in N_k\setminus N_j}.
\]
One can verify that $\vec{Y}^{(k,n)}=\vec{Y}'$.

%
%

%

%

%
%
Then, we verify that the input of \RecursiveSampling$(\vec{Y}'; I', B_k ,\epsilon_0,\gamma_0,\delta_0,\gamma_0)$ satisfies \Cref{cond:warm-up-recursive-without-cd}.
The first property can be verified easily. 
For the second property, according to \Cref{lem:substituting}, it holds that $\nu(\Omega_{I^{(k,j)}})\geq (1-\epsilon_0)(\nu(\Omega_{I^{(k,j-1)}}) -\delta_0)$ for all $k<j\leq n$.
Note that the LLL instance $I^{(k,k)}$ is the original LLL instance $I$, which is $\gamma$-satisfiable.
Thus, it holds that 
\[
\begin{aligned}
    \nu(\Omega_{I^{(k,n)}}) &\geq (1-\epsilon_0)^n \cdot \nu(\Omega_{I^{(k,k)}}) - n\cdot \delta_0 \\
    & \geq  \left(1-\frac{1}{2n^3 }\right)^n\nu(\Omega_{I^{(k,k)}})-n\cdot \frac{\zeta_0\cdot \gamma}{24 n^3}\geq \frac{\gamma}{8} \geq \gamma_0 ,
\end{aligned}
\]
which means that $I^{(k,n)}(V\setminus B_k)$ is also $\gamma_0$-satisfiable.
It remains to verify the third property, which states that the $(\vec{Y}^{(k,n)}, B_k)$ satisfies the augmented conditional Gibbs property on instance $I^{(k,n)}$ with parameter $(\epsilon_0,\gamma_0, \delta_0)$.
This is proved by the following claim.

\begin{claim}\label{cla:resample-induction-2}
    For $k\leq j\leq n$, it holds that $(\vec{Y}^{(k,j)}, \{B_k\}\cup \{B_i\}_{i\in N_j})$ satisfies the clustered conditional Gibbs property on $I^{(k,j)}$ with parameter $(\epsilon_0,\gamma_0,\delta_0)$.
\end{claim}

\begin{proof}

For $k\leq j\leq n$, define
\[\begin{aligned}
    & \hat{U}^{(k,j)} \triangleq U^{(k-1,n)} \cup \{\beta_i\}_{i\in N_k\setminus N_j}, \\ 
    & \hat{I}^{(k,j)} \triangleq  \left(\{X_i\}_{i\in U^{(k,j)}}\cup \{X_{\beta_i}\}_{i\in N_k\setminus N_j}, \{A_v\}_{v\in V^{(k,j)}}\cup \{A_{\kappa_i}\}_{i\in N_k\setminus N_j} \cup \{A_{\lambda_i}\}_{i\in\{k\}\cup N_j}\right). 
\end{aligned}\]

%
%
The goal is to prove that $Y^{(k,j)}_{\hat{U}^{(k,j)}}$ follows the distribution $\mu^{\sigma_k \wedge \bigwedge_{i\in N_j} \sigma_i}_{\hat{I}^{(k,j)}, \hat{U}^{(k,j)}}$  for $k\leq j\leq n$. 

This is proved by induction on $j$.
For the induction basis, assuming $j=k$, in this case we have $(\vec{Y}^{(k,k)}, \{B_k\}\cup \{B_i\}_{i\in N_k})=(\vec{Y}^{(k-1)}, \{B_i\}_{i\in N_{k-1}})$ and $I^{(k,k)}=I$.
By~\eqref{eq:resample-induction-1}, we have 
$$Y^{(k,k)}_{\overline{U}^{(k-1)}}\sim \mu^{\bigwedge_{i\in N_{k-1}} \sigma_i}_{\overline{I}^{(k-1)}, U^{(k-1)}}.$$
Observe that $U^{(k-1,n)}=\hat{U}^{(k,k)}$, $\left(\bigwedge_{i\in N_{k-1}} \sigma_i\right)=\left(\sigma_k \wedge \bigwedge_{i\in N_{k}} \sigma_i\right)$ and $\overline{I}^{(k-1)}=\hat{I}^{(k,k)}$.
It holds that  $$Y^{(k,k)}_{\hat{U}^{(k,k)}}\sim \mu^{\sigma_k \wedge \bigwedge_{j\in N_{k}} \sigma_j}_{\hat{I}^{(k,k)}, \hat{U}^{(k,k)}}.$$

Next, consider the general case. By induction hypothesis,  for $j-1$ for $k<j\leq n$,
the clustered conditional Gibbs property is satisfied by $(\vec{Y}^{(k,j-1)}, \{B_k\}\cup \{B_i\}_{i\in N_{j-1}})$  on $I^{(k,j-1)}$ with parameter $(\epsilon_0,\gamma_0,\delta_0)$,
which means
\begin{align}\label{eq:induction-hypothesis-2}
    Y^{(k,j-1)}_{\hat{U}^{(k,j-1)}}\sim \mu^{\sigma_k \wedge \bigwedge_{i\in N_{j-1}} \sigma_i}_{\hat{I}^{(k,j-1)}, \hat{U}^{(k,j-1)}}.
\end{align}
%
%

Depending on whether $B_j=\emptyset$, there are two cases. 
First, if $B_j$ is $\emptyset$, then  $(\vec{Y}^{(k,j)}, \{B_k\} \cup \{B_i\}_{i\in N_{j}})=(\vec{Y}^{(k,j-1)},\{B_k\} \cup \{B_i\}_{i\in N_{j-1}})$ and $I^{(k,j)}=I^{(k,j-1)}$.
By \eqref{eq:induction-hypothesis-2}, it holds that 
\[ Y^{(k,j)}_{\hat{U}^{(k,j)}}\sim \mu^{\sigma_k \wedge \bigwedge_{i\in N_{j}} \sigma_i}_{\hat{I}^{(k,j)}, \hat{U}^{(k,j)}}.\]
%
%
Next, assume  $B_j\neq \emptyset$.
%
%
%
%
%
%
By \Cref{lem:substituting}, for any $\tau\in \Sigma_{\hat{U}^{(k,j-1)}}$,
\[\mu^{\sigma_k\wedge \bigwedge_{i\in N_{j-1}} \sigma_i }_{\hat{I}^{(k,j-1)}, \hat{U}^{(k,j-1)}} (\tau) = \mu^{\sigma_k\wedge  \bigwedge_{i\in N_j} \sigma_i}_{\hat{I}^{(k,j)}, \hat{U}^{(k,j-1)}} (\tau).\]
Note that $Y^{(k,j-1)}_{\hat{U}^{(k,j-1)}}=Y^{(k,j)}_{\hat{U}^{(k,j-1)}}$.
By \eqref{eq:induction-hypothesis-2}, we have $$Y^{(k,j)}_{\hat{U}^{(k,j-1)}} \sim \mu^{\sigma_k\wedge  \bigwedge_{i\in N_j} \sigma_i}_{\hat{I}^{(k,j)}, \hat{U}^{(k,j-1)}}.$$
Recall that each $Y^{(k,j)}_{\beta_j}$ is drawn independent from the distribution $\mu^{Y^{(k-1)}_{\vbl(\kappa_j)\setminus\{\beta_j\}}}_{I^{(k,n)},\beta_j}$.
Recall that we assume the input satisfies the properties asserted by \Cref{lem:clustering-correctness}. According to this, it holds that $(Y^{(0)}, \{B_i\}_{i\in N_0})$ satisfies \Cref{cond:clustering}. Thus, for any distinct $i,j\in N_0$, the distance between $B_i$ and $B_j$ is at least $2(\ell+2)$.
Thus, we have $Y^{(k-1)}_{\vbl(\kappa_j)\setminus\{\beta_j\}}=Y^{(k,j)}_{\vbl(\kappa_j)\setminus\{\beta_j\}}$ and for any $c\in \Sigma_{\beta_j}$,
\[\mu^{Y^{(k-1)}_{\vbl(\kappa_j)\setminus\{\beta_j\}}}_{I^{(k,n)},\beta_j}(c) = \mu^{Y^{(k,j)}_{\vbl(\kappa_j)\setminus\{\beta_j\}}}_{\hat{I}^{(k,j)},\beta_j}(c).\]
It then holds that for any $\tau\in \Sigma_{\hat{U}^{(k,j)}}$,
\[\begin{aligned}
    \Pr\left(Y^{(k,j)}_{\hat{U}^{(k,j)}} = \tau\right) &=  \mu^{\sigma_k\wedge  \bigwedge_{i\in N_j} \sigma_i}_{\hat{I}^{(k,j)}, \hat{U}^{(k,j-1)}}(\tau_{ \hat{U}^{(k,j-1)}}) \cdot \mu^{\tau_{\vbl(\kappa_j)\setminus\{\beta_j\}}}_{\hat{I}^{(k,j)},\beta_j}(\tau_{\beta_j}) \\
    &=\mu^{\sigma_k\wedge  \bigwedge_{i\in N_j} \sigma_i}_{\hat{I}^{(k,j)}, \hat{U}^{(k,j)}}(\tau) .
\end{aligned}\]
This proves that $(\vec{Y}^{(k,j)},\{B_k\}\cup \{B_i\}_{i\in N_j})$ satisfies the clustered conditional Gibbs property on instance $I^{(k,j)}$ with parameter $(\epsilon_0,\delta_0,\gamma_0)$.
\end{proof}

By \Cref{cla:resample-induction-2}, we have  that $(\vec{Y}^{(k,n)}, \{B_k\})$ satisfies the clustered conditional Gibbs property on $I^{(k,n)}$ with parameter $(\epsilon_0,\gamma_0,\delta_0)$.
Thus, $(\vec{Y}^{(k,n)}, B_k)$ satisfies the augmented conditional Gibbs property on $I^{(k,n)}$ with parameter $(\epsilon_0,\gamma_0,\delta_0)$.
Thus, upon the execution of \Cref{alg:resampling} on  $v_k$, the input to \Cref{line:callrecursivesampling}  satisfies \Cref{cond:warm-up-recursive-without-cd}.
By \Cref{lem:recursive-sample-correctness-complexity}, after the execution of \Cref{line:callrecursivesampling} of \Cref{alg:resampling}, $\vec{Y}^{(k,n)}$ follows the distribution  $\mu_{I^{(k,n)}}$.
Thus, it holds that $Y^{(k,n)}_{U^{(k,n)}}$ is identical to $Y^{(k)}_{U^{(k,n)}}$ and follows the distribution $\mu_{I^{(k,n)}, U^{(k,n)}}$.

For $k\leq j\leq n$, define
\[\begin{aligned}
    & \overline{I}^{(k,j)} \triangleq  \left(\{X_i\}_{i\in U^{(k,j)}}\cup \{X_{\beta_i}\}_{i\in N_k\setminus N_j}, \{A_v\}_{v\in V^{(k,j)}}\cup \{A_{\lambda_i}\}_{i\in N_k\setminus N_j} \cup \{A_{\kappa_i}\}_{i\in N_j}\right). 
\end{aligned}\]
By \Cref{lem:substituting},  for any $\pi\in \Sigma_{U^{(k,n)}}$, it holds that 
\[\mu_{I^{(k,n)}, U^{(k,n)}}(\pi)=\mu_{\bar{I}^{(k,n-1)}, U^{(k,n)}}^{\bigwedge_{i\in N_{n-1}}\sigma_i}(\pi)=\mu_{\bar{I}^{(k,n-2)}, U^{(k,n)}}^{\bigwedge_{i\in N_{n-2}}\sigma_i}(\pi)= ... = \mu_{\bar{I}^{(k,k)}, U^{(k,n)}}^{\bigwedge_{i\in N_k}\sigma_i}(\pi).\]
%
Thus, it holds that
\[Y^{(k)}_{U^{(k,n)}} \sim \mu_{\bar{I}^{(k,k)}, U^{(k,n)}}^{\bigwedge_{i\in N_k}\sigma_i}(\pi).\]
Recall that $U^{(k,n)} = \overline{U}^{(k)}$ and $\overline{I}^{(k,k)}=\overline{I}^{(k)}$.
It holds that $Y^{(k)}_{\overline{U}^{(k)}}$ follows the distribution $\mu^{\bigwedge_{j\in N_k} \sigma_j}_{\overline{I}^{(k)}, \overline{U}^{(k)}}$. 
Thus, after executing \Cref{line:resample-updating} of \Cref{alg:resampling}, we have that $(\vec{Y}^{(k)}, \{B_i\}_{i\in N_k})$ satisfies the clustered conditional Gibbs property on instance $I$ with parameter $(\epsilon_0,\gamma_0,\delta_0)$.
This completes the induction.

Therefore, $(\Vec{Y}^{(n)}, \emptyset)$ satisfies the clustered conditional Gibbs property on instance $I$ with parameter $(\epsilon_0,\gamma_0,\delta_0)$.
As argued before, this guarantees that After \Cref{alg:resampling} is sequentially executed on all node $v\in \RV{R}$,  the final output assignment $\Vec{Y}=\vec{Y}^{(n)}$  follows the distribution $\mu_I$.
\end{proof}


\begin{proof}[Proof of \Cref{lem:item-resample-complexity} of \Cref{lem:resample-correctness-complexity}]
%
Recall the choice of parameter $(\epsilon_0,\gamma_0,\delta_0)$ in \eqref{eq:def-parameter-without-cd-expected-complexity}. 
Fix any $0<\eta<1$, the follows hold.

\begin{itemize}
    \item 
    According to \Cref{lem:item:recursive-sample-complexity} of \Cref{lem:recursive-sample-correctness-complexity}, 
    with probability $1-\frac{\eta}{2}$, 
    at every node $v\in \RV{R}$, the call to the \RecursiveSampling{} procedure in \Cref{alg:resampling} returns
    within radius  %
    \[r_1=r_v+ \ell_0(\epsilon_0,\gamma_0,\delta_0)+\tO\left(\log^4 n \cdot \left(\log\frac{1}{\gamma}\cdot \log^4 \frac{1}{\eta} + \log^2\frac{1}{\gamma} \cdot \log^2 \frac{1}{\eta}\right) \right)\]
    from the center $p_v$ in the LLL instance $I'$.
    And $I'(B_{r_1}(p_v))$ can be locally constructed within $I({B_{r_2}(v)})$, where \[r_2=r_1+O(\mathcal{D}+|\RV{R}|\cdot \ell_0(\epsilon_0,\gamma_0,\delta_0)).\]
    
    \item According to \Cref{lem:item:clustering-complexity} of \Cref{lem:clustering-correctness}, with probability $1-\frac{\eta}{2}$, we have $\mathcal{D}=\Tilde{O}(|\RV{R}| \cdot \log^2 n \log^2 \frac{1}{\gamma}\log\frac{1}{\eta})$,
    which implies $r_v=\Tilde{O}(|\RV{R}| \cdot \log^2 n \log^2 \frac{1}{\gamma}\log\frac{1}{\eta})$.
\end{itemize}
Overall, with probability at least $1-\eta$, \Cref{alg:resampling} returns at every node $v\in \RV{R}$ with radius bounded by
\[\tO\left(|\RV{R}| \cdot \log^2 n\cdot \log^2 \frac{1}{\gamma}\cdot\log\frac{1}{\eta}\right)+ \tO\left(\log^4 n \cdot \log^2\frac{1}{\gamma} \cdot\log^4 \frac{1}{\eta}\right). \qedhere
\]
\end{proof}


\bibliographystyle{alpha}
\bibliography{refs}

\newpage
\begin{appendix}

\section{Simulation of \SLV{} in \LOCAL{} Model}\label{sec: details of SLV-local to local}

We prove~\Cref{thm: SLV-local to local}. The \LOCAL{} algorithm for simulating the \SLV{} algorithm is described in \Cref{alg:SLV-local-to-local}. Recall that whether  a  node $v\in V$ is an active node (which means that $v\in A$) is indicated in the local memory $M_v$. 
For each active node $v\in A$, \Cref{alg:SLV-local-to-local} is executed on $v$, with the initial $M_v$ presented to the \LOCAL{} algorithm as $v$'s input. 
Also recall that we define $G^t$ to be the $t$-th power graph with an edge between any pair of $u,v$ where $\dist_G(u,v)\le t$. 

\begin{algorithm}
	\caption{The \LOCAL{} algorithm for simulating \SLV{} at active node $v\in A$}\label{alg:SLV-local-to-local}
	initialize $A_v\leftarrow\emptyset$\;
	\For{$i=1$ to $\lceil\log n\rceil$}{
		in the induced subgraph $G^{(2^{i+1})}[A]$, find the maximal connected component containing $v$, denoted by $A'$\;
		
		simulate the \SLV{} algorithm with set $A'$ of active nodes, where the simulation \emph{fails} if the maximum radius accessed by the \SLV{} algorithm exceeds $2^i$\;
		
		\If{the simulation does not fail}
		{
			let $A_v\leftarrow A'$ and update the output states according to the simulation of $v$\;
   \textbf{break}\;
		}
	}
	
	\If{there exists $u\in A$ with $A_u\subseteq A_v$}{cancel all the updates to output made by $u$\;\label{line5} 
 \tcp{In this case, we say that $u$ is {canceled} by $v$.}}
	
	\If{there exists $u\in A$ with $A_v\subseteq A_u$}{cancel all the updates to output made by $v$\;\label{line6} 
 \tcp{In this case, we say that $v$ is {canceled} by $u$.}}
\end{algorithm}


It is easy to see that~\Cref{alg:SLV-local-to-local} terminates, and when it terminates, a set $A_v\subseteq A$ is constructed for each $v\in A$.
The following lemma follows from the property of $G^{(2^{i+1})}$.

\begin{lemma}\label{lem: occupiedproperty}
	For distinct $u,v\in A$, one of the followings must be true:
	\begin{enumerate}
		\item $A_u\cap A_v=\emptyset$.
		\item $A_u\subseteq A_v$.
		\item $A_v\subseteq A_u$.
	\end{enumerate} 
\end{lemma}

\begin{proof}
	Assume $A_u$ is a maximum connected component of $G^{t_u}$ and $A_v$ is a maximum connected component of $G^{t_v}$, for some integers $t_u$ and $t_v$. Suppose $t_u\le t_v$. If $A_u\cap A_v\not=\emptyset$, then in $G^{t_v}$, $A_u$ and $A_v$ must be in the same connected components, since every edge in $G^{t_u}$ is also an edge in $G^{t_v}$. Thus, $A_u\subseteq A_v$. Similarly, if $t_u\ge t_v$, then $A_v\subseteq A_u$. 
\end{proof}

The following lemma is because of the condition of canceling a vertex.

\begin{lemma}\label{lem: contain all}
	For any $v\in A$, there exists $u\in A$ such that $v\in A_u$ and $u$ is not canceled. 
\end{lemma}
\begin{proof}
	Consider the last time $v$ appears in some $A_u$. If $v$ is not in any $A_u$ in the end, $u$ must be canceled by some other vertex $u'$, in which case we have $A_u\subseteq A_{u'}$, and hence $v\in A_{u'}$, a contradiction.
\end{proof}

For any $v\in A$, note that the simulation must succeed if $i=\lceil\log n\rceil$. 
Thus, there exists an $i$ for $v$ such that the simulation in the $i$-th iteration does not fail. 
Denote such $i$ as $I_v$. 
Suppose that this successful simulation has the radius $\ell'_{v,j}$ for $v\in A'$ in the $j$-th scan. 
Define $T_v=\cup_{v\in A',j\in [N]}B_{\ell'_{v,j}}(v)$, i.e., 
the set of all vertices whose outputs are supposed to be updated by $v$. 
The following lemmas guarantee the correctness of the simulation.

\begin{lemma}\label{lem: no intersection of simulation}
	If $u,v\in A$ and both $u,v$ are not canceled, then $T_u\cap T_v=\emptyset$.
\end{lemma}
\begin{proof}
    According to~\Cref{lem: occupiedproperty}, we have $A_u\cap A_v=\emptyset$. If $T_u\cap T_v\not=\emptyset$, there must exist $u'\in A_u,v'\in A_v$ such that $B_x(u')\cap B_y(v')\not=\emptyset$, where $x,y$ are the radii accessed by $u'$ and $v'$ and satisfy $x\le 2^{I_u},y\le 2^{I_v}$. w.o.l.g, we assume $I_u\le I_v$. Then we have $\dist_G(u',v')\le x+y\le 2^{I_v+1}$. That is to say, $u'$ and $v'$ has an edge in $G^{(2^{I_v+1})}$, which causes a contradiction since $A_v$ is a maximum connected components.
\end{proof}
\begin{lemma}[correctness of simulation]\label{lem: correctness of simulation}
    For any $u\in A$, if $u$ is not canceled, then \Cref{alg:SLV-local-to-local} at $u$ produces the same outputs for the nodes in $A_u$ as in the \SLV{} algorithm.
\end{lemma}
\begin{proof}
    We will prove by induction on $k$ that, the first $k$ vertices accessed by the \SLV{} algorithm is simulated exactly the same in our \LOCAL{} algorithm by some $u\in A$. For $k=1$, According to~\Cref{lem: contain all}, there exists $A_u$ containing the first vertex accessed in the \SLV{} algorithm, which is simulated correctly by $u$. Assume that the first $k$ vertices are simulated correctly. For the $k+1$-th vertex accessed, suppose it is vertex $v$ in the $j$-th phase. By~\Cref{lem: contain all}, there exists $A_u$ containing the $k+1$-th vertex accessed in the \SLV{} algorithm, and it will only access $B_{\ell_{v,j}}(v)\in T_u$, which only contains the information updated by the vertices in $A_u$ according to~\Cref{lem: no intersection of simulation}. Thus, the simulation is correct. 
\end{proof}
The following lemma bounds the round complexity of \Cref{alg:SLV-local-to-local}. Recall that $\ell_{v,j}$ is the radius of the ball accessed by the \SLV{} algorithm at node $v\in A$ in the $j$-th phase. 

\begin{lemma}\label{lem: simulation complexity}
	For any $v\in A$, we have $I_v\le\lceil\log \max_{v\in A,j\in[N]}\ell_{v,j}\rceil$.
\end{lemma}
\begin{proof}
	Consider the $\lceil\log \max_{v\in A,j\in[N]}\ell_{v,j}\rceil$-th iteration for vertex $v$. We claim that for any $u\in A\backslash A_v$, and $k\in[N]$, we have $B_{\ell_{u,k}}\cap \left(\cup_{v\in A',j\in [N]}B_{\ell_{v,j}}(v)\right)=\emptyset$. Otherwise, $u$ has distance at most $2\cdot \max_{v\in A,j\in[N]}\ell_{v,j}$ to a vertex in $A_v$, contradicting the fact that $A_v$ is a maximal connected components of $G^{x}$ for $x\ge 2\cdot \max_{v\in A,j\in[N]}\ell_{v,j}$. By the same argument as in the proof of~\Cref{lem: correctness of simulation}, the simulation on $A_v$ will produce the same output as in the \SLV{} algorithm, in which case the simulation must not fail. 
\end{proof}

Now we are ready to prove~\Cref{thm: SLV-local to local}.

\begin{proof}[Proof of~\Cref{thm: SLV-local to local}]
	The correctness of~\Cref{alg:SLV-local-to-local} follows from~\Cref{lem: correctness of simulation,lem: contain all}.
	
	Then we bound the complexity of~\Cref{alg:SLV-local-to-local}. Let $M=\max_{v\in A,j\in[N]}\ell_{v,j}$. For a vertex $v\in A$, to run the $i$-th loop, $v$ need to find the maximal connected component in $G^{2^{i}+1}[A]$, and grow a ball of radius at most $2^i$ from each vertex in the connected component.  According to~\Cref{lem: simulation complexity}, we have $2^{i}+1=O(M)$ since $i<I_v$. The maximal connected components in $G^{O(M)}[A]$ has diameter at most $O(|A|\cdot M)$ in $G$. Thus, $u$ only needs to collect the information in $B_\ell(u)$ for some $\ell=O(|A|\cdot M)$. The only problem is that $u$ do not know $\ell$ initially. This can be solved by guessing $\ell$ by scanning $\ell'=2^i$ for $i=1,2,\ldots$, collecting information in $B_{\ell'}(u)$ and determining whether stops. 
\end{proof}
\end{appendix}

\end{document}